\documentclass[acmsmall]{acmart}

\acmJournal{PACMPL}
\acmVolume{1}
\acmNumber{CONF} 
\acmArticle{1}
\acmYear{2018}
\acmMonth{1}
\acmDOI{} 
\startPage{1}

\setcopyright{none}

\usepackage{booktabs}   
\usepackage{subcaption} 

\usepackage[normalem]{ulem}

\usepackage{ifthen}

\newboolean{talk}
\setboolean{talk}{false}
\newboolean{paper}
\setboolean{paper}{false}

\newboolean{IEEEstyle}
\setboolean{IEEEstyle}{false}
\newboolean{lipicsstyle}
\setboolean{lipicsstyle}{false}
\newboolean{eptcsstyle}
\setboolean{eptcsstyle}{false}
\newboolean{sigplanstyle}
\setboolean{sigplanstyle}{false}
\newboolean{easychairstyle}
\setboolean{sigplanstyle}{false}
\newboolean{PACMPL}
\setboolean{PACMPL}{false}

\newboolean{needstheorems}
\setboolean{needstheorems}{false}
\newboolean{withimages}
\setboolean{withimages}{false}
\newboolean{withproofs}
\setboolean{withproofs}{true}

\newboolean{french}
\setboolean{french}{false}

\setboolean{paper}{true}
\setboolean{needstheorems}{true}
\setboolean{PACMPL}{true}
\usepackage{amssymb}
\usepackage{amsmath}
\usepackage{stmaryrd}
\usepackage{prooftree}
\usepackage{multicol}
\usepackage{amsthm}
\usepackage{mathtools}
\usepackage{yfonts}
\usepackage{appendix}
\usepackage{compactspacing}
\usepackage{xcolor}
\usepackage{fancybox}

\newcommand\dom[1]{{\tt dom}({#1})}

\renewcommand\l{\lambda}

\newcommand\recdef{ \Coloneqq }

\newcommand\sysnormalpr[1]{{\tt normal}_{#1}}
\newcommand\sysneutralpr[1]{{\tt neutral}_{#1}}
\newcommand\sysabspr[1]{{\tt abs}_{#1}}
\newcommand\sysnotabspr[1]{\neg \sysabspr{#1}}
\newcommand\sysnormal[2]{\sysnormalpr{#1}(#2)}
\newcommand\sysneutral[2]{\sysneutralpr{#1}(#2)}
\newcommand\sysabs[2]{\sysabspr{#1}(#2)}
\newcommand\sysnotabs[2]{\sysnotabspr{#1}(#2)}

\newcommand\spnormalpr{\sysnormalpr\spi}
\newcommand\spneutralpr{\sysneutralpr\spi}

\newcommand\spnormal[1]{\sysnormal\spi{#1}}
\newcommand\spneutral[1]{\sysneutral\spi{#1}}

\newcommand\sknormalpr{\sysnormalpr\ske}
\newcommand\skneutralpr{\sysneutralpr\ske}
\newcommand\skabspr{\sysabspr\ske}

\newcommand\sknormal[1]{\sysnormal\ske{#1}}
\newcommand\skneutral[1]{\sysneutral\ske{#1}}

\newcommand\lhneutralvoid{{\tt neutral}_{\lsp}}
\newcommand\lhnormalvoid{{\tt normal}_{\lsp}}
\newcommand\lhnormal[1]{\lhnormalvoid(#1)}
\newcommand\lhnormalclosepr{{\tt normal}_{\lsp}^{\#}}
\newcommand\lhnormalclose[1]{\lhnormalclosepr(#1)}
\newcommand\lhneutral[1]{{\tt neutral}_{\lsp}(#1)}

\newcommand\hdneutral[1]{\sysneutral\hd{#1}}
\newcommand\hdnormal[1]{\sysnormal\hd{#1}}

\newcommand\loneutral[1]{\sysneutral\lo{#1}}
\newcommand\lonormal[1]{\sysnormal\lo{#1}}
\newcommand\lhnormalpr[1]{{\tt normal}_{\lsp}^{#1}}
\newcommand\lhnormalp[2]{\lhnormalpr{#2}(#1)}

\newcommand\lhneutralpr[1]{{\tt neutral}_{\lsp}^{#1}}
\newcommand\lhneutralp[2]{\lhneutralpr{#2}(#1)}

\newcommand\neutral{\mathtt{neutral}}

\newcommand\neutype{\neutral}

\newcommand\abstype{\mathtt{abs}}

\newcommand\col{ : }

\newcommand\Deribbase[5]{{#3}\ {\pmb\vdash}_{#2}^{#1} {#4}\  {:}\  {#5}}

\makeatletter

\newcommand{\Deribase}[1]{%
  \def\DeribW[##1]{\Deribbase{##1}{#1}}%
  \def\DeribWO{\Deribbase{}{#1}}%
  \@ifnextchar[\DeribW\DeribWO%
  }

  \newcommand{\Deri}{%
  \def\DeriW_##1{\Deribase{##1}}%
  \def\DeriWO{\Deribase{}}%
  \@ifnextchar_\DeriW\DeriWO%
  }
\makeatother

\newcommand\single[1]{\mult{#1}}

\newcommand{\tightpred}[1]{{\tt tight}({#1})}

\newcommand\MSigma[2]{[#1]_{{#2}}}

\newcommand{\steps}{b}
\newcommand{\esteps}{e}
\newcommand{\estepstwo}{\esteps'}
\newcommand{\estepsthree}{\esteps''}
\newcommand{\Steps}{B}
\newcommand{\ESteps}{E}
\newcommand{\stepstwo}{\steps'}
\newcommand{\stepsthree}{\steps''}

\newcommand{\result}{r}
\newcommand{\Result}{R}
\newcommand{\resulttwo}{\result'}
\newcommand{\resultthree}{\result''}

\newcommand{\garbage}{g}

\newcommand{\mtype}{{\tt M}}
\newcommand{\mtypetwo}{{\tt N}}

\newcommand{\M}{\mtype}
\newcommand{\N}{\mtypetwo}

\newcommand{\type}{\tau}
\newcommand{\typetwo}{\sig}
\newcommand{\typethree}{\rho}
\newcommand{\typefour}{\type'}

\newcommand{\atype}{T}
\newcommand{\atypetwo}{U}
\newcommand{\atypethree}{V}

\newcommand{\mneutral}{\mathtt{Neutral}}
\newcommand{\mtight}{\mathtt{Tight}}

\newcommand{\iI}{{i \in I}}
\newcommand{\jJ}{{j \in J}}

\makeatletter
\newcommand{\exder}{%
  \def\exderW[##1]{\triangleright_{##1}\ }%
  \def\exderWO{\triangleright\ }%
  \@ifnextchar[\exderW\exderWO%
  }
\makeatother

\newcommand{\Gam}{\Gamma}
\newcommand{\Del}{\Delta}

\newcommand{\tight}{{\tt tight}}
\newcommand\nf{\tight}

 \newcommand\TDeri[5][]{{#2} \tri {#3} \vdash^{#1} {#4} : {#5}}

\renewcommand\single[1]{[#1]}

\newcommand{\subslemma}{{\tt subs}}
\newcommand{\partialsubslemma}{{\tt partial\mbox{-}subs}}
\newcommand{\fun}{{\tt fun}}
\newcommand{\funsteps}{\fun_\steps}
\newcommand{\funresult}{\fun_\result}
\newcommand{\app}{{\tt app}}
\newcommand{\appsteps}{\app_\steps}
\makeatletter
\newcommand{\appresult}{%
  \def\appresultW<##1>{\app_\result^{##1}}%
  \def\appresultWO{\app_\result}%
  \@ifnextchar<\appresultW\appresultWO%
  }
\makeatother
\newcommand{\esrule}{{\tt ES}}

\newcommand{\sig}{\sigma}
\newcommand{\tarrow}[2]{#1 \rightarrow #2}
\newcommand{\ty}[2]{\tarrow{#1}{#2}}
\newcommand{\mult}[1]{[ #1 ] }
\newcommand{\tyjn}[4]{{#3} \vdash^{#1} #2{:}#4}

\newcommand{\tyjnpre}[2]{\vdash_{\nf}^{#1} #2}

\newcommand{\emm}{ [\,] }
\newcommand{\cset}[1]{ \{ #1 \} }

\newcommand{\es}{\emptyset}
\newcommand{\tingD}[2]{{#1} \tri z {#2}}

\newcommand{\tri}{\triangleright}

\newcommand{\tderiv}{\Phi}
\newcommand{\tderivtwo}{\tderiv'}
\newcommand{\tderivthree}{\tderiv''}
\newcommand{\tderivfour}{\tderiv'''}

\newcommand{\emptyctx}{\epsilon}
\newcommand{\typctx}{\Gamma}
\newcommand{\typctxtwo}{\Delta}

\newcommand{\typctxthree}{\Pi}

\newcommand{\Spine}{\Result}
\newcommand{\spine}{\result}
\newcommand{\spinetwo}{\spine'}
\newcommand{\spinethree}{\spine''}

\newcommand{\intprecise}{garbage-tight\xspace}

\newcommand{\precise}{tight\xspace}

\newcommand{\maxprecise}{mx-tight\xspace}
\newcommand{\Maxprecise}{Mx-tight\xspace}

\newcommand{\Precise}{Tight\xspace}

\newcommand{\preciseness}{tightness\xspace}

\newcommand{\dynamic}{traditional\xspace}
\newcommand{\Dynamic}{Traditional\xspace}

\newcommand{\communicative}{\dynamic}

\newcommand{\atomtype}{X}

\newcommand{\iz}{i_1}

\newcommand{\typeocc}[2]{{{\tt Occ}_{#1}(#2)}}
\newcommand{\possubtype}[1]{\typeocc{+}{#1}}
\newcommand{\negsubtype}[1]{\typeocc{-}{#1}}

\newcommand{\emptymset}{\emm}

\newcommand{\putinctx}[2]{#1 \langle #2 \rangle }

\newcommand{\polcomp}[2]{\delta(#1,#2)}

\newcommand{\losize}[1]{| #1|_{\lo} }
\newcommand{\spsize}[1]{| #1|_{\spi}}
\newcommand{\hdsize}[1]{| #1|_{\hd}}
\newcommand{\sksize}[1]{| #1|_{\ske} }
\newcommand{\mxsize}[1]{| #1|_{\systemax}}
\newcommand{\lhsize}[1]{| #1|_{\lhd} }
\newcommand{\lspsize}[1]{| #1|_{\lsp} }

\newcommand{\tysize}[1]{| #1|_{ty} }
\newcommand{\syssize}[2]{| #2|_{#1}}

\newcommand{\lhvar}{{\tt lhvar}}
\newcommand{\lhnn}{{\tt lhnn}}
\newcommand{\lhnlam}{{\tt lhnlam}}
\newcommand{\lhhclose}{{\tt lhhclose}}

\newcommand{\lhnapp}{{\tt lhnapp}}
\newcommand{\lhnesub}{{\tt lhnesub}}
\newcommand{\lhnosub}{{\tt lhnosub}}
\newcommand{\lhmsub}{{\tt lhmsub}}
\newcommand{\lhlamm}{{\tt lhlamm }}
\newcommand{\lhnpn}{{\tt lhnpn }}
\newcommand{\lhnmn}{{\tt lhnmn }}

\newcommand{\spi}{hd}

\newcommand{\ske}{{lo}}
\newcommand{\toske}{\Rew{\ske}}

\newcommand{\hd}{{hd}}
\newcommand{\tohd}{\Rew{\hd}}
\newcommand{\lod}{{lo}}
\newcommand{\tolod}{\Rew{\lod}}

\newcommand\system{\stratsym}

\newcommand\systemax{mx}

\newcommand\systemlsp{\lh}

\newcommand{\abslh}[1]{\sysabspr{\lh}(#1)}
\newcommand{\abslhvoid}{\sysabspr{\lh}}

\newcommand{\stratsym}{S}
\newcommand{\tostrat}{\Rew{\stratsym}}
\newcommand{\sm}{\mathbin{{\setminus}\mspace{-5mu}{\setminus}}}

\newcommand{\Terms}{\mathcal T}

\newcommand{\id}{{\tt I}}

\ifthenelse{\boolean{talk}}{}{\usepackage[utf8]{inputenc}}
\ifthenelse{\boolean{french}}{\usepackage[french]{babel}}{
  \ifthenelse{\boolean{lipicsstyle}}{}{\usepackage[english]{babel}}}
\usepackage{amssymb}
\usepackage{amsmath}
\usepackage{graphicx}
\usepackage{cmll}
\usepackage{stmaryrd}
\usepackage{multirow}
\usepackage{fancybox}
\usepackage{xspace}

\ifthenelse{\boolean{IEEEstyle}}{
\setboolean{needstheorems}{true}}{}

\ifthenelse{\boolean{eptcsstyle}}{
\setboolean{needstheorems}{true}}{}

\ifthenelse{\boolean{sigplanstyle}}{
\setboolean{needstheorems}{true}}{}

\ifthenelse{\boolean{needstheorems}}{

\usepackage{amsthm}

    \newtheorem{theorem}{Theorem}[section]
    \newtheorem{lemma}[theorem]{Lemma}

    \newtheorem{proposition}[theorem]{Proposition}
    \newtheorem{definition}[theorem]{Definition}

}{}

\newcommand{\myproof}[1]{
\ifthenelse{\boolean{withproofs}}{#1}{}
}

\newcommand{\lterms}{\Lambda}
\newcommand{\ltermslsc}{\Lambda_{{\tt lsc}}}
\newcommand{\la}[1]{\lambda #1.}

\newcommand{\tm}{t}
\newcommand{\tmtwo}{p}
\newcommand{\tmthree}{u}
\newcommand{\tmfour}{q}
\newcommand{\tmfive}{m}

\newcommand{\tmb}{u}  
\newcommand{\tmbb}{v}  %

\newcommand{\var}{x}
\newcommand{\vartwo}{y}
\newcommand{\varthree}{z}
\newcommand{\varfour}{w}

\makeatletter

\newcommand{\Rewbase}{%
  \def\RewbaseW[##1]##2##3{\ {\xrightarrow{##1}}{}_{##2}^{##3}\ }%
  \def\RewbaseWO##1##2{\ {\xrightarrow{}}{}_{##1}^{##2}\ }%
  \@ifnextchar[\RewbaseW\RewbaseWO%
  }

\newcommand{\Rew}[1]{%
  \def\RewW[##1]{\Rewbase[##1]{#1}{}}%
  \def\RewWO{\Rewbase{#1}{}}%
  \@ifnextchar[\RewW\RewWO%
  }

\newcommand{\Rewn}[2][*]{%
  \def\RewnW[##1]{\Rewbase[##1]{#2}{#1}}%
  \def\RewnWO{\Rewbase{#2}{#1}}%
  \@ifnextchar[\RewnW\RewnWO%
  }

\makeatother

\renewcommand{\to}{\Rew{}}

\newcommand{\Bsym}{{\tt B}}

\newcommand{\esym}{{\mathtt e}}

\newcommand{\msym}{{\mathtt m}}

\newcommand{\ctxholep}[1]{\langle #1\rangle}
\newcommand{\ctxhole}{\ctxholep{\cdot}}

\newcommand{\nbvctxtwo}[1]{\nbvctxtwo{#1}}

\newcommand{\defeq}{:=}
\newcommand{\grameq}{::=}

\newcommand{\isub}[2]{\{#1/#2\}}
\newcommand{\esub}[2]{[#1/#2]}
\renewcommand{\esub}[2]{[#1 \backslash #2]}
\renewcommand{\isub}[2]{\{#1{\shortleftarrow}#2\}}
\renewcommand{\L}{{\tt L}}

\newcommand{\letexp}[3]{{\tt let}\ #1=#2\ {\tt in}\ #3}

\renewcommand{\lhd}{{lhd}}
\newcommand{\lh}{\lhd}
\newcommand{\lsp}{\lhd}
\newcommand{\lspb}{(\lsp_\msym)}

\newcommand{\lspe}{(\lsp_{\esym})}
\newcommand{\lspapp}{(\lsp_{@})}
\newcommand{\lspabs}{(\lsp_{\l})}
\newcommand{\lspsub}{(\lsp_{s})}
\newcommand{\lhb}{(\lh_\msym)}
\newcommand{\lhs}{(\lh_{\esym})}
\newcommand{\lhapp}{(\lh_{@})}
\newcommand{\lhabs}{(\lh_{\l})}
\newcommand{\lhsub}{(\lh_{s})}

\newcommand{\tolh}{\Rew{\lh}}
\newcommand{\tolsp}{\rightarrow_{\lsp}}
\newcommand{\lhsp}{\lsp}
\newcommand{\tolhb}{\Rew{\msym}}
\newcommand{\tolhs}{\Rew{\esym}}
\newcommand{\tolspe}{\Rew{\esym}}
\newcommand{\tolspb}{\Rew{\msym}}
\newcommand{\lhc}{{\tt H}}

\newcommand{\cwc}[1]{\langle  \! \langle    #1 \rangle  \! \rangle}

\newcommand{\llbrace}{\{ \kern -0.27em \vert}
\newcommand{\rrbrace}{\vert \kern -0.27em \}}

\renewcommand{\l}{\lambda}

\newcommand{\ie}{{\em i.e.}\xspace}
\newcommand{\eg}{{\em e.g.}\xspace}
\newcommand{\ih}{{\emph{i.h.}}\xspace}
\newcommand{\fv}[1]{{\tt fv}(#1)}

\newcommand{\ignore}[1]{}
\newcommand{\msep}{\hspace*{0.2cm}}
\newcommand{\sep}{\hspace*{0.5cm}}

\newcommand{\colspace}{@{\hspace{.5cm}}}
\newcommand{\myinput}[1]{\ifthenelse{\boolean{withimages}}{\input{#1}}{}}

\newcommand{\refappendix}[1]{Appendix~\ref{app:#1}}
\newcommand{\reflemma}[1]{Lemma~\ref{l:#1}}

\newcommand{\refprop}[1]{Proposition~\ref{prop:#1}}
\newcommand{\refpropp}[2]{Proposition~\ref{prop:#1}.\ref{p:#1-#2}}
\newcommand{\refsect}[1]{Sect.~\ref{sect:#1}}

\ifthenelse{\boolean{easychairstyle}}{}{
}
\newcommand{\reffig}[1]{Fig.~\ref{fig:#1}}

\newcommand{\refdef}[1]{Definition~\ref{def:#1}}

\newcommand{\refsection}[1]{Section~\ref{s:#1}}

\newcommand{\axresult}{\ax_\result}
\newcommand{\axres}{\axresult}
\newcommand{\ax}{\mathsf{ax}}
\newcommand{\many}{\mathsf{many}}
\newcommand{\manystrict}{\mathsf{many}_{{>}0}}
\newcommand{\none}{\mathsf{none}}

\newcommand{\set}[1]{\{#1\}}

\newcommand{\nat}{\mathbb{N}}

\newcommand{\tom}{\Rew{\msym}}
\newcommand{\toe}{\Rew{\esym}}

\newcommand{\size}[1]{|#1|}

\newcommand{\unfsym}{\rotatebox[origin=c]{-90}{$\rightarrow$}}
\newcommand{\unf}[1]{#1\unfsym}

\newcommand{\lo}{{LO}\xspace}

\renewcommand{\lo}{{lo}\xspace}

\newcommand{\tolo}{\Rew{\lo}}

\newcommand{\perpe}{mx}

\newcommand{\syssizeSN}[1]{\syssize{\perpe}{#1}}

\newcommand{\toperpNE}{\Rew{\systemax}[0]}

\newcommand\mplus{\uplus}
\newcommand\indset{I}

\newcommand{\lintransfpr}{\mathcal L}
\newcommand{\lintransf}[1]{\lintransfpr(#1)}
\newcommand{\nonlintransfpr}{\mathcal{N}}
\newcommand{\nonlintransf}[1]{\nonlintransfpr(#1)}

\newcommand{\sem}[1]{[\![#1]\!]}

\usepackage{nameref}
\usepackage{appendix}
\capturecounter{theorem}
\capturecounter{figure}

\newcommand\longshort[2]{#1}

\begin{document}

\title[Tight Typings and Split Bounds]{Tight Typings and Split Bounds}


\author{Beniamino Accattoli}
\affiliation{
  \position{}
  \department{LIX}              
  \institution{Inria \& \'Ecole Polytechnique}            
  \streetaddress{}
  \city{}
  \state{}
  \postcode{}
  \country{France}                    
}
\email{beniamino.accattoli@inria.fr}          

\author{St\'ephane Graham-Lengrand}
\affiliation{
  \position{}
  \department{LIX}              
  \institution{CNRS, Inria \& \'Ecole Polytechnique}            
  \streetaddress{}
  \city{}
  \state{}
  \postcode{}
  \country{France}                    
}
\email{graham-lengrand@lix.polytechnique.fr}          

\author{Delia Kesner}
\affiliation{
  \position{}
  \department{IRIF}              
  \institution{CNRS and Universit\'e Paris-Diderot}            
  \streetaddress{}
  \city{}
  \state{}
  \postcode{}
  \country{France}                    
}
\email{kesner@irif.fr}          

\begin{abstract}

Multi types---aka non-idempotent intersection types---have been used
to obtain quantitative bounds on higher-order programs, as pioneered
by de Carvalho. Notably, they bound at the same time the number of
evaluation steps \emph{and} the size of the result.  Recent results
show that the number of steps can be taken as a reasonable time
complexity measure. At the same time, however, these results 
suggest that
multi types provide quite lax complexity bounds, because the size of
the result can be exponentially bigger than the number of steps.

Starting from this observation, we refine and generalise a technique introduced
by Bernadet \& Graham-Lengrand to provide {\it exact} bounds for the maximal strategy.
Our typing judgements carry two counters, one measuring evaluation lengths and
the other measuring result sizes.  In order to emphasise the modularity of the
approach, we provide exact bounds for four evaluation strategies, both in the
$\l$-calculus (head, leftmost-outermost, and maximal evaluation) and in the
linear substitution calculus (linear head evaluation).

Our work aims at both capturing the results in the literature and
extending them with new outcomes. Concerning the literature, it
unifies de Carvalho and Bernadet \& Graham-Lengrand via a uniform technique
and a complexity-based perspective. The two main novelties are exact
split bounds for the leftmost strategy---the only known strategy that
evaluates terms to full normal forms and provides a reasonable
complexity measure---and the observation that the computing device
hidden behind multi types is the notion of substitution at a distance,
as implemented by the linear substitution calculus.

\end{abstract}



\maketitle

\section{Introduction
}

Type systems enforce properties of programs, such as termination,
deadlock-freedom, or productivity. This paper studies a class of type
systems for the $\l$-calculus that refines termination 
by providing
exact bounds for evaluation lengths and normal forms.

\paragraph{Intersection types and multi types.} One of the
cornerstones of the theory of $\l$-calculus is that intersection types
\emph{characterise} termination: not only typed programs terminate,
but all terminating programs are typable as well
\cite{CDC78,CDC80,Pottinger80,krivine1993lambda}. In fact, the $\l$-calculus comes
with different notions of evaluation (\eg call-by-name, call-by-value,
call-by-need, etc) to different notions of normal forms
(head/weak/full, etc) and, accordingly, with different systems of
intersection types.

Intersection types are a flexible tool and, even when one fixes a
particular notion of evaluation and normal form, the type system can
be formulated in various ways. A flavour that became quite convenient
in the last 10 years is that of \emph{non-idempotent} intersection
types \cite{DBLP:conf/tacs/Gardner94,DBLP:journals/logcom/Kfoury00,DBLP:conf/icfp/NeergaardM04,Carvalho07} (a survey can be
found in~\cite{BKV17}), where the
intersection $A \cap A$ is not equivalent to $A$. Non-idempotent
intersection types are more informative than idempotent ones because
they give rise to a \emph{quantitative} approach, that allows counting
resource consumption. 

Non-idempotent intersections can be seen as multi-sets, which is why, to ease the
language, we prefer to call them \emph{multi types} rather than
\emph{non-idempotent intersection types}.
Multi types have two main features:
\begin{enumerate}
  \item \emph{Bounds on evaluation lengths}: they go beyond simply
    qualitative characterisations of termination, as typing derivations provide 
    quantitative bounds
    on the length of evaluation (\ie\ on the number of
    $\beta$-steps). Therefore, they give intensional insights on
    programs, and seem to provide a tool to reason about the complexity 
    of  programs.
  
  \item \emph{Linear logic interpretation}: multi types are deeply
    linked to linear logic. The relational model~\cite{Girard88,DBLP:journals/apal/BucciarelliE01} of linear logic
    (often considered as a sort of canonical model of linear logic) is
    based on multi-sets, and multi types can be seen as a syntactic
    presentation of the relational model of the $\l$-calculus induced
    by the interpretation into linear logic. 
\end{enumerate}

These two facts together have a potential, fascinating consequence:
they suggest that denotational semantics may provide abstract tools
for complexity analyses, that are theoretically solid,
being grounded on linear logic.

Various works in the literature explore the bounding power of multi
types. Often, the bounding power is used \emph{qualitatively}, \ie without
explicitely counting the number of steps, to characterise termination and / or
the properties of the induced relational model.  Indeed, multi types provide
combinatorial proofs of termination that are simpler than those developed for
(idempotent) intersection types (\eg reducibility candidates).  Several papers
explore this approach under the call-by-name
\cite{DBLP:journals/apal/BucciarelliEM12,DBLP:conf/ictac/KesnerV15,DBLP:conf/rta/KesnerV17,DBLP:journals/mscs/PaoliniPR17,DBLP:conf/lics/Ong17}
or the call-by-value
\cite{DBLP:conf/csl/Ehrhard12,DBLP:conf/lfcs/Diaz-CaroMP13,DBLP:conf/fossacs/CarraroG14}
operational semantics, or both \cite{DBLP:conf/ppdp/EhrhardG16}. Sometimes,
precise \emph{quantitative} bounds are provided instead, as
in~\cite{Carvalho07,Bernadet-Lengrand2013}. Multi types can also be used to
provide characterisation of complexity classes
\cite{DBLP:journals/iandc/BenedettiR16}. Other qualitative
\cite{DBLP:conf/csl/Carvalho16,DBLP:conf/rta/GuerrieriPF16} and quantitative
\cite{DBLP:journals/tcs/CarvalhoPF11,DBLP:journals/iandc/CarvalhoF16} studies
are also sometimes done in the more general context of linear logic, rather than
in the $\l$-calculus.

\paragraph{Reasonable cost models.}  
 Usually, the quantitative works 
define  a measure for  typing derivations and show that the measure
provides a bound on the length of evaluation sequences for typed terms.
A criticism that could be raised against these results is, or rather
was, that the number of $\beta$-steps of the bounded evaluation
strategies might not be a reasonable cost model, that is, it might not
be a reliable complexity measure. This is because no reasonable cost
models for the $\l$-calculus were known at the time. But the
understanding of cost models for the $\l$-calculus made significant
progress in the last few years. Since the nineties, it is known that
the number of steps for \emph{weak} strategies (\ie not reducing under
abstraction) is a reasonable cost
model~\cite{DBLP:conf/fpca/BlellochG95}, where \emph{reasonable} means
polynomially related to the cost model of Turing machines. It is only
in 2014, that a solution for the general case has been obtained: the
length of leftmost evaluation to normal form was shown to be a
reasonable cost model in~\cite{DBLP:journals/corr/AccattoliL16}. In
this work we essentially update the study of the bounding power of
multi types with the insights coming from the study of reasonable cost
models. In particular, we provide new answers to the question of
whether denotational semantics can really be used as an accurate tool
for complexity analyses.

\paragraph{Size explosion and lax bounds.} The study of cost models made
clear that evaluation lengths are independent from the size of their results.
The skepticism about taking the number of $\beta$-steps as a reliable complexity
measure comes from the \emph{size explosion problem}, that is, the fact that the
size of terms can grow exponentially with respect to the number of
$\beta$-steps.  When $\lambda$-terms are used to encode decision
  procedures, the normal forms (encoding true or false) are of constant size,
  and therefore there is no size explosion issue. But when $\lambda$-terms are
  used to compute other normal forms than Boolean values, there are families of
terms $\{\tm_n\}_{n\in\nat}$ where $\tm_n$ has size linear in $n$, it evaluates
to normal form in $n$ $\beta$-steps, and produces a result $\tmtwo_n$ of size
$\Omega(2^n)$, \ie exponential in $n$. Moreover, the size explosion problem is
extremely robust, as there are families for which the size explosion is
independent of the evaluation strategy. The difficulty in proving that the
length of a given strategy provides a reasonable cost model lies precisely in
the fact that one needs a compact representation of normal forms, to avoid to
fully compute them (because they can be huge and it would be too expensive).  A
divulgative introduction to reasonable cost models and size explosion is
\cite{AccattoliLSFA}.
 
Now, multi typings do bound the number of $\beta$-steps of reasonable
strategies, but these bounds are too generous since they bound at the
same time the length of evaluations and the size of the normal
forms. Therefore, even a notion of \emph{minimal} typing (in
  the sense of being the smallest derivation) provides a bound that
in some cases is exponentially worse than the number of $\beta$-steps.

Our observation is that the typings themselves are in fact much
bigger than evaluation lengths, and so the widespread point of view for which multi
types---and so the relational model of linear logic---faithfully
capture evaluation lengths, or even the complexity, is misleading.

\subsection*{Contributions}
\paragraph{The tightening technique.} Our starting point is a technique
introduced in a technical
report by~\cite{bernadet13}.  They study the case of strong
normalisation, and present a multi type system where typing
derivations of terms provide an \emph{upper bound}
on the number of $\beta$-steps to normal
form.  More interestingly, they show that every strongly normalising
term admits a typing derivation that is sufficiently tight, 
where the
obtained bound is \emph{exactly} the length of the longest
$\beta$-reduction path.  This improved on previous results,
\eg~\cite{DBLP:conf/fossacs/BernadetL11,Bernadet-Lengrand2013} where
multi types provided the exact measure of longest 
evaluation paths
\emph{plus the size of the normal forms} which, as discussed above,
can be exponentially bigger.  Finally, they enrich the structure of
base types so that, for those typing derivations providing the exact
lengths, the type of a term gives the structure (and hence the
size) of its normal form.  This paper embraces this tightening
technique, simplifying it with the use of \emph{tight constants} for
base types, and generalising it to a range of other evaluation strategies,
described below.

It is natural to wonder how natural the tightening technique is---a
malicious reader may indeed suspect that we are cooking up an ad-hoc
way of measuring evaluation lengths, betraying the
\emph{linear-logic-in-disguise} spirit of multi types. To remove any
doubt, we show that our tight typings are actually isomorphic to
minimal multi typings without tight constants. Said differently, the
tightening technique turns out to be a way of characterising minimal
typings in the standard multi type framework (aka the relational model).
Let us point out that, in the literature, there are
characterisations of minimal typings (so-called principal typings) only for normal forms, and they
extend to non-normal terms only \emph{indirectly}, that is, by subject
expansion of those for normal forms. Our approach, instead, provides a
\emph{direct} description, for any typable term.

\paragraph{Modular approach.} 
We develop all our results by using a unique schema that modularly
applies to different evaluation strategies.  Our approach isolates the key
concepts for the correctness and completeness of multi types,
providing a powerful and modular technique, having at least two
by-products. First, it reveals the relevance of \emph{neutral terms}
and of their properties with respect to types.  Second, the concrete
instantiations of the schema on four different cases always require
subtle definitions, stressing the key conceptual properties of each
case study.

\paragraph{Head and leftmost evaluation.} Our first application of the
tightening technique is to the head and leftmost evaluation
strategies. The head case is the simplest possible one. The leftmost
case is the natural iteration of the head one, and the only known
strong strategy whose number of steps provides a reasonable cost
model~\cite{DBLP:journals/corr/AccattoliL16}. 
Multi types bounding the lengths of leftmost normalising terms
have been also studied in~\cite{DBLP:conf/ifipTCS/KesnerV14}, but the
exact number of steps taken by the leftmost strategy has not
been measured via multi types before---therefore, this is a new
result, as we now explain.

The study of the head and the leftmost strategies, at first sight, seems to be a minor reformulation of de Carvalho's results about measuring via  multi types the length of executions of the Krivine abstract machine (shortened KAM)---implementing weak head evaluation---and of the iterated KAM---that implements leftmost evaluation \cite{DBLP:journals/corr/abs-0905-4251}. The study of cost models is here enlightening: de Carvalho's iterated KAM does implement leftmost evaluation, but the overhead of the machine (that is counted by de Carvalho's measure) is exponential in the number of $\beta$-steps, while here we only measure the number of $\beta$-steps, thus providing a much more parsimonious (and yet reasonable) measure.

The work of de Carvalho, Pagani and Tortora de Falco \cite{DBLP:journals/tcs/CarvalhoPF11}, using the relational model of linear logic to measure evaluation lengths in proof nets, is also closely related. They do not however split the bounds, that is, they do not have a way to measure separately the number of steps and the size of the normal form. Moreover, their notion of cut-elimination by levels does not correspond to leftmost evaluation.

\paragraph{Maximal evaluation.} We also apply the technique to the
\emph{maximal strategy}, which takes the maximum number of steps to normal form,
if any, and diverges otherwise.  The maximal strategy has been bounded in
\cite{DBLP:conf/fossacs/BernadetL11}, and exactly measured in~\cite{bernadet13}
via the idea of tightening, as described above.  With respect to~\cite{bernadet13}, our technical development is simpler. The
differences are:

\begin{enumerate}
\item 
  \emph{Uniformity with other strategies}: The typing system used
  in~\cite{bernadet13} for the maximal strategy has a special rule for typing a
  $\lambda$-abstraction whose bound variable does not appear in the body. This
  special case is due to the fact that the \emph{empty} multi type is forbidden
  in the grammar of function types. Here, we align the type grammar with that
  used for other evaluation strategies, allowing the empty multi type, which in
  turn allows the typing rules for $\lambda$-abstractions to be the same as for
  head and leftmost evaluation. This is not only simpler, but it also
  contributes to making the whole approach more uniform across the different
  strategies that we treat in the paper. Following the head and leftmost
  evaluation cases, our completeness theorem for the maximal strategy bears
  quantitative information (about \eg evaluation lengths), in contrast
  with~\cite{bernadet13}.

\item 
  \emph{Quantitative aspects of normal forms}: Bernadet and Graham-Lengrand
  encode the shape of normal forms into base types. We simplify this by only
  using two tight constants for base types. On the other hand, we decompose the
  actual size of a typing derivation as the sum of two quantities: the first one
  is shown to match the maximal evaluation \emph{length} of the typed term, and
  the second one is shown to match the \emph{size} of its normal form together
  with the size of all terms that are erased by the evaluation process.
  Identifying what the second quantity captures is a new contribution.
\item \emph{Neutral terms}: 
  we emphasise the key role of neutral terms in the technical development
  by describing their specificities with respect to typing.
  This is not explicitly broached in~\cite{bernadet13}.

\end{enumerate}

\paragraph{Linear head evaluation.} Last, we apply the tightening
technique to \emph{linear} head evaluation \cite{DBLP:journals/tcs/MascariP94,Danos04headlinear} ($\lh$ for short), formulated in
the linear substitution calculus (LSC) \cite{DBLP:conf/rta/Accattoli12,DBLP:conf/popl/AccattoliBKL14}, a $\l$-calculus with explicit substitutions that is strongly related to linear logic proof nets, and also a minor variation over a calculus by Milner \cite{DBLP:journals/entcs/Milner07}. The literature contains a
characterisation of $\lh$-normalisable
terms~\cite{DBLP:conf/ifipTCS/KesnerV14}. Moreover,~\cite{Carvalho07}  measures
the executions of the KAM, a result that can also be interpreted as a measure of
$\lh$-evaluation. What we show however is stronger, and somewhat
unexpected.

To bound $\lh$-evaluation, in fact, we can strongly 
stand on the bounds obtained for head evaluation.
More precisely, the result for the exact bounds for \emph{head} evaluation takes
only into account  the number of abstraction and
application typing rules. For \emph{linear} head evaluation, instead, we simply
need to count also the axioms, \ie\  the rules typing variable
occurrences, nothing else. It turns out  that the length of a
linear head evaluation plus the size of the linear head normal form is
\emph{exactly} the size of the tight typing.

Said differently, multi typings simply encode evaluations in the LSC.
In particular, we do not have to adapt multi types to the LSC, as for
instance de Carvalho does to deal with the KAM. It actually is the
other way around. As they are, multi typings naturally measure
evaluations in the LSC. To measure evaluations in the $\l$-calculus,
instead, one has to forget the role of the axioms. The best
way to stress it, probably, is that the LSC is the computing device
behind multi types.

\longshort{Most proofs have been moved to the Appendix.}{
  Most proofs are provided in the long version of this paper~\cite{AccattoliGLK2018-long}.}

\subsection*{Other Related Works}

Apart from the papers already cited, let us mention some other related works. A
recent, general categorical framework to define intersection and multi type
systems is in~\cite{DBLP:journals/pacmpl/MazzaPV18}.

While the inhabitation problem is undecidable for idempotent intersection
types~\cite{Urzyczyn99}, the quantitative aspects provided by multi types make
it decidable~\cite{DBLP:conf/ifipTCS/BucciarelliKR14}.
Intersection type are also used in~\cite{DBLP:conf/popl/DudenhefnerR17} to give
a bounded dimensional description of $\lambda$-terms via a notion of
\emph{norm}, which is resource-aware and orthogonal to that of {\it rank}. It is
proved that inhabitation in bounded dimension is decidable (EXPSPACE-complete)
and subsumes decidability in rank $2$~\cite{Urzyczyn09}.

Other works propose a more practical perspective on resource-aware analyses for
functional programs. In particular, type-based techniques for automatically
inferring bounds on higher-order functions have been developed, based on sized
types~\cite{HughesPS96,DBLP:conf/ifl/PortilloHLV02,DBLP:conf/ifl/VasconcelosH03,DBLP:journals/pacmpl/AvanziniL17}
or amortized
analysis~\cite{DBLP:conf/popl/HofmannJ03,HoffHof10,DBLP:journals/jar/JostVFH17}.
This led to practical cost analysis tools like \emph{Resource-Aware
  ML}~\cite{HAH12} (see \url{raml.co}).  Intersection types have been
used~\cite{DBLP:conf/types/SimoesHFV06} to address the \emph{size aliasing}
problem of sized types, whereby cost analysis sometimes overapproximates cost to
the point of losing all cost information~\cite{DBLP:conf/ifl/PortilloHLV02}. How
our multi types could further refine the integration of intersection types with
sized types is a direction for future work, as is the more general combination
of our method with the type-based cost analysis techniques mentioned above.

\section{A Bird's Eye View} 
\label{s:bird}
Our study is based on a schema that is repeated for different
evaluation strategies, making most notions parametric
in the strategy $\tostrat$ under study. The following concepts
constitute the main ingredients of our technique:




\begin{enumerate}
\item \emph{Strategy, together with the normal, neutral, and abs predicates}: there
  is a (deterministic) evaluation strategy $\tostrat$ whose normal
  forms are characterised via two related predicates, $\sysnormal\system\tm$
    and $\sysneutral\system\tm$, the intended meaning
    of the second one is that $\tm$ is $\system$-normal and can never behave as an
    abstraction (that is, it does not create a redex when applied to
    an argument). We further parametrise  also this last notion by using a
    predicate $\sysabs\system\tm$ identifying abstractions, because
    the definition of deterministic strategies  requires some subterms
    to not be abstractions.


  \item \emph{Typing derivations}: there is a multi types system which has three features:
\begin{itemize}
    \item \emph{Tight constants}: there are two new type constants
      $\neutype$ and $\abstype$, and rules to introduce them.  As
      their name suggests, the constants $\neutype$ and $\abstype$
      are  used to type terms whose normal form is neutral
      or an abstraction, respectively.

  \item \emph{Tight derivations}: there is a notion of tight derivation 
 that requires a special use of the constants. 

    \item \emph{Indices}: typing judgements have the shape $\Deri[(\steps, \result)] {\typctx}{\tm}{\type}$, where $\steps$ and $\result$ are indices meant to count,
respectively, the number of steps to normal form and the size of the normal form.
  \end{itemize}

  \item \emph{Sizes}: there is a notion of size of terms that depends on the strategy, noted $\syssize{\stratsym}{\tm}$. Moreover, there is a notion of size of typing derivations $\syssize{\stratsym}{\tderiv}$ that also depends on the strategy / type system, that coincides with the sum of the indices associated to the last judgement of $\tderiv$.

  \item \emph{Characterisation}: we prove that  $\Deri[(\steps,
    \result)] {\typctx}{\tm}{\type}$ is a tight typing relatively to
    $\tostrat$ if and only if
    there exists an  $\stratsym$ normal term     $\tmtwo$
   such that $\tm \rightarrow_\stratsym^{\steps/2} \tmtwo$ and 
    $\syssize{\stratsym}{\tmtwo} = \result$.

  \item \emph{Proof technique}: the characterisation is obtained
    always through the same sequence of intermediate
    results. Correctness follows from the fact that all tight typings
    of normal forms precisely measure their size, a \emph{substitution
      lemma} for typing derivations and \emph{subject reduction}.
    Completeness follows from the fact that every normal form admits a
    tight typing, an \emph{anti-substitution lemma} for typing
      derivations, and \emph{subject expansion}.

\item \emph{Neutral terms}: we  stress the relevance of neutral terms in normalisation proofs from a typing perspective. In particular, correctness theorems always rely on a lemma about them. Neutral terms are a common concept in the
  study of $\l$-calculus, playing a key role in, for instance, the
  reducibility candidate technique \cite{Girard:1989:PT:64805}. 

\end{enumerate}

The proof schema is illustrated in the next section on two standard strategies,
namely \emph{head} and \emph{leftmost-outermost evaluation}. It is then slightly adapted to deal with \emph{maximal evaluation} in \refsect{maximal} and \emph{linear head evaluation} in \refsect{linear}.

\paragraph{Evaluation systems.}
Each case study treated in the paper relies on the same properties of the strategy $\tostrat$ and the related predicates $\sysnormal\system\tm$, $\sysneutral\system\tm$, and $\sysabs\system\tm$, that we collect under the notion of \emph{evaluation system}. 

\begin{definition}[Evaluation system]

Let $\Terms_\system$ be a set of terms, $\tostrat$ be a (deterministic) strategy and $\sysnormalpr\system$, $\sysneutralpr\system$, and $\sysabspr\system$ be predicates on $\Terms_\system$. All together they form an \emph{evaluation system} $\system$ if for all $\tm, \tmtwo, \tmtwo_1, \tmtwo_2 \in \Terms_\system$:
\begin{enumerate}
\item \label{p:determinism} \emph{Determinism of $\tostrat$}: if $\tm \tostrat \tmtwo_1$ and $\tm \tostrat \tmtwo_2$ 
then $\tmtwo_1 = \tmtwo_2$.
\item \label{p:normal-forms}\emph{Characterisation of $\system$-normal terms}: $\tm$ is $\tostrat$-normal if and only if $\sysnormal\system\tm$. 
\item \label{p:neutral-forms}\emph{Characterisation of $\system$-neutral terms}: $\sysneutral\system\tm$  if and only if  $\sysnormal\system\tm$ and $\sysnotabs\system\tm$.
\end{enumerate}
\end{definition}

Given a strategy $\tostrat$ we use $\Rewn[k]{\stratsym}$ for its $k^{th}$ iteration and $\Rewn{\stratsym}$ for its transitive closure.


\section{Head and Leftmost-Outermost Evaluation}
\label{sect:head-skeleton}

In this section we consider two evaluation systems at once. The two strategies are the famous \emph{head} and \emph{leftmost-outermost
  evaluation}. We treat the two cases together to stress the modularity of our technique.
The set of $\l$-terms $\lterms$ is given by ordinary $\l$-terms:
\begin{center}
$\begin{array}{c\colspace \colspace \colspace ccc}
    \textsc{$\l$-Terms} & \tm,\tmtwo & \recdef & \var \mid \la\var\tm\mid \tm\tmtwo
  \end{array}$
\end{center}
\paragraph{Normal, neutral, and abs predicates.} The predicates $\spnormalpr$ and $\sknormalpr$ defining head and leftmost-outermost (shortened LO in the text and $\lo$ in mathematical symbols) normal terms are in \reffig{normalforms}, and they are based on two  auxiliary predicates defining neutral terms: \emph{$\spneutralpr$} and \emph{$\skneutralpr$}---note that $\skneutral{\tm}$ implies $\spneutral{\tm}$. The predicates \emph{$\sysabs\spi\tm$} and \emph{$\sysabs\ske\tm$} are
equal for the systems $\spi$ and $\ske$ and they are true simply when $\tm$ is  an abstraction.

\paragraph{Small-step semantics.} The \emph{head} and \emph{leftmost-outermost}
strategies $\tohd$ and $\tolo$  are both defined in \reffig{strategies}. Note that these
definitions rely on the predicates defining neutral terms and
abstractions. 

\begin{proposition}[Head and LO evaluation systems]
  Let $\system \in \set{\hd,\lo}$.
  Then\\ $(\lterms,\tostrat, \sysneutralpr \system, \sysnormalpr \system,
  \sysabspr\system)$ is an evaluation system.
\end{proposition}

\longshort{The proof is routine, and it is then omitted also from the Appendix.}{
  The proof is routine.
}

\begin{figure*}
  \centering
  \ovalbox{\small
    $
    \begin{array}{c@{\qquad}c@{\qquad}c@{\qquad}c@{\qquad}c}
      \multicolumn{4}{c}{\textsc{Head normal forms}}
      \\[3pt]
      
      \infer[]{\spneutral{\var}}{ \phantom{\spneutral{\var}} } 
      &
      \infer[]{ \spneutral{\tm \tmtwo} }{ \spneutral{\tm} }
      &
      \infer[]{ \spnormal{\tm} }{ \spneutral{\tm } }
      &
      \infer[]{ \spnormal{\la\vartwo\tm} }{ \spnormal{\tm}} 
      \\\\[-3pt]
      \hline\\[-8pt]
      \multicolumn{4}{c}{\textsc{Leftmost-outermost normal forms}}
      \\[3pt]
      \infer[]{\skneutral{\var}}{ \phantom{\skneutral{\var}} } 
      &
      \infer[]{ \skneutral{\tm \tmtwo} }{ \skneutral{\tm} \quad \sknormal{\tmtwo}}
      &
      \infer[]{ \sknormal{\tm} }{ \skneutral{\tm}  }
      &
      \infer[]{ \sknormal{\la\vartwo\tm} }{ \sknormal{\tm}} 
    \end{array}
    $
  }
  \caption{Head and leftmost-outermost neutral and normal terms}
  \label{fig:normalforms}
\end{figure*}

\begin{figure*}
  \centering
  \ovalbox{\small
    $\begin{array}{c}
      \multicolumn1{c}{\textsc{Head evaluation}}\\
      \infer{(\la\var \tmthree) \tmfour \tohd \tmthree \isub \var \tmfour}
            {\phantom{(\la\var \tmthree) \tmfour \tohd \tmthree \isub \var \tmfour}}
      \sep
      \infer{\la\var \tm \tohd \la\var \tmtwo}{\tm \tohd \tmtwo}
      \sep
      \infer{\tm \tmthree \tohd \tmtwo \tmthree}{
        \tm \tohd \tmtwo \quad \sysnotabs\spi\tm
      }
      \\\\[-5pt]\hline\\[-8pt]
      \multicolumn1{c}{\textsc{Leftmost-outermost evaluation}}\\[3pt]
      \infer{(\la\var \tmthree) \tmfour \tolod \tmthree \isub \var \tmfour}
            {\phantom{(\la\var \tmthree) \tmfour \tolod \tmthree \isub \var \tmfour}}
      \sep
      \infer{\la\var \tm \tolod \la\var \tmtwo}{\tm \tolod \tmtwo}
      \sep
      \infer{\tm \tmthree \tolod \tmtwo \tmthree}{
        \tm \tolod \tmtwo \quad  \sysnotabs\ske\tm
        }      
      \sep
      \infer{\tmthree \tm \tolod \tmthree \tmtwo}{
        \skneutral {\tmthree} \quad  \tm \tolod \tmtwo
      } 
    \end{array}$
    \newline
   }
  \caption{Head and leftmost-outermost strategies}
  \label{fig:strategies}
\end{figure*}

\paragraph{Sizes.} The notions of \emph{head size} $\spsize\tm$ and \emph{LO size} $\sksize\tm$ of a term $\tm$ are  defined as follows---the difference is on applications:
\[
\begin{array}{r@{\ }l\colspace \colspace \colspace r@{\ }l}
\multicolumn{2}{c}{\textsc{Head size}} & \multicolumn{2}{c}{\textsc{LO size}}\\
\spsize\var &\defeq 0 & \sksize\var &\defeq 0\\
\spsize{\la\var\tmtwo} &\defeq \spsize\tmtwo + 1
&\sksize{\la\var\tmtwo} &\defeq \sksize\tmtwo + 1 \\
\spsize{\tmtwo \tmthree} &\defeq \spsize\tmtwo + 1
&\sksize{\tmtwo \tmthree} &\defeq \sksize\tmtwo+\sksize\tmthree + 1
\end{array}
\]



\paragraph{Multi types.} We define the following notions about types.
\label{s:typing}

\begin{figure}
  \centering
  \ovalbox{
    $\begin{array}{c}
      \infer[\ax]{\Deri[(0, 0)] {\var \col \single \type} \var \type}{} 
      \\\\
      \infer[\funsteps]{
        \Deri[(\steps+ 1, \spine)]{\typctx \sm \var} {\la\var\tm} {\typctx(\var) \rightarrow \type}
      }{\Deri[(\steps, \spine)] {\typctx } \tm \type} 
      \msep 
      \infer[\funresult]{
        \Deri[(\steps, \spine +1)] {\typctx \sm \var } { \la\var\tm } \abstype
      }{
        \Deri[(\steps, \spine)] {\typctx } \tm \nf\sep \tightpred{\typctx(\var)}
      }  \\\\

      \infer[\appsteps]{
        \Deri[(\steps + \steps' + 1, \spine + \spine')] {\typctx \mplus \typctxtwo}
             {\tm \tmtwo} \type
      }{
        \Deri[(\steps, \spine)] \typctx \tm {\M \rightarrow \type}
        \quad
        \Deri[(\steps', \spine')] {\typctxtwo} \tmtwo {\M}
      }
      \msep 
      \infer[\appresult<\spi>]{
        \Deri[(\steps, \spine+1)] {\typctx} {\tm \tmtwo} \neutype
      } {\Deri[(\steps, \spine)] \typctx \tm \neutype}\\\\

      \infer[\many]{
        \Deri[(+_{\iI} \steps_i, +_{\iI} \spine_i )] {
          \mplus_{\iI}\typctxtwo_i  } \tm {\MSigma {\type_i} {\iI}
        }
      }{
        (\Deri[(\steps_i, \spine_i)] {\typctxtwo_i} \tm {\type_i})_{\iI}
      } 
      \msep 
      \infer[\appresult<\ske>]{
        \Deri[(\steps + \stepstwo, \spine + \spinetwo+1)] {\typctx \mplus \typctxtwo} {\tm \tmtwo} \neutype
      } {
        \Deri[(\steps, \spine)] \typctx \tm \neutype 
        \quad
        \Deri[(\stepstwo, \spinetwo)] \typctxtwo \tmtwo \nf
      }

    \end{array}$

  }
  \caption{Type system for head and LO evaluations}
  \label{fig:type-system}
\end{figure}

\begin{itemize}
  \item \emph{Multi types} are defined by the following grammar:
  \[
  \begin{array}{r\colspace \colspace rcl}
    \textsc{Tight constants} & \tight & \grameq  &  \neutype \mid \abstype \\
    \textsc{Types} & \type, \typetwo  & \grameq  & \tight \mid  \atomtype \mid \ty{\M}{\type} \\   
    \textsc{Multi-sets} &\M  & \grameq & \mult{\type_i}_{\iI}\  (I \mbox{ a finite set}) \\
  
  \end{array}
  \]
  where $\atomtype$ ranges over a non-empty set of \emph{atomic types}
  and $\MSigma{\ldots}{}$ denotes the multi-set constructor.
  
  \item \emph{Examples} of multisets: $\MSigma{\type, \type, \typetwo}{}$ is a multi-set containing two occurrences of
  $\type$ and one occurrence of $\typetwo$, and $\emm$ is the empty multi-set.
  
  \item A \emph{typing context} $\typctx$ is a map
  from variables to finite multisets $\M$ of types
  such that only finitely many variables are not mapped to the empty multi-set $\emm$.
  We write $\dom{\typctx}$ for the domain of $\typctx$, \ie\ the set
  $\set{x \mid \typctx(\var) \neq \emm}$.
  
  \item \emph{Tightness}: we use the notation
    $\mtight$ for $ \mult{\tight}_{\iI}\ (I \mbox{ a finite set})$.
    Moreover, we write
    $\tightpred \type$ if $\type$ is of the form $\tight$,
    $\tightpred\M$ if $\M$ is of the form $\mtight$,
    and $\tightpred \typctx$ if $\tightpred  {\typctx(\var)}$ for all $\var$,
    in which case we also say that $\typctx$ is \emph{\precise}.
  
  \item The 
  \emph{multi-set union} $\mplus$ is extended to typing contexts point-wise,
  \ie\  $\typctx \mplus \typctxtwo$ 
  maps  each variable $\var$ to $\typctx(\var) \mplus \typctxtwo(\var)$.
  This notion is extended to several contexts as expected so that
  $\mplus_{\iI} \typctx_i$ denotes a finite union of contexts (when $I = \emptyset$ the
notation is to be understood as the empty context).
  We write $\typctx; \var \col \M$ for $\typctx\mplus (\var \mapsto \M)$
  only if  $\var \notin \dom{\typctx}$. More generally,
  we write $\typctx; \typctxtwo$ if the intersection between
  the domains of $\typctx$ and $\typctxtwo$ is empty.
  
  \item The \emph{restricted} context $\typctx$ with respect to the variable $\var$,
  written $\typctx \sm \var$ is defined by 
  $(\typctx \sm \var)(\var) \defeq \emm$ and
   $(\typctx \sm \var)(\vartwo) \defeq \typctx(\vartwo)$ if $\vartwo \neq \var$.

\end{itemize}

\paragraph{Typing systems.} There are two typing systems,
one for head and one for LO evaluation.
Their typing rules 
are presented in \reffig{type-system},
the head system $\spi$ contains
all the rules except $\appresult<\ske>$,
the LO system $\ske$ contains
all the rules except $\appresult<\spi>$. 

Roughly, the intuitions behind the typing rules are 
(please ignore the indices $\steps$ and
$\result$ for the time being):
\begin{itemize}
  \item \emph{Rules $\ax$, $\funsteps$, and $\appsteps$}: this rules are essentially the traditional rules for multi types for head and LO evaluation~(see \eg\ \cite{BKV17}), modulo the presence of the indices.
  \item \emph{Rule $\many$}: this is a structural rule allowing typing terms with a multi-set of  types. In some presentations of multi types $\many$ is hardcoded in the right premise of the $\appsteps$ rule (that requires a multi-set). For technical reasons, 
it is preferable to separate it from $\appsteps$. Morally, it corresponds to the $!$-promotion rule in linear logic.
  \item \emph{Rule $\funresult$}: $\tm$ has already been tightly typed, and all the types associated to $\var$ are also tight constants. Then $\la\var\tm$ receives the tight constant $\abstype$ for abstractions. The consequence is that this abstraction can no longer be applied, because it has not an arrow type, and there are no rules to apply terms of type $\abstype$. Therefore, the abstraction constructor cannot be consumed by evaluation and it ends up in the normal form of the term, that has the form $\la\var\tm'$.

  \item \emph{Rule $\appresult<\hd>$}: $\tm$ has already been tightly typed with $\neutral$ and so morally it head normalises to a term $\tm'$ having neutral form $\var \tmthree_1 \dots \tmthree_k$. The rule adds a further argument $\tmtwo$ that cannot be consumed by evaluation, because $\tm$ will never become an abstraction. Therefore, $\tmtwo$ ends up in the head normal form $\tm'\tmtwo$ of $\tm\tmtwo$, that is still neutral---correctly, 
so that $\tm\tmtwo$ is also typed with $\neutral$. Note that there is no need to type $\tmtwo$ because head evaluation never enters into arguments.
  
  \item \emph{Rule $\appresult<\lo>$}: similar to rule $\appresult<\hd>$, except that LO evaluation enters into arguments and so the added argument $\tmtwo$ now also has to be typed, and with a tight constant. Note a key difference with $\appsteps$: in $\appresult<\lo>$ the argument $\tmtwo$ is typed exactly once (that is, the type is not a multi-set)---correctly, because its LO normal form $\tmtwo'$ appears exactly once in the LO normal form $\tm'\tmtwo'$ of $\tm\tmtwo$ (where $\tm'$ is the LO normal form of $\tm$).

\item 
\emph{Tight constants and predicates}: there is of course a correlation between the tight constants $\neutral$ and $\abstype$ and the predicates $\sysneutralpr\system$ and $\sysabspr\system$. Namely, a term $\tm$ is $\system$-typable with $\neutral$ if and only if the $\system$-normal form of $\tm$ verifies the predicate $\sysneutralpr\system$, as we shall prove. For the tight constant $\abstype$ and the predicate $\sysabspr\system$ the situation is similar but weaker: if the $\system$-normal form of $\tm$ verifies $\sysabspr\system$ then $\tm$ is typable with $\abstype$, but not the other way around---for instance a variable is typable with $\abstype$ without being an abstraction.

\item The type systems are not syntax-directed, \eg\ given an abstraction (resp. an application), it can be typed with rule $\funresult$ or $\funsteps$ (resp. $\appresult$ or $\appsteps$), depending on whether the constructor typed by the rule ends up in the normal form or not. Thus for example, given the term $\id \id$, where $\id$ is the identity function $\la \varthree \varthree$, the second occurrence of $\id$ can be typed with $\abstype$ using rule $\funresult$, while the first one can be typed with $\ty{\mult{\abstype}}{\abstype}$ using rule $\funsteps$.
\end{itemize}

Typing judgements are of the form
$\Deri[(\steps, \result)] {\typctx}{\tm}{\type}$,
where $(\steps, \result)$ is a pair of integers
whose intended meaning is explained in the next paragraph.
We write
$\tderiv \exder[\system] \Deri[(\steps, \result)] {\typctx}{\tm}{\type}$,
with $\system$ being either $\hd$ or $\lo$,
if $\tderiv$ is a typing derivation in the system
$\system$ and ends
in the judgement
$\Deri[(\steps, \result)] { \typctx}{\tm}{ \type}$.

\paragraph{Indices.} The roles of $\steps$ and $\result$ can be described as follows:
\begin{itemize}

\item \emph{$\steps$ and $\beta$-steps}:
  $\steps$ counts the rules of the derivation that can be used to form $\beta$-redexes,
  \ie\  the number of $\funsteps$ and $\appsteps$ rules.
  Morally, $\steps$ is at least twice the number of $\beta$-steps to normal form
  because typing a $\beta$-redex requires two rules.
  For tight typing derivations (introduced below), we are going to prove that $\steps$ is the exact (double of the) length of the evaluation of
  the typed term to its normal form, according to the chosen evaluation strategy. 

\item \emph{$\result$ and size of the result}:
  $\result$ counts the rules typing constructors
  that cannot be consumed by $\beta$-reduction
  according to the chosen evaluation strategy.
  It counts the number of $\funresult$ and $\appresult$.
  These rules type the result of the evaluation,
  according to the chosen strategy, and measure the size of the result.
  Both the notion of result and the way its size is measured depend on the evaluation strategy.

\end{itemize}

\paragraph{Typing size.} We define both the \emph{head} and the \emph{LO size} $\syssize\spi\tderiv$ and $\syssize\ske\tderiv$ of a typing derivation $\tderiv$ as the number of rules in $\tderiv$,
not counting rules $\ax$ and $\many$. The size of a derivation is
 reflected by the  pair of indices $(\steps, \result)$ on
its final judgement: whenever $\tderiv \exder[\system] \Deri[(\steps,
  \result)] {\typctx}{\tm}{\type}$, we have $\steps + \result =
\syssize{\stratsym}{\tderiv}$.  Note indeed that every rule (except $\ax$ and $\many$)
adds exactly $1$ to this size.

For systems $\hd$ and $\lo$, the indices on typing judgements are not really needed, as $\steps$ can be recovered as the number of $\funsteps$ and $\appsteps$ rules, and $\result$ as the number of $\funresult$ and $\appresult<\hd>$/$\appresult<\lo>$ rules. We prefer to make them explicit because 1) we want to stress the separate counting, and 2) for linear head evaluation in \refsect{linear} the counting shall be more involved, and the indices shall not be recoverable.

The fact that $\ax$ is not counted for $\syssize\spi\tderiv$ and
$\syssize\ske\tderiv$ shall change in \refsect{linear}, where we show
that counting $\ax$ rules corresponds to measure evaluations in the
linear substitution calculus. The fact that $\many$ is not counted,
instead, is due to the fact that it does not correspond to any
constructor on terms. A further reason is that the rule may be
eliminated by absorbing it in the $\appsteps$ rule, that is the only
rule that uses multi-sets---it is however technically convenient to
separate the two.

\paragraph{Subtleties and easy facts.} Let us overview some peculiarities and consequences of the definition of our type systems.

  \begin{enumerate} 

  \item \emph{Relevance}: 
    No weakening is allowed in axioms. An easy induction on typing derivations shows that a variable declaration $\var : \M \neq \emptymset$ appears explicitly in the typing context $\typctx$ of a type derivation for $\tm$ only if $\var$ occurs free in some \emph{typed} subterm of $\tm$. In system $\ske$, all subterms of $\tm$ are typed, and so $\var : \M \neq \emptymset$  appears in $\typctx$ if and only if $\var \in \fv\tm$. In system $\spi$, instead, arguments of applications might not be typed (because of rule $\appresult<\spi>$), and so there may be $\var \in \fv\tm$ but not appearing in $\typctx$.

\item \emph{Vacuous abstractions}:  we rely on the convention that the two abstraction rules can always abstract a variable $\var$ not explicitly occurring in the context.
Indeed, in the $\funsteps$ rule, if $\var \notin \dom{\typctx}$, then
$\typctx \sm \var$ is equal to $\typctx$ and $\typctx(\var)$ is $\emm$, while in
the $\funresult$ rule, if  $\var \notin \dom{\typctx}$, then $\typctx(\var)$ is $\emm$
and thus $\tightpred{\emm}$ holds. 

\item \emph{Head typings and applications}: note the $\appresult<\spi>$ rule types an application $\tm\tmtwo$ without typing the right subterm $\tmtwo$. This matches the fact that $\tm\tmtwo$ is a head normal form when $\tm$ is, independently of the status of $\tmtwo$.
  \end{enumerate}

\paragraph{Tight derivations.} A given term $\tm$ may have many different typing derivations, indexed by different pairs $(\steps, \result)$. They always provide upper bounds on $\tostrat$-evaluation lengths and lower bounds on the $\system$-size $\syssize{\stratsym}{\cdot}$ of $\system$-normal forms, respectively. The interesting aspect of our type systems, however, is that there is a simple description of a class of typing derivations that provide \emph{exact} bounds for these quantities, as we shall show. Their definition relies on tight constants.

\begin{definition}[\Precise derivations]\label{def:tightderiv}\strut
  

  Let $\system \in \set{\spi, \ske}$. A derivation $\tderiv \exder[\system] \Deri[(\steps, \result)]{\typctx}{\tm}{\typetwo}$
  is \emph{\precise} 
  if $\tightpred\typetwo$ and $\tightpred{\typctx}$. 
\end{definition}
Let us stress that, remarkably, tightness is expressed as a property of the last judgement only. This is however not so unusual: characterisations of weakly normalising terms via intersection/multi types also rely on properties of the last judgement only, as discussed in~\refsect{shrinking}.

In \refsect{shrinking}, in particular, we show the the size of a tight derivation for a
term $\tm$ is \emph{minimal} among derivations for $\tm$. Moreover, it is
also the same size of the minimal derivations making no use of tight
constants nor rules using them. Therefore, tight derivations may be
thought as a characterisation of minimal derivations.

\paragraph{Example.} Let $\tm_0 = (\la {\var_1} (\la {\var_0} \var_0 \var_1) \var_1) \id$, where
$\id$ is the  identity function $\la \varthree \varthree$.
Let us first consider the head evaluation of $\tm_0$ to $\spi$ normal-form:
\[ \begin{array}{l}
   (\la {\var_1} (\la {\var_0} \var_0 \var_1) \var_1) \id  \tohd 
   (\la {\var_0} \var_0 \id) \id  \tohd 
   \id \id \tohd
   \id
   \end{array} \] 
The evaluation sequence has length $3$. The head normal form has size $1$.  
To give a \precise\ typing for the term $\tm_0$ let us write $\abstype_1$ for $\ty{\mult{\abstype}}{\abstype}$.
Then, 
{\scriptsize
\[ \begin{prooftree}
   \prooftree
   \prooftree
     \prooftree
       \prooftree
         \Deri[(0,0)]{\var_0:\mult{\abstype_1}}{ \var_0 }{\abstype_1} \quad 
        \prooftree
         \Deri[(0,0)]{\var_1:\mult{\abstype}}{\var_1}{\abstype}
       \justifies{\Deri[(0,0)]{\var_1:\mult{\abstype}}{\var_1}{\mult{\abstype}}}
        \endprooftree
       \justifies{\Deri[(1,0)]{\var_0:\mult{\abstype_1}, \var_1:\mult{\abstype}}{ \var_0 \var_1}{\abstype}}
      \endprooftree 
      \justifies{\Deri[(2,0)]{\var_1:\mult{\abstype}}{\la {\var_0} \var_0 \var_1}{\ty{\mult{\abstype_1}}{\abstype}} }
      \endprooftree \quad
      \prooftree
      \Deri[(0,0)]{ \var_1:\mult{ \abstype_1} }{\var_1}{\abstype_1 }
\justifies{\Deri[(0,0)]{ \var_1:\mult{ \abstype_1} }{\var_1}{\mult{\abstype_1} } } 
      \endprooftree
    \justifies{\Deri[(3,0)] {\var_1:\mult{ \abstype, \abstype_1} }{ (\la {\var_0} \var_0 \var_1) \var_1}{\abstype }}
   \endprooftree
   \justifies{\Deri[(4,0)] {}{\la {\var_1} (\la {\var_0} \var_0 \var_1) \var_1}{\ty{\mult{ \abstype, \abstype_1}}{\abstype} } }
   \endprooftree
   \quad
   \prooftree
   \vdots
   \justifies{\Deri[(1,1)] {}{ \id  }{\mult{ \abstype, \abstype_1}}}
   \endprooftree
   \justifies{\Deri[(6,1)] {}{(\la {\var_1} (\la {\var_0} \var_0 \var_1) \var_1) \id  }{\abstype}} 
   \end{prooftree} \] 
}
Indeed, the pair $(6,1)$ represents $6/2=3$ evaluation steps to $\spi$ normal-form
and a head normal form of size $1$.


\subsection{Tight Correctness}\label{s:tightcorrectness}
Correctness of tight typings is the fact that whenever a term is
\emph{tightly} typable with indices $(\steps,\result)$,
then $\steps$ is exactly  (the
double of) the number of evaluation steps to $\system$-normal form
while  $\result$ is exactly the size of the $\system$-normal form.
Thus, tight typing in system $\spi$ (resp. $\ske$)
gives information about $\spi$-evaluation to $\spi$-normal form
(resp. $\ske$-evaluation to $\ske$-normal form).  The correctness theorem is always obtained
via three intermediate steps.

\paragraph{First step: tight typings of normal forms.}
The first step is to show that,
when a tightly typed term is a $\system$-normal form,
then the first index $\steps$ of its type derivation is $0$,
so that it correctly captures the (double of the) number of steps,
and the second index $\result$ coincides exactly with its $\system$-size.

\begin{toappendix}
\begin{proposition}[Properties of $\spi$ and $\ske$ tight typings for normal forms]
  \label{prop:tight-normal-forms-indicies}
Let $\system \in \set{\spi, \ske}$, $\tm$ be such that $\sysnormal\system\tm$, and
$\tderiv\exder[\system] \Deri[(\steps, \result)]\typctx{\tm}\type$ be a
typing derivation. 
\begin{enumerate}
\item 
\emph{Size bound}: $\syssize\system\tm \leq \syssize\system\tderiv$.

\item 
\emph{Tightness}: if $\tderiv$ is \precise\  then $\steps = 0$ and $\result = \syssize\system\tm$. 

\item 
\emph{Neutrality}: if $\type=\neutype$ then $\sysneutral{\system}{\tm}$.
\end{enumerate}
\end{proposition}
\end{toappendix}

The proof is by induction on the typing derivation $\tderiv$. Let us stress three points:
\begin{enumerate}
\item \emph{Minimality}: the size of typings of a normal form $\tm$ always bounds the size of $\tm$ (\refprop{tight-normal-forms-indicies}.1), and therefore tight typings, that provide an exact bound (\refprop{tight-normal-forms-indicies}.2), are typing of minimal size. For the sake of conciseness, in most of the paper we focus on tight typings only. In \refsect{shrinking}, however, we study in detail the relationship between arbitrary typings and tight typings, extending their minimality beyond normal forms.

\item \emph{Size of tight typings}: note that \refprop{tight-normal-forms-indicies}.2 indirectly shows that all tight typings have the same indices, and therefore the same size. The only way in which two tight typings can differ, in fact, is whether the variables in the typing context are typed with $\neutral$ or $\abstype$, but the structure of different typings is necessarily the same (which is also the structure of the $\system$-normal form itself).

\item \emph{Unveiling of a key structural property}: \refprop{tight-normal-forms-indicies} relies on the following interesting lemma about $\system$-neutral terms and tight typings.

\begin{toappendix}
\begin{lemma}[Tight spreading on neutral terms]
\label{l:tight-spreading-spi-ske}
  Let $\system \in \set{\spi, \ske}$, $\tm$ be such that $\spneutral \tm$, and $\tderiv\exder[\system] \Deri[(\steps, \result)]\typctx{\tm}\type$ be a typing derivation such that $\tightpred \typctx$.
  Then $\tightpred{\type}$.
\end{lemma}
\end{toappendix}

The lemma expresses the fact that tightness of neutral terms only
depends on their contexts. Morally, this fact is what makes tightness
to be expressible as a property of the final judgement only. We shall
see in \refsect{shrinking} that a similar property is hidden in more
traditional approaches to weak normalisation (see
\reflemma{shrinking-neutral-spreading}). Such a spreading property
  appears repeatedly in our study, and we believe that its isolation
  is one of the contributions of our work, induced by the modular and
  comparative study of various strategies.
\end{enumerate}

\paragraph{Second step: substitution lemma.}
Then one has to show that types, typings, and indices
behave well with respect to substitution,
which is essential, given that $\beta$-reduction is based on it. 

\begin{toappendix}
\begin{lemma}[Substitution and typings for $\spi$ and $\ske$]
  \label{l:typing-substitution-overview}
  The following rule is admissible in both systems $\spi$ and $\ske$: 
  \[
   \infer[\subslemma]{
     \Deri[(\steps + \steps', \result + \result')]
          {\typctx\mplus\typctxtwo}{\tm\isub\var\tmtwo }{\type}
   }{
     \Deri[(\steps, \result)]{\typctx} {\tmtwo}{{\M}}
     \quad
     \Deri[(\steps', \result')] {\typctxtwo; \var: {\M}} {\tm}{\type}
   }
  \]
  Moreover if the derivations of the premisses are \precise 
  then so is the derivation of the conclusion.
\end{lemma}
\end{toappendix}

The proof is by induction on the derivation of $\Deri[(\steps', \result')] {\typctxtwo; \var: {\M}} {\tm}{\type}$.

Note that the lemma also holds for $\M =\emm$, in which
case $\typctx$ is necessarily empty. In system $\ske$, it is also true that if $\M =\emm$ then $\var \notin \fv\tm$ and $\tm\isub\var\tmtwo = \tm$, because all free variables of $\tm$ have non empty type in the typing context. As already pointed out, in system $\spi$ such a matching between free variables and typing contexts does not hold, and it can be that $\M =\emm$ and yet $\var \in \fv\tm$ and $\tm\isub\var\tmtwo \neq \tm$.  

\paragraph{Third step: quantitative subject reduction.}
Finally, one needs to shows a quantitative form of type preservation along evaluation.
When the typing is tight,
every evaluation step decreases the first index $\steps$ of exactly 2 units,
accounting for the application and abstraction constructor \emph{consumed}
by the firing of the redex.

\begin{toappendix}
\begin{proposition}[Quantitative subject reduction for $\spi$ and $\ske$]
  \label{prop:subject-reduction}
  Let $\system \in \set{\spi, \ske}$. If  $\tderiv\exder[\system] \Deri[(\steps, \result)]\typctx{\tm}\type$ is \precise
  and $\tm\Rew{\system}\tmtwo$
  then $\steps\geq 2$ and
  there exists a \precise  typing $\tderivtwo$ such that
  $\tderivtwo\exder[\system] \Deri[(\steps-2, \result)]\typctx{\tmtwo}\type$.
\end{proposition}  
\end{toappendix}

The proof is by induction on $\tm\Rew{\system}\tmtwo$, and it relies
on the substitution lemma (\reflemma{typing-substitution-overview}) for the
base case of $\beta$-reduction at top level.

It is natural to wonder what happens when the typing is not tight. In the head case, the index $\steps$ still decreases exactly of 2. In the \lo case things are subtler---they are discussed in \refsect{shrinking}.

\paragraph{Summing up.} The tight correctness theorem is proved by a straightforward induction on the evaluation length relying on quantitative subject reduction (\refprop{subject-reduction}) for the inductive case, and the properties of tight typings for normal forms (\refprop{tight-normal-forms-indicies}) for the base case.

\begin{toappendix}
\begin{theorem}[Tight correctness for $\spi$ and $\ske$]
  \label{thm:correctness} 
  Let $\system \in \set{\spi, \ske}$ and $\tderiv \exder[\system]
    \Deri[(\steps, \result)] {\typctx}{\tm}{\type}$ be a tight
    derivation. Then there exists $\tmtwo$ such that $\tm
  \Rewn[\steps/2]{\system} \tmtwo$, $\sysnormal\system\tmtwo$, and
  $\syssize\system\tmtwo = \result$.  Moreover, if $\type=\neutype$ then $\sysneutral{\system}{\tmtwo}$.
\end{theorem}
\end{toappendix}


\subsection{Tight Completeness}
\label{s:tightcompleteness}
Completeness of tight typings (in system $\system \in \set{\spi, \ske}$)
expresses the fact that every $\system$-normalising term has a tight derivation (in system $\system$). As for correctness, the completeness theorem is always obtained via three intermediate steps, dual to those for correctness. Essentially, one shows that every normal form has a tight derivation and then extends the result to $\system$-normalising term by pulling typability back through evaluation using a subject expansion property.

\paragraph{First step: normal forms are tightly typable.} A simple induction on the structure of normal forms proves the following proposition.

\begin{toappendix}
\begin{proposition}[Normal forms are tightly typable for $\spi$ and $\ske$]
\label{prop:normal-forms-are-tightly-typable} 
Let $\system \in \set{\spi, \ske}$ and $\tm$ be such that $\sysnormal\system\tm$. Then 
there exists a \precise derivation $\tderiv\exder[\system] \Deri[(0, \syssize\system\tm)]\typctx{\tm}\type$. Moreover, if $\sysneutral\system\tm$ then $\type = \neutype$, and if $\sysabs\system\tm$ then $\type = \abstype$.
\end{proposition}
\end{toappendix}

In contrast to the proposition for normal forms of the correctness part (\refprop{tight-normal-forms-indicies}), here there are no auxiliary lemmas, so the property is simpler.

\paragraph{Second step: anti-substitution lemma.} In order to pull typability back along evaluation sequence, we have to first show that typability can also be pulled back along substitutions.
\begin{toappendix}
\begin{lemma}[Anti-substitution and typings for $\spi$ and $\ske$]
  \label{l:anti-substitution}
  Let $\system \in \set{\spi, \ske}$ and $\tderiv \exder[\system] \Deri[(\steps, \spine)]\typctx {\tm\isub\var\tmtwo } \type$. Then there exist:
  \begin{itemize}
  \item a multi-set ${\M}$;
  \item a typing derivation
    $\tderiv_\tm \exder[\system] \Deri[(\steps_\tm, \spine_\tm)]
    {\typctx_\tm;  \var \col {\M}}{\tm}  \type$;
    and
  \item a typing derivation
    $\tderiv_\tmtwo \exder[\system] \Deri[(\steps_\tmtwo, \spine_\tmtwo)]{\typctx_\tmtwo}{\tmtwo}{\M}$
  \end{itemize}
  such that:
  \begin{itemize}
  \item \emph{Typing context}: $\typctx = \typctx_\tm \mplus \typctx_\tmtwo$;

  \item \emph{Indices}: $(\steps, \spine) = (\steps_\tm + \steps_\tmtwo, \spine_\tm + \spine_\tmtwo)$.
  \end{itemize}
  Moreover, if $\tderiv$ is \precise\ then so are
  $\tderiv_\tm$ and $\tderiv_\tmtwo$.
\end{lemma}
\end{toappendix}

The proof is by induction on $\tderiv$. 

Let us point out that the anti-substitution lemma holds also in the
degenerated case in which $\var$ does not occur in $\tm$ and $\tmtwo$
is not $\system$-normalising: rule $\many$ can indeed be used to type
\emph{any} term $\tmtwo$ with $\Deri[(0, 0)] {}{\tmtwo
}{{\emm}}$ by taking an empty set $I$ of indices for the
premises. Note also that this is \emph{forced} by the fact that $\var
\notin \fv \tm$, and so $\typctx_\tm(\var) = \emm$. Finally,
this fact does not contradict the correctness theorem, because here
$\tmtwo$ is typed with a multi-set, while the theorem requires a type.

\paragraph{Third step: quantitative subject expansion.}
This property guarantees that typability can be pulled back along evaluation sequences.
\begin{toappendix}
\begin{proposition}[Quantitative subject expansion for $\spi$ and $\ske$]
  \label{prop:subject-expansion}
  Let $\system \in \set{\spi, \ske}$ and $\tderiv \exder[\system] \Deri[(\steps, \spine)]\typctx\tmtwo\type$ be a tight derivation. If $\tm \tostrat \tmtwo$
  then there exists a (tight) typing $\tderivtwo$ such that
  $\tderivtwo \exder[\system] \Deri[(\steps+2, \spine)] \typctx \tm \type$.
\end{proposition}
\end{toappendix}
The proof is a simple induction over $\tm \tostrat \tmtwo$ using the anti-substitution lemma in the base case of evaluation at top level.

\paragraph{Summing up.} The tight completeness theorem is proved by a straightforward induction on the evaluation length relying on quantitative subject expansion (\refprop{subject-expansion}) for the inductive case, and the existence of tight typings for normal forms (\refprop{normal-forms-are-tightly-typable}) for the base case.
\begin{toappendix}
  \begin{theorem}[Tight completeness for $\spi$ and $\ske$]
    \label{thm:completeness}
    Let $\system \in \set{\spi, \ske}$ and $\tm \Rewn[k]{\system} \tmtwo$ with $\sysnormal\system\tmtwo$.
    Then there exists a \precise typing
    $\tderiv \exder[\system] \Deri[(2k, \syssize\system\tmtwo)] \typctx \tm \type$.
    Moreover, if $\sysneutral\system\tmtwo$ then $\type = \neutype$, and if $\sysabs\system\tmtwo$ then $\type = \abstype$.  
  \end{theorem}
\end{toappendix}

\section{Extensions and Deeper Analyses}
In the rest of the paper we are going to further explore the properties of the tight approach to multi types along three independent axes:

\begin{enumerate}
\item \emph{Maximal evaluation}:
    we adapt the methodology to the case of maximal evaluation,
    which relates to strong normalisation in that
    the maximal evaluation strategy terminates only
    if the term being evaluated is strongly normalising.
    This case is a simplification of~\cite{bernadet13}
    that can be directly related  to the head and leftmost evaluation cases.
    It is in fact very close to leftmost evaluation
    but for the fact that, during evaluation,
    typing contexts are not necessarily preserved
    and the size of the terms being erased has to be taken into account.
    The statements of the properties in Sections~\ref{s:tightcorrectness}
    and~\ref{s:tightcompleteness} have to be adapted accordingly.

\item \emph{Linear head evaluation}: we
  reconsider head evaluation in the linear substitution calculus
  obtaining exact bounds on the number of steps and on the size of
  normal forms. The surprise here is that the type system is
  essentially unchanged and that it is enough to count also axiom
  rules (that are ignored for head evaluation in the $\l$-calculus)
  in order to exactly bound also the number of \emph{linear substitution} steps.

\item \emph{LO evaluation and minimal typings}: we explore the
  relationship between tight typings and traditional typings without
  tight constants. This study is done in the context of LO
  evaluation, that is the more relevant one with respect to cost
  models for the $\l$-calculus. We show in particular that tight
  typings are isomorphic to minimal traditional typings.
  
\end{enumerate}
Let us stress that these three variations on a theme can be read independently.

\section{Maximal Evaluation}

\label{sect:maximal}
In this section we consider the maximal strategy, which gives the
longest evaluation sequence from any strongly normalising term to its normal form.  The
maximal evaluation strategy is \emph{perpetual} in that, if a term
$\tm$ has a diverging evaluation path then the maximal strategy
diverges on $\tm$.  Therefore, its termination subsumes the
termination of any other strategy, which is why it is often used to
reason about the strong normalisation
property~\cite{DBLP:journals/iandc/RaamsdonkSSX99}.

\paragraph{Strong normalisation and erasing steps} 

It is well-known that in the framework of relevant (\ie\ without weakening) multi  types it is technically harder to deal with strong
normalisation (all evaluations terminate)---which 
is equivalent to the termination of the maximal strategy---
than with weak normalisation (there is a terminating evaluation)---which
is equivalent to the termination of the LO strategy. The reason is
that one has to ensure that all
subterms that are erased along any evaluation are themselves strongly
normalising.

The simple proof technique that we used in the previous section does
not scale up---in general---to strong normalisation (or to the maximal
strategy), because subject reduction breaks for erasing steps, as they
change the final typing judgement. Of course the same is true for
subject expansion. There are at least three ways of circumventing
this problem:
\begin{enumerate}
  \item \emph{Memory}: to add a memory constructor, as in Klop's calculus~\cite{phdklop}, that records the erased terms and allows evaluation inside the memory, so that diverging subterms are preserved. Subject reduction then is recovered.
  
  \item \emph{Subsumption/weakening}: adding a simple form of sub-typing, that allows stabilising  the final typing judgement in the case of an erasing step, or more generally, adding a strong form of weakening, 
that essentially removes the empty multi type.
  
  \item \emph{Big-step subject reduction}: abandon the preservation of the typing judgement
in the erasing cases, and rely on a more involved big-step subject reduction property relating the term directly to its normal form, stating in particular that the normal form is typable, potentially by a different type.
\end{enumerate}

Surprisingly, the tight characterisation of the maximal strategy that
we are going to develop does not need any of these workarounds: in the
case of tight typings subject reduction for the maximal strategy
holds, and the simple proof technique used before adapts smoothly. To
be precise, an evaluation step may still change the final typing
judgement, but the key point is that the judgement stays
tight. Morally, we are employing a form of subsumption of \precise\  contexts, but an
extremely light one, that in particular does not require a sub-typing
relation. We believe that this is a remarkable feature of tight multi
types.

\paragraph{Maximal evaluation and predicates} 
The maximal strategy shares with LO evaluation the predicates $\skneutralpr$,
$\sknormalpr$, $\skabspr$, and the notion of term size $\sksize\tm$, which we
respectively write $\sysneutralpr\systemax$, $\sysnormalpr\systemax$,
$\sysabspr\systemax$, and $\mxsize\tm$. We actually define, in
\reffig{maxevaluation}, a version of the maximal strategy, denoted
$\Rew{\perpe}[r]$, that is indexed by an integer $r$ representing the size of
what is erased by the evaluation step.  We define the transitive closure of
$\Rew{\perpe}[r]$ as follows:
\[
\infer{\tm\Rewn[0]{\perpe}[0]\tm}{\strut}
\qquad
\infer{\tm\Rewn[k+1]{\perpe}[r_1+r_2]\tmthree}{
  \tm\Rew{\perpe}[r_1]\tmtwo\quad\tmtwo\Rewn[k]{\perpe}[r_2]\tmthree
}
\qquad
\infer{\tm\Rewn{\perpe}[r]\tmtwo}{\tm\Rewn[k]{\perpe}[r]\tmtwo}
\]

\begin{figure*}
  \centering
  \ovalbox{
    $\begin{array}{c}
      \infer
            {
    (\la\var \tmthree) \tmfour \Rew{\perpe}[0] \tmthree \isub \var \tmfour
    }{\var\in \fv\tmthree }
    \qquad 
    \infer
        {
      (\la\var \tmthree) \tmfour \Rew{\perpe}[\mxsize\tmfour] \tmthree
    }{
      {\sysnormalpr \systemax {( \tmfour )} } \msep
      \var\notin \fv\tmthree
    }
    \qquad 
    \infer
          {
      (\la\var \tmthree) \tm \Rew{\perpe}[r] (\la\var \tmthree) \tmtwo
    }{
      \tm \Rew{\perpe}[r]\tmtwo\msep \var\notin \fv\tmthree
    }
    \\ \\ 
    
    \infer
          {\la\var \tm \Rew{\perpe}[r] \la\var \tmtwo}{\tm \Rew{\perpe}[r] \tmtwo}
    \qquad 
    \infer
          {
      \tm \tmthree \Rew{\perpe}[r] \tmtwo \tmthree
    }{
     { \sysnotabs \systemax \tm } \msep \tm \Rew{\perpe}[r]\tmtwo 
    }
    \qquad 
    \infer
          {
      \tmthree \tm \Rew{\perpe}[r]\tmthree \tmtwo
    }{
       { \sysneutral\systemax \tmthree } \msep   \tm \Rew{\perpe}[r]\tmtwo
    }
    \end{array}$
  }
  \caption{Deterministic maximal strategy}
  \label{fig:maxevaluation}
\end{figure*}

\begin{figure}
  \centering
  \ovalbox{
    $\begin{array}{@{}c@{}}

      \mbox{ Typing rules }
      \{ \ax, \funsteps, \funresult, \appsteps, \appresult<\ske>\}
      \mbox{ plus }\\\\

      \infer[\manystrict]{
        \Deri[(+_{\iI} \steps_i, +_{\iI} \result_i)] 
             { +_{\iI}\typctxtwo_i  } \tm {\MSigma {\type_i} {\iI}}
      }{
        (\Deri[(\steps_i, \result_i)]
        {\typctxtwo_i} \tm {\type_i})_{\iI}\quad \size{I}>0
      }

      \qquad
      
      \infer[\none]{
        \Deri[(\steps, \result)] 
             { \typctxtwo } \tm {\emm}
      }{
        \Deri[(\steps, \result)]
             {\typctxtwo} \tm {\type}
      }
    \end{array}$  }
  \caption{Type system for maximal evaluation}
  \label{fig:max-type-system}
\end{figure}

\begin{proposition}[$\perpe$ evaluation system]
\label{prop:memory-bievaluation-max} 
  $(\lterms,\Rew{\perpe},  \sysneutralpr \systemax, \sysnormalpr \systemax,
  \sysabspr\systemax)$ 
is an evaluation system.
\end{proposition}

\longshort{Also in this case the proof is routine, and it is then omitted even from the Appendix.}{Also in this case the proof is routine.}

\paragraph{Multi types}

Multi types are defined exactly as in~\refsection{typing}.
The type system $\systemax$ for $\perpe$-evaluation is defined
in~\reffig{max-type-system}.  Rules $\manystrict$ and $\none$, which
is a special 0-ary version of $\many$, are used to
prevent an argument $\tmtwo$ in rule $\appsteps$ to be {\it untyped}:
either it is typed by means of rule $\manystrict$---and thus it is
typed with at least one type---or it is typed by means of rule
$\none$---and thus it is typed with exactly one type: the type itself
is then forgotten, but requiring the premise to have a type forces the
term to be normalising. The fact that arguments are always typed, even
those that are erased during reduction, is essential to guarantee
strong normalisation: system $\systemax$ cannot type anymore a term
like $\var \Omega$.  Note that if
$\tderiv \exder[\systemax] \Deri[(\steps, \result)]{\typctx}{\tm}{\type}$,
then $\var\in\fv\tm$ if and only if $\typctx(\var)\neq\emm$.

Similarly to the head and leftmost-outermost cases,
we define the \emph{size} $\syssizeSN\tderiv$ of a typing derivation $\tderiv$
as the number of rule applications in $\tderiv$,
not counting rules $\ax$ and $\manystrict$ and $\none$.
And again if $\tderiv \exder[\systemax] \Deri[(\steps, \result)]{\typctx}{\tm}{\type}$
then $\steps + \result = \syssizeSN\tderiv$.

For maximal evaluation, we need also to refine the notion of
\preciseness of typing derivations,
which becomes a global condition because it is no longer a property of
the final  judgment only: 
\begin{definition}[\Maxprecise derivations]
  A derivation
  $\tderiv \exder[\systemax] \Deri[(\steps, \result)]{\typctx}{\tm}{\type}$
  is \emph{\intprecise}
  if in every instance of rule $(\none)$ in $\tderiv$ we have $\tightpred\type$.
  It is \emph{\maxprecise}
  if also 
  $\tderiv$ is \precise, in the sense of Definition~\ref{def:tightderiv}.
\end{definition}

Similarly to the head and LO cases,
the quantitative information in \maxprecise derivations
characterises evaluation lengths and sizes of normal forms,
as captured by the correctness and completeness theorems.

\subsection{Tight Correctness}
The correctness theorem is proved following the same schema used for
head and LO evaluations. Most proofs are similar,
and are therefore omitted\longshort{ even from the Appendix}{}.

We start with the properties of typed normal forms. As before, we need
an auxiliary lemma about neutral terms, analogous
to~\refprop{tight-normal-forms-indicies}.
\begin{lemma}[Tight spreading on neutral terms for $\systemax$]
\label{l:mxtight-spreading}
  If $\spneutral \tm$ and $\tderiv\exder[\systemax] \Deri[(\steps, \result)]\typctx{\tm}\type$ such that $\tightpred \typctx$,
  then $\tightpred{\type}$.
\end{lemma}

The general properties of typed normal forms hold as well.

\begin{proposition}[Properties of \maxprecise typings for normal forms]
  \label{prop:mxtight-normal-forms-indices}
  Given $\tderiv\exder[\systemax] \Deri[(\steps, \result)]\typctx{\tm}\type$
  with ${\sysnormal\systemax\tm}$,
  \begin{enumerate}
  \item \label{p:mxtight-normal-forms-indices-size-bound}
    \emph{Size bound}: $\syssize\systemax\tm \leq \steps+\result$.

  \item \label{p:mxtight-normal-forms-indices-tightness}
    \emph{Tightness}: if $\tderiv$ is \maxprecise\  then $\steps = 0$ and $\result = \syssize\systemax\tm$. 

  \item \label{p:mxtight-normal-forms-indices-neutrality}
    \emph{Neutrality}: if $\type=\neutype$ then $\sysneutral{\systemax}{\tm}$.
  \end{enumerate}
\end{proposition}

Then we can type substitutions:
\begin{lemma}[Substitution and typings for $\systemax$]
  \label{l:mxtyping-substitution}
  The following rule is admissible in system $\systemax$: 
  \[
   \infer[\subslemma]{
     \Deri[(\steps + \steps', \result + \result')]
          {\typctx\mplus\typctxtwo}{\tm\isub\var\tmtwo }{\type}
   }{
     \Deri[(\steps, \result)]{\typctx} {\tmtwo}{{\M}}
     \quad
     \Deri[(\steps', \result')] {\typctxtwo; \var: {\M}} {\tm}{\type}
     \quad
     \M \neq \emm
   }
  \]
  Moreover if the derivations of the premisses are \intprecise,
  then so is the derivation of the conclusion.
\end{lemma}
Note that, in contrast to~\reflemma{typing-substitution-overview}
in~\refsection{tightcorrectness},
we assume that the multi-set $\M$ is not empty,
so that the left premiss is derived with rule $\manystrict$
rather than $\none$.

\paragraph{Subject reduction} The statement here slightly
differs from the corresponding one in~\refsection{tightcorrectness}.
Indeed, the typing environment $\typctx$ for term $\tm$ is
not necessarily preserved when typing $\tmtwo$,
because the evaluation step may erase a subterm.
Consider for instance term $\tm =
(\la\var \var')(\vartwo\vartwo)$.  In
any $\systemax$-typing derivation of $\tm$, the typing context
must declare $\vartwo$ with an appropriate
type that ensures that, when applying a well-typed substitution to
$\tm$, the resulting term is still normalising for
$\Rew{\perpe}$. For instance, the context should declare
  $\vartwo:\mult{\ty{\mult{\type}}{\type}, \type }$,
  or even $\vartwo: \mult{\neutype}$ if
  the typing derivation for $\tm$ is \maxprecise.
However, as
$\tm\Rew{\perpe}[1]\var'$, the typing derivation for $\var'$ will
clearly have a typing environment $\typctx'$ that maps
$\vartwo$ to $\emm$. Hence, the subject reduction property has to take into account the change of typing context, as shown below.

\begin{toappendix}
  \begin{proposition}[Quantitative subject reduction for $\systemax$]
    \label{prop:mxsubject-reduction}
    If $\tderiv\exder[\systemax] \Deri[(\steps, \result)]\typctx{\tm}\type$ is \maxprecise
    and $\tm\Rew{\systemax}[e]\tmtwo$,
    then there exist $\typctx'$ and an \maxprecise  typing $\tderivtwo$ such that
    $\tderivtwo\exder[\systemax] \Deri[(\steps-2, \result-e)]{\typctx'}{\tmtwo}\type$.
  \end{proposition}
\end{toappendix}

\begin{toappendix}[
    \begin{proof}See \longshort{Appendix~\thisappendix}{the long version of this paper~\cite{AccattoliGLK2018-long}}.\end{proof}
  ]

\begin{proof}
  We prove, by induction on $\tm\Rew{\systemax}[e]\tmtwo$,
  the stronger statement:
  
  Assume
  $\tm\Rew{\perpe}[e]\tmtwo$,
  $\tderiv\exder[\systemax] \Deri[(\steps, \result)]\typctx{\tm}\type$ is \intprecise,
  $\tightpred\typctx$,
  and either $\tightpred\type$ or $\sysnotabs \systemax \tm$.
  
  Then there exist $\typctx'$ and a \intprecise typing
  $\tderivtwo\exder[\systemax] \Deri[(\steps-2, \result-e)]{\typctx'}{\tmtwo}\type$
  such that $\tightpred{\typctx'}$.

  \begin{itemize}
  \item Rule 
    \[\infer{
      (\la\var \tmthree) \tmfour \toperpNE \tmthree \isub \var \tmfour
    }{\var\in\fv\tmthree}
    \]
    Assume 
    $\tderiv \exder[\systemax]\Deri[(\steps,\result)]
    {\typctx}{(\la\var \tmthree) \tmfour}\type$ is \intprecise
    and $\tightpred\typctx$.
    The derivation $\tderiv$ must end with rule $\appsteps$,
    and the derivation of its premiss for $(\la\var \tmthree)$
    must end with $\funsteps$.
    Hence, there are two \intprecise derivations
    $\tderiv_\tmthree\exder[\systemax]\Deri[
      (\steps_\tmthree, \result_\tmthree)
    ]{\typctx_\tmthree; \var: {\M}}{\tmthree}{\type}$
    and
    $\tderiv_\tmtwo\exder[\systemax]\Deri[
      (\steps_\tmtwo, \result_\tmtwo)
    ]{\typctx_\tmtwo}{\tmtwo}{\M}$,
    with $(\steps,\result) = (\steps_u+\steps_\tmfour+2, \result_u+\result_\tmfour)$
    and $\typctx=\typctx_\tmthree\mplus\typctx_\tmtwo$.
    Moreover, $\M\neq\emm$ as $\var\in\fv\tmthree$,
    and therefore applying Lemma~\ref{l:mxtyping-substitution}
    gives a \intprecise $\tderiv'\exder[\systemax]\Deri[(\steps_u+\steps_\tmfour,\result_u+\result_\tmfour)]
    {\typctx}{\tmthree \isub \var \tmfour}\type$.

  \item Rule 
    \[\infer{
      (\la\var \tmthree) \tmfour \Rew{\perpe}[\mxsize\tmfour] \tmthree
    }{
      \sysnormal\systemax  \tmfour\quad
      \var\notin \fv\tmthree
    }\]
    Assume 
    $\tderiv \exder[\systemax]\Deri[(\steps,\result)]
    {\typctx}{(\la\var \tmthree) \tmfour}\type$ is \intprecise
    and $\tightpred\typctx$.
    The derivation $\tderiv$ must end with rule $\appsteps$,
    and the derivation of its premiss for $(\la\var \tmthree)$
    must end with $\funsteps$.
    Since $\var\notin \fv\tmthree$,
    $\tderiv$ must be of the form
    \[
    \infer{
      \Deri[
        (\steps_\tmthree+\steps_\tmfour+2,\result_\tmthree+\result_\tmfour)
      ]{\typctx_\tmthree\mplus\typctx_\tmfour}{(\la\var \tmthree) \tmfour}\type
    }{
      \infer{
        \Deri[
          (\steps_\tmthree+1,\result_\tmthree)
        ]{\typctx_\tmthree}{\la\var\tmthree}{\emm\rightarrow\type}
      }{
        \tderiv_\tmthree \exder[\systemax]
        \Deri[
          (\steps_\tmthree,\result_\tmthree)
        ]{\typctx_\tmthree}{\tmthree}\type
      }
      \qquad
      \infer{
        \Deri[(\steps_\tmfour,\result_\tmfour)]{\typctx_\tmfour}{\tmfour}{\emm}
      }{
        \tderiv_\tmfour \exder[\systemax]\Deri[
          (\steps_\tmfour,\result_\tmfour)
        ]{\typctx_\tmfour}{\tmfour}{\type_\tmfour}
      }
    }
    \]
    with $(\steps,\result) = (\steps_u+\steps_\tmfour+2, \result_u+\result_\tmfour)$
    and $\typctx=\typctx_\tmthree\mplus\typctx_\tmtwo$.
    Since $\tderiv$ is \intprecise,
    $\type_\tmfour$ must be $\nf$,
    and since $\sysnormal\systemax  \tmfour$,
    we can apply Theorem~\ref{prop:mxtight-normal-forms-indices}
    and get $(\steps_\tmfour,\result_\tmfour)=(0,\mxsize\tmfour)$,
    so that $(\steps_\tmthree,\result_\tmthree)=(\steps-2,\result-\mxsize\tmfour)$,.
    Since $\tightpred{\typctx_\tmthree\mplus\typctx_\tmfour}$
    we have $\tightpred{\typctx_\tmthree}$,
    so $\tderiv_\tmthree$ is the desired \intprecise derivation.

  \item Rule 
    \[
    \infer{\la\var \tm \Rew{\perpe}[e] \la\var \tmtwo}{
      \tm \Rew{\perpe}[e] \tmtwo
    }
    \]
    Assume
    $\tderiv\exder[\systemax]\Deri[(\steps,\result)]
    {\typctx}{\la\var \tm}\type$ is \intprecise
    and $\tightpred\typctx$.
    Since $\sysabs \systemax {\la\var \tm}$ we must have hypothesis $\tightpred\type$,
    and as $\tderiv$ must then finish with rule $\funresult$
    we must have a subderivation
    $\tderiv_\tm \exder[\systemax]
    \Deri[
      (\steps, \result-1)
    ]{\typctx, \var \col \mtight}{\tm}{\nf}$.
    As $\tderiv_\tm$ is \intprecise and $\tightpred{\typctx, \var \col \mtight}$
    we can apply the \ih and get the premiss of the derivation $\tderiv'$ below,
    where $\tderiv_\tmtwo$ is \intprecise and $\tightpred{\typctx', \var \col \mtight}$:
    \[
    \infer{
      \Deri[(\steps+2, \result-e)]{
        \typctx'
      }{\la\var\tmtwo}\type
    }{
      \tderiv_\tmtwo \exder[\systemax]
      \Deri[(\steps+2, \result-1-e)]{\typctx', \var \col \mtight}{\tmtwo}\nf
    }
    \]
    Then $\tderiv'$ is \intprecise and $\tightpred{\typctx'}$.

  \item Rule 
    \[
    \infer{
      \tm \tmthree \Rew{\perpe}[e] \tmtwo \tmthree
    }{
      \sysnotabs \systemax \tm \quad \tm \Rew{\perpe}[e]\tmtwo 
    }
    \]
    Assume
    $\tderiv\exder[\systemax]\Deri[(\steps,\result)]
    {\typctx}{\tm \tmthree}\type$ is \intprecise
    and $\tightpred\typctx$.
    The derivation $\tderiv$ must end
    with rule $\appsteps$ or $\appresult<\ske>$,
    and therefore there are two \intprecise derivations
    $\tderiv_\tm\exder[\systemax]\Deri[
      (\steps_\tm, \result_\tm)
    ]{\typctx_\tm}{\tm}{\type_\tm}$
    and
    $\tderiv_\tmthree\exder[\systemax]\Deri[
      (\steps_\tmthree, \result_\tmthree)
    ]{\typctx_\tmthree}{\tmthree}{\type_\tmthree}$,
    for some types $\type_\tm$ and $\type_\tmthree$,
    with $\typctx=\typctx_\tm\mplus\typctx_\tmthree$.
    Since $\tightpred{\typctx}$ we have
    $\tightpred{\typctx_\tm}$ and $\tightpred{\typctx_\tmthree}$.
    Since $\sysnotabs \systemax \tm$, we can apply the \ih and get
    the \intprecise derivation
    $\tderiv_\tmtwo \exder[\systemax]
    \Deri[
      (\steps_\tm-2, \result_\tm-e)
    ]{\typctx_\tmtwo}{\tmtwo}{\type_\tm}$,
    with $\tightpred{\typctx_\tmtwo}$.
    Then the same rule $\appsteps$ or $\appresult<\ske>$
    can be applied to get the \intprecise derivation
    $\tderiv'\exder[\systemax]\Deri[(\steps-2,\result-e)]
    {\typctx_\tmtwo\mplus\typctx_\tmthree}{\tmtwo \tmthree}\type$,
    with $\tightpred{\typctx_\tmtwo\mplus\typctx_\tmthree}$.
    
  \item Rule 
    \[
    \infer{
      \tmthree \tm \Rew{\perpe}[e]\tmthree \tmtwo
    }{
      \sysneutral\systemax \tmthree \quad  \tm \Rew{\perpe}[e]\tmtwo
    }
    \]
    Assume
    $\tderiv\exder[\systemax]\Deri[(\steps,\result)]
    {\typctx}{\tmthree \tm}\type$ is \intprecise
    and $\tightpred\typctx$.
    The derivation $\tderiv$ must end
    with rule $\appsteps$ or $\appresult<\ske>$,
    and therefore there are two \intprecise derivations
    $\tderiv_\tmthree\exder[\systemax]\Deri[
      (\steps_\tmthree, \result_\tmthree)
    ]{\typctx_\tmthree}{\tmthree}{\type_\tmthree}$
    and
    $\tderiv_\tm\exder[\systemax]\Deri[
      (\steps_\tm, \result_\tm)
    ]{\typctx_\tm}{\tm}{\type_\tm}$,
    for some types $\type_\tmthree$ and $\type_\tm$,
    with $\typctx=\typctx_\tmthree\mplus\typctx_\tm$.
    Since $\tightpred{\typctx}$ we have
    $\tightpred{\typctx_\tmthree}$ and $\tightpred{\typctx_\tm}$.
    Theorem~\ref{l:mxtight-spreading}
    concludes $\tightpred{\type_\tmthree}$ from $\sysneutral\systemax \tmthree$.
    So the last rule of $\tderiv$ must be $\appresult<\ske>$,
    whence $\type=\neutype$ and $\type_\tm=\tight$.
    Therefore we can apply the \ih to get the \precise derivation
    $\tderiv_\tmtwo \exder[\systemax]
    \Deri[
      (\steps_\tm-2, \result_\tm-e)
    ]{\typctx_\tmtwo}{\tmtwo}{\nf}$.
    Then $\appresult<\ske>$
    can be applied to get the \precise derivation
    $\tderiv'\exder[\systemax]\Deri[(\steps-2,\result-e)]
    {\typctx_\tmthree\mplus\typctx_\tmtwo}{\tmthree\tmtwo}\neutype$.

  \item Rule 
    \[\infer{
      (\la\var \tmthree) \tm \Rew{\perpe}[e] (\la\var \tmthree) \tmtwo
    }{
      \tm \Rew{\perpe}[e]\tmtwo\quad\var\notin \fv\tmthree
    }\]
    Assume
    $\tderiv\exder[\systemax]\Deri[(\steps,\result)]
    {\typctx}{(\la\var \tmthree) \tm}\type$ is \intprecise
    and $\tightpred\typctx$.
    The derivation $\tderiv$ must end
    with rule $\appsteps$,
    and therefore there are two \intprecise derivations
    $\tderiv_\tmthree\exder[\systemax]\Deri[
      (\steps_\tmthree, \result_\tmthree)
    ]{\typctx_\tmthree}{\la\var\tmthree}{\emm\rightarrow\type}$
    and
    $\tderiv_\tm\exder[\systemax]\Deri[
      (\steps_\tm, \result_\tm)
    ]{\typctx_\tm}{\tm}{\nf}$,
    with $\typctx=\typctx_\tmthree\mplus\typctx_\tm$
    and $(\steps,\result)=(\steps_\tm+\steps_\tmthree+1,\result_\tm+\result_\tmthree)$.
    We can apply the \ih to get
    the \precise derivation
    $\tderiv_\tmtwo \exder[\systemax]
    \Deri[
      (\steps_\tm-2, \result_\tm-e)
    ]{\typctx_\tmtwo}{\tmtwo}{\nf}$.
    Then $\appsteps$
    can be applied to get the \intprecise derivation
    $\tderiv'\exder[\systemax]\Deri[(\steps-2,\result-e)]
    {\typctx_\tmthree\mplus\typctx_\tmtwo}{(\la\var \tmthree)\tmtwo}\type$,
    with $\tightpred{\typctx_\tmthree\mplus\typctx_\tmtwo}$.
  \end{itemize}
\end{proof}

\end{toappendix}

\paragraph{Correctness theorem} Now the correctness theorem easily follows. It differs from the corresponding theorem in~\refsection{tightcorrectness}
in that the second index in the \maxprecise typing judgement does not
only measure the size of the normal form but also the sizes of all the
terms erased during evaluation (and necessarily in normal form).

\begin{theorem}[Tight correctness for $\perpe$-evaluation]\label{th:mxcorrect}
Let $\tderiv \exder[\systemax]  \Deri[(\steps, \result)] {\typctx}{\tm}{\type}$
be a \maxprecise\ derivation.
Then there is an integer $e$ and  a  term  $\tmtwo$ such that
$\sysnormal\systemax\tmtwo$, $\tm\Rewn[\steps/2]{\perpe}[e] \tmtwo$
and $\mxsize\tmtwo  + e = \result$. Moreover, if $\type=\neutype$
  then $\sysneutral\systemax\tmtwo$.
\end{theorem}

\subsection{Tight Completeness}
Completeness is again similar to that in~\refsection{tightcompleteness},
and differs from it in the same way as correctness differs
from that in~\refsection{tightcorrectness}. Namely, the second index in the completeness theorem also accounts for the size of erased terms, and the \longshort{appendix}{long version of this paper~\cite{AccattoliGLK2018-long}} provides the proof of the subject expansion property. The completeness statement follows.

\begin{proposition}[Normal forms are tightly typable in \systemax]
  \label{prop:perp-normal-forms}
  Let $\tm$ be such that $\sysnormal\systemax\tm$. Then 
  there exists a \maxprecise derivation $\tderiv\exder[\systemax] \Deri[(0, \mxsize\tm)]\typctx{\tm}\type$.
  Moreover, if $\sysneutral\systemax\tm$ then $\type = \neutype$, and if $\sysabs\systemax\tm$ then $\type = \abstype$.
\end{proposition}

\begin{lemma}[Anti-substitution and typings for $\systemax$]
  \label{l:maxanti-substitution}
  If $\tderiv \exder[\systemax] \Deri[(\steps, \result)] \typctx {\tm\isub\var\tmtwo } \type$
  and $x\in\fv t$,
  then there exist:
  \begin{itemize}
  \item a multiset ${\M}$ different from $\emm$;
  \item a typing derivation
    $\tderiv_\tm \exder[\systemax] \Deri[(\steps_\tm, \result_\tm)]
    {\typctx_\tm;  \var \col {\M}}{\tm}  \type$;
    and
  \item a typing derivation
    $\tderiv_\tmtwo \exder[\systemax] \Deri[(\steps_\tmtwo, \result_\tmtwo)]{\typctx_\tmtwo}{\tmtwo}{\M}$
  \end{itemize}
  such that:
  \begin{itemize}
  \item \emph{Typing context}: $\typctx = \typctx_\tm \mplus \typctx_\tmtwo$;

  \item \emph{Indices}: $(\steps, \result) = (\steps_\tm + \steps_\tmtwo, \result_\tm + \result_\tmtwo)$.
  \end{itemize}
  Moreover, if $\tderiv$ is \intprecise\ then so are $\tderiv_\tm$ and $\tderiv_\tmtwo$.
\end{lemma}


\begin{toappendix}
  \begin{proposition}[Quantitative subject expansion for $\systemax$]
    \label{prop:mxsubject-expansion}
    If $\tderiv\exder[\systemax] \Deri[(\steps, \result)]\typctx{\tmtwo}\type$ is \maxprecise
    and $\tm\Rew{\systemax}[e]\tmtwo$,
    then there exist $\typctx'$ and an \maxprecise typing $\tderivtwo$ such that
    $\tderivtwo\exder[\systemax] \Deri[(\steps+2, \result+e)]{\typctx'}{\tm}\type$.
  \end{proposition}
\end{toappendix}

\begin{toappendix}[
    \begin{proof}See \longshort{Appendix~\thisappendix}{the long version of this paper~\cite{AccattoliGLK2018-long}}.\end{proof}
  ]
\begin{proof}
  We prove, by induction on $\tm\Rew{\systemax}[e]\tmtwo$,
  the stronger statement:
  
  Assume
  $\tm\Rew{\perpe}[e]\tmtwo$,
  $\tderiv\exder[\systemax] \Deri[(\steps, \result)]\typctx{\tmtwo}\type$ is \intprecise,
  $\tightpred\typctx$,
  and either $\tightpred\type$ or $\sysnotabs \systemax \tm$.
  
  Then there exist $\typctx'$ and a \intprecise typing
  $\tderivtwo\exder[\systemax] \Deri[(\steps+2, \result+e)]{\typctx'}{\tm}\type$
  such that $\tightpred{\typctx'}$.

  \begin{itemize}
  \item Rule 
    \[\infer{
      (\la\var \tmthree) \tmfour \toperpNE \tmthree \isub \var \tmfour
    }{\var\in\fv\tmthree}
    \]
    Assume 
    $\tderiv \exder[\systemax]\Deri[(\steps,\result)]
    {\typctx}{\tmthree \isub \var \tmfour}\type$ is \intprecise
    and $\tightpred\typctx$.
    By applying Lemma~\ref{l:maxanti-substitution}
    we get the premisses of the following derivation $\tderiv'$:
    \[
    \infer{
      \Deri[
        (\steps_u+\steps_\tmfour+2,\result_u+\result_\tmfour)
      ]{\typctx_u\mplus \typctx_\tmfour}{(\la\var \tmthree) \tmfour}\type
    }{
      \infer{
        \Deri[
          (\steps_u+1,\result_u)
        ]{\typctx_u}{\la\var\tmthree}{\M\rightarrow\type}
      }{
        \tderiv_\tmthree \exder[\systemax]
        \Deri[
          (\steps_u,\result_u)
        ]{\typctx_u, \var \col \M}{\tmthree}\type
      }
      \qquad
      \tderiv_\tmfour \exder[\systemax]
      \Deri[
        (\steps_\tmfour,\result_\tmfour)
      ]{\typctx_\tmfour}{\tmfour}{\M}
    }
    \]
    with $(\steps,\result) = (\steps_u+\steps_\tmfour, \result_u+\result_\tmfour)$
    and $\typctx=\typctx_u\mplus\typctx_\tmfour$.
    Moreover, $\tderiv_u$ and $\tderiv_\tmfour$ are all \intprecise,
    so $\tderiv'$ is \intprecise.

  \item Rule 
    \[\infer{
      (\la\var \tmthree) \tmfour \Rew{\perpe}[\mxsize\tmfour] \tmthree
    }{
      \sysnormal\systemax  \tmfour\quad
      \var\notin \fv\tmthree
    }\]

    Assume
    $\tderiv \exder[\systemax]\Deri[(\steps,\result)]
    {\typctx}{\tmthree}\type$ is \intprecise
    and $\tightpred\typctx$.
    By applying \refprop{perp-normal-forms}
    we get the \precise derivation $\tderiv_\tmfour$
    used in the construction of derivation $\tderiv'$ below:
    \[
    \infer{
      \Deri[
        (\steps+2,\result+\mxsize\tmfour)
      ]{\typctx\mplus\typctx_\tmfour}{(\la\var \tmthree) \tmfour}\type
    }{
      \infer{
        \Deri[
          (\steps+1,\result)
        ]{\typctx}{\la\var\tmthree}{\emm\rightarrow\type}
      }{
        \tderiv \exder[\systemax]
        \Deri[
          (\steps,\result)
        ]{\typctx}{\tmthree}\type
      }
      \qquad
      \infer{
        \Deri[(0,\mxsize\tmfour)]{\typctx_\tmfour}{\tmfour}{\emm}
      }{
        \tderiv_\tmfour \exder[\systemax]\Deri[(0,\mxsize\tmfour)]{\typctx_\tmfour}{\tmfour}{\nf}
      }
    }
    \]
    Moreover, $\tightpred{\typctx\mplus\typctx_\tmfour}$ and $\tderiv'$ is \intprecise.

  \item Rule 
    \[
    \infer{\la\var \tm \Rew{\perpe}[e] \la\var \tmtwo}{
      \tm \Rew{\perpe}[e] \tmtwo
    }
    \]
    Assume
    $\tderiv\exder[\systemax]\Deri[(\steps,\result)]
    {\typctx}{\la\var \tmtwo}\type$ is \intprecise
    and $\tightpred\typctx$.
    Since $\sysabs \systemax {\la\var \tm}$ we must have hypothesis $\tightpred\type$,
    and as $\tderiv$ must then finish with rule $\funresult$
    we must have a subderivation
    $\tderiv_\tmtwo \exder[\systemax]
    \Deri[
      (\steps, \result-1)
    ]{\typctx, \var \col \mtight}{\tmtwo}{\nf}$.
    As $\tderiv_\tmtwo$ is \intprecise and $\tightpred{\typctx, \var \col \mtight}$
    we can apply the \ih and get the premiss of the derivation $\tderiv'$ below,
    where $\tderiv_\tm$ is \intprecise and $\tightpred{\typctx', \var \col \mtight}$:
    \[
    \infer{
      \Deri[(\steps+2, \result+e)]{
        \typctx'
      }{\la\var\tm}\type
    }{
      \tderiv_\tm \exder[\systemax]
      \Deri[(\steps+2, \result-1+e)]{\typctx', \var \col \mtight}{\tm}\nf
    }
    \]
    Then $\tderiv'$ is \intprecise and $\tightpred{\typctx'}$.

  \item Rule 
    \[
    \infer{
      \tm \tmthree \Rew{\perpe}[e] \tmtwo \tmthree
    }{
      \sysnotabs \systemax \tm \quad \tm \Rew{\perpe}[e]\tmtwo 
    }
    \]
    Assume
    $\tderiv\exder[\systemax]\Deri[(\steps,\result)]
    {\typctx}{\tmtwo \tmthree}\type$ is \intprecise
    and $\tightpred\typctx$.
    The derivation $\tderiv$ must end
    with rule $\appsteps$ or $\appresult<\ske>$,
    and therefore there are two \intprecise derivations
    $\tderiv_\tmtwo\exder[\systemax]\Deri[
      (\steps_\tmtwo, \result_\tmtwo)
    ]{\typctx_\tmtwo}{\tmtwo}{\type_\tmtwo}$
    and
    $\tderiv_\tmthree\exder[\systemax]\Deri[
      (\steps_\tmthree, \result_\tmthree)
    ]{\typctx_\tmthree}{\tmthree}{\type_\tmthree}$,
    for some types $\type_\tmtwo$ and $\type_\tmthree$,
    with $\typctx=\typctx_\tmtwo\mplus\typctx_\tmthree$.
    Since $\tightpred{\typctx}$ we have
    $\tightpred{\typctx_\tmtwo}$ and $\tightpred{\typctx_\tmthree}$.
    Since $\sysnotabs \systemax \tm$, we can apply the \ih and get
    the \intprecise derivation
    $\tderiv_\tm \exder[\systemax]
    \Deri[
      (\steps_\tmtwo+2, \result_\tmtwo+e)
    ]{\typctx_\tm}{\tm}{\type_\tmtwo}$,
    with $\tightpred{\typctx_\tm}$.
    Then the same rule $\appsteps$ or $\appresult<\ske>$
    can be applied to get the \intprecise derivation
    $\tderiv'\exder[\systemax]\Deri[(\steps+2,\result+e)]
    {\typctx_\tm\mplus\typctx_\tmthree}{\tm \tmthree}\type$,
    with $\tightpred{\typctx_\tm\mplus\typctx_\tmthree}$.
    
  \item Rule 
    \[
    \infer{
      \tmthree \tm \Rew{\perpe}[e]\tmthree \tmtwo
    }{
      \sysneutral\systemax \tmthree \quad  \tm \Rew{\perpe}[e]\tmtwo
    }
    \]
    Assume
    $\tderiv\exder[\systemax]\Deri[(\steps,\result)]
    {\typctx}{\tmthree \tmtwo}\type$ is \intprecise
    and $\tightpred\typctx$.
    The derivation $\tderiv$ must end
    with rule $\appsteps$ or $\appresult<\ske>$,
    and therefore there are two \intprecise derivations
    $\tderiv_\tmthree\exder[\systemax]\Deri[
      (\steps_\tmthree, \result_\tmthree)
    ]{\typctx_\tmthree}{\tmthree}{\type_\tmthree}$
    and
    $\tderiv_\tmtwo\exder[\systemax]\Deri[
      (\steps_\tmtwo, \result_\tmtwo)
    ]{\typctx_\tmtwo}{\tmtwo}{\type_\tmtwo}$,
    for some types $\type_\tmthree$ and $\type_\tmtwo$,
    with $\typctx=\typctx_\tmthree\mplus\typctx_\tmtwo$.
    Since $\tightpred{\typctx}$ we have
    $\tightpred{\typctx_\tmthree}$ and $\tightpred{\typctx_\tmtwo}$.
    Theorem~\ref{l:mxtight-spreading}
    concludes $\tightpred{\type_\tmthree}$ from $\sysneutral\systemax \tmthree$.
    So the last rule of $\tderiv$ must be $\appresult<\ske>$,
    whence $\type=\neutype$ and $\type_\tmtwo=\tight$.
    Therefore we can apply the \ih to get the \precise derivation
    $\tderiv_\tm \exder[\systemax]
    \Deri[
      (\steps_\tmtwo+2, \result_\tmtwo+e)
    ]{\typctx_\tm}{\tm}{\nf}$.
    Then $\appresult<\ske>$
    can be applied to get the \precise derivation
    $\tderiv'\exder[\systemax]\Deri[(\steps+2,\result+e)]
    {\typctx_\tm\mplus\typctx_\tmthree}{\tmthree\tm}\neutype$.

  \item Rule 
    \[\infer{
      (\la\var \tmthree) \tm \Rew{\perpe}[e] (\la\var \tmthree) \tmtwo
    }{
      \tm \Rew{\perpe}[e]\tmtwo\quad\var\notin \fv\tmthree
    }\]
    Assume
    $\tderiv\exder[\systemax]\Deri[(\steps,\result)]
    {\typctx}{(\la\var \tmthree) \tmtwo}\type$ is \intprecise
    and $\tightpred\typctx$.
    The derivation $\tderiv$ must end
    with rule $\appsteps$,
    and therefore there are two \intprecise derivations
    $\tderiv_\tmthree\exder[\systemax]\Deri[
      (\steps_\tmthree, \result_\tmthree)
    ]{\typctx_\tmthree}{\la\var\tmthree}{\emm\rightarrow\type}$
    and
    $\tderiv_\tmtwo\exder[\systemax]\Deri[
      (\steps_\tmtwo, \result_\tmtwo)
    ]{\typctx_\tmtwo}{\tmtwo}{\nf}$,
    with $\typctx=\typctx_\tmthree\mplus\typctx_\tmtwo$
    and $(\steps,\result)=(\steps_\tmtwo+\steps_\tmthree+1,\result_\tmtwo+\result_\tmthree)$.
    We can apply the \ih to get
    the \precise derivation
    $\tderiv_\tm \exder[\systemax]
    \Deri[
      (\steps_\tmtwo+2, \result_\tmtwo+e)
    ]{\typctx_\tm}{\tm}{\nf}$.
    Then $\appsteps$
    can be applied to get the \intprecise derivation
    $\tderiv'\exder[\systemax]\Deri[(\steps+2,\result+e)]
    {\typctx_\tm\mplus\typctx_\tmthree}{(\la\var \tmthree)\tm}\type$,
    with $\tightpred{\typctx_\tm\mplus\typctx_\tmthree}$.
  \end{itemize}
\end{proof}

\end{toappendix}

\begin{theorem}[Tight completeness for for $\systemax$]\label{th:mxcomplete}
  If $\tm \Rewn[k]{\perpe}[e] \tmtwo$ with $\sysnormal\systemax\tmtwo$, then
  there exists an \maxprecise typing
  $\tderiv \exder[\systemax] \Deri[(2k, \mxsize\tmtwo
  +e)]\typctx\tm\type$.  Moreover, if $\sysneutral\systemax\tmtwo$ then $\type = \neutype$, and if $\sysabs\systemax\tmtwo$ then $\type = \abstype$.
\end{theorem}



\section{Linear Head Evaluation}

\label{sect:linear}
In this section we consider the linear version of the head evaluation system, where \emph{linear} comes from the \emph{linear substitution calculus} (LSC) \cite{DBLP:conf/rta/Accattoli12,DBLP:conf/popl/AccattoliBKL14}, a refinement of the $\l$-calculus where the language is extended with an explicit substitution constructor $\tm\esub\var\tmtwo$, and \emph{linear substitution} is a micro-step rewriting rule replacing one occurrence at a time---therefore, \emph{linear} does not mean that variables have at most one occurrence, only that their occurrences are replaced one by one. Linear head evaluation---first studied in \cite{DBLP:journals/tcs/MascariP94,Danos04headlinear}---admits various presentations. The one in the LSC adopted here has been introduced in \cite{DBLP:conf/rta/Accattoli12} and is the simplest one.

The insight here is that switching from head to linear head, and from the $\l$-calculus to the LSC only requires counting $\ax$ rules for the size of typings and the head variable for the size of terms---the type system, in particular is the same. The correspondence between the two system is spelled out in the last subsection of this part. Of course, switching to the LSC some details have to be adapted: a further index traces linear substitution steps, there is a new typing rule to type the new explicit substitution constructor, and the proof schema slightly changes, as the (anti-)substitution lemma is replaced by a partial substitution one---these are unavoidable and yet inessential modifications.

Thus, the main point of this section is to
split the complexity measure among the multiplicative steps (beta
steps) and the exponential ones (substitutions). Moreover, linear logic proof-nets are known to simulate the $\lambda$-calculus,
and LSC is known to be isomorphic to the proof-nets used in the simulation. Therefore, the results of
this section directly apply to those proof-nets.

\paragraph{Explicit substitutions.} We start by introducing the syntax of our language, which is given by the
following set $\ltermslsc$ of terms, where $\tm\esub{\var}{\tmtwo}$ is a new constructor called \emph{explicit substitution} (shortened ES), that is 
equivalent to $\letexp\var\tmtwo\tm$: 
\[\begin{array}{c\colspace \colspace \colspace ccc}
    \textsc{LSC Terms} & \tm,\tmtwo & \recdef & \var \mid \la\var\tm\mid \tm\tmtwo\mid \tm\esub{\var}{\tmtwo}
  \end{array}\] 
The notion of \emph{free} variable is defined as expected, in particular,
$\fv{\tm\esub{\var}{\tmtwo}} :=
(\fv{\tm} \setminus \set{\var}) \cup \fv{\tmtwo}$.
\emph{(List of) substitutions} and \emph{linear head contexts} are given by the following grammars: 
\[ \begin{array}{r\colspace \colspace lll}
\textsc{(List of) substitution contexts} &   \L & \grameq & \ctxhole \mid \L\esub{\var}{\tm} \\
\textsc{Linear head contexts} &    \lhc   & \grameq & \ctxhole \mid \la \var \lhc \mid \lhc \tm \mid \lhc\esub{\var}{\tm}\\
   \end{array} \] We write $\putinctx{\L}{\tm}$
(resp. $\putinctx{\lhc}{\tm}$) for the term obtained by replacing the
whole $\ctxhole$ in context $\L$ (resp. $\lhc$) by the term
$\tm$. This \emph{plugging} operation, as usual with contexts, can
capture variables. We write $\lhc\cwc{\tm}$ when we want to stress that the context $\lhc$ does not capture the free variables of $\tm$. 

\paragraph{Normal, neutral, and abs predicates.} The predicate 
\emph{$\lhnormalvoid$}  defining linear head normal terms 
and  \emph{$\lhneutralvoid$}  defining linear head neutral terms are
introduced in~\reffig{linear head-normal-forms}. They are a bit more involved than before, because switching to the micro-step granularity of the LSC the study of normal forms requires a finer analysis. The predicates are now based on three auxiliary predicates $\lhneutralpr{\var}$, $\lhnormalpr{\var}$, and $\lhnormalclosepr$: the first two characterise neutral and normal terms whose head variable $\var$ is free, the third instead characterises normal forms whose head variable is bound. Note also that the abstraction predicate $\abslhvoid$ is now  defined \emph{modulo} ES, that is, a term such as $(\la\var\tm)\esub\varthree\tmtwo\esub\vartwo\tmthree$ satisfies the predicate. 
It is worth noticing that a term $\tm$  of the form $\lhc\cwc{\vartwo}$ 
does not necessarily verify $\lhnormal \tm$, \eg\ 
$(\l \varthree. (\vartwo \var)\esub{\var}{\vartwo})\tmtwo$. 
Examples of linear head normal forms
are $\la \var \var \vartwo$ and $(\vartwo \var)\esub{\var}{\varthree}(I I)$.

\begin{figure*}
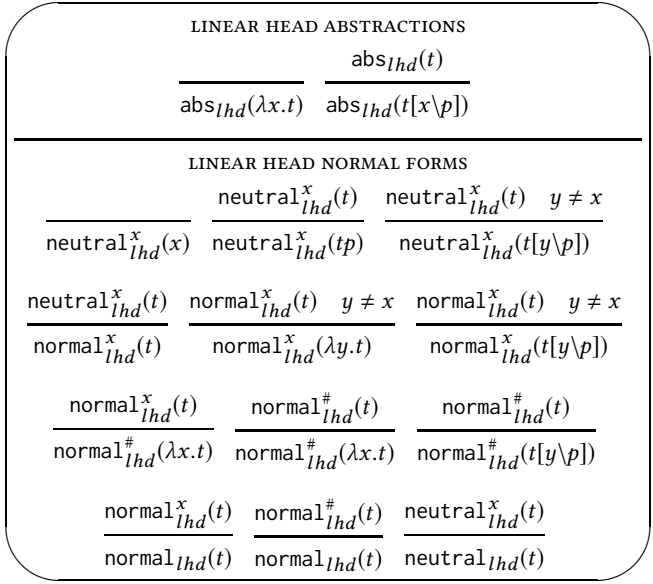

\centering
  \ovalbox{\small
$\begin{array}{c}
\multicolumn{1}{c}{\textsc{linear head abstractions}}\\[3pt]
\infer{\abslh{\l \var. \tm}}{\phantom{\abslh{\l \var. \tm}}} 
\msep 
\infer{\abslh{\tm\esub{\var}{\tmtwo}}}{\abslh{\tm}} \\\\[-3pt]
      \hline\\[-8pt]
\multicolumn{1}{c}{\textsc{linear head normal forms}}\\[3pt]
\infer[
      ]{\lhneutralp{\var}{\var}}{ \phantom{\lhneutralp{\var}{\var}} } 
\msep
\infer[
      ]{ \lhneutralp {\tm \tmtwo}\var }
        { \lhneutralp \tm \var } 
\msep
\infer[
      ]{ \lhneutralp {\tm \esub\vartwo\tmtwo} \var }
        { \lhneutralp \tm \var \quad \vartwo\neq\var } 
\\\\

\infer[
      ]{ \lhnormalp \tm \var }
        { \lhneutralp \tm \var } 
\msep 
\infer[
      ]{ \lhnormalp {\la\vartwo\tm} \var }
        { \lhnormalp \tm \var \quad \vartwo \neq \var} 
\msep
\infer[
      ]{ \lhnormalp {\tm \esub\vartwo\tmtwo} \var }
        { \lhnormalp \tm \var \quad \vartwo\neq\var } 
\\\\
\infer[
      ]{ \lhnormalclose {\la\var\tm} }{ \lhnormalp {\tm} \var } 
\msep 
\infer[
      ]{ \lhnormalclose {\la\var\tm} }{ \lhnormalclose \tm} 
\msep
\infer[
      ]{ \lhnormalclose {\tm \esub\vartwo\tmtwo} }
        { \lhnormalclose \tm  } \\ \\
\infer[
      ]{ \lhnormal \tm}
        { \lhnormalp \tm \var} 
\msep  
\infer[
      ]{ \lhnormal \tm}
        { \lhnormalclose \tm  }
\msep
\infer[]{\lhneutral{\tm}}{ \lhneutralp \tm \var} 
\end{array}$}
\caption{linear head  neutral and normal forms}
  \label{fig:linear head-normal-forms}
\end{figure*}

\paragraph{Small-step semantics.}
Linear head evaluation is often specified by means of a
non-deterministic strategy (having the diamond
property)~\cite{DBLP:conf/rta/Accattoli12}. Here, however, we
present a minor deterministic variant, in order to follow the general
schema presented in the introduction. The deterministic
notion of linear head evaluation $ \lh$ is given in
\reffig{deterministic-linear-head}. An example of $\tolsp$-sequence is
\[ \begin{array}{lllll}
((\l \varthree. (\var \var)\esub{\var}{\vartwo}) \tmtwo) \esub{\vartwo}{\varfour} & \tolh &
 (\var \var)\esub{\var}{\vartwo} \esub{\varthree}{\tmtwo} \esub{\vartwo}{\varfour} & \tolh \\
 (\vartwo \var)\esub{\var}{\vartwo} \esub{\varthree}{\tmtwo} \esub{\vartwo}{\varfour} & \tolh & 
(\varfour \var)\esub{\var}{\vartwo} \esub{\varthree}{\tmtwo} \esub{\vartwo}{\varfour} \\
\end{array} \]

\begin{figure*}
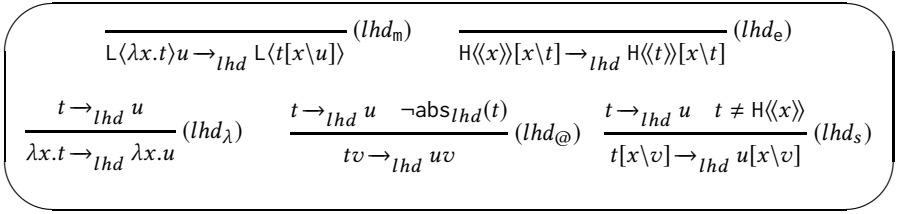

\centering
  \ovalbox{\small
$\begin{array}{c}
\infer[\lhb]{\putinctx{\L}{\la \var \tm} \tmb \tolh \putinctx \L {\tm \esub{\var}{\tmb}}  }{ \phantom{.} } \qquad

\infer[\lhs]{ \lhc\cwc{\var}\esub{\var}{\tm} \tolh  \lhc\cwc{\tm}\esub{\var}{\tm}}{\phantom{.} } \\\\

\infer[\lhabs]{\la \var \tm \tolh \la \var  \tmb}
              {\tm \tolh  \tmb } \qquad

\infer[\lhapp]{\tm \tmbb  \tolh \tmb \tmbb}
              {\tm \tolh \tmb\quad\neg \abslh{\tm} } \quad

\infer[\lhsub]{\tm \esub{\var}{\tmbb}  \tolh \tmb \esub{\var}{\tmbb}}
              {\tm \tolh \tmb \quad \tm \neq \lhc\cwc\var } \\ \\

\end{array}$}
\caption{Deterministic linear-head evaluation}
\label{fig:deterministic-linear-head}
\end{figure*}

From now on, we split the evaluation relation $\tolsp$ in two different
relations, \emph{multiplicative} $\tolspb$ and \emph{exponential}
$\tolspe$ evaluation, where $\tolspb$ (resp. $\tolspe$) is generated
by the base case $\lspb$ (resp. $\lspe$) and closed by the three rules
$\lspabs, \lspapp, \lspsub$. The
terminology \emph{multiplicative} and \emph{exponential} comes from
the linear logic interpretation of the LSC. The literature contains
also an alternative terminology, using \emph{$\Bsym$ at a distance}
for $\tolspb$ (or \emph{distant $\Bsym$}, where $\Bsym$ is a common name
for the variant of $\beta$ introducing an ES instead of using
meta-level substitution) and \emph{linear substitution} for
$\tolspe$.

\begin{toappendix}
\begin{proposition}[linear head  evaluation system]
\label{prop:memory-bievaluation} \mbox{} \\
  $(\ltermslsc,\tolsp, \sysneutralpr \systemlsp, \sysnormalpr \systemlsp, \sysabspr \systemlsp)$ is an evaluation system.
\end{proposition}
\end{toappendix}

In the linear case the proof is subtler than for the head, LO, and maximal cases.
\longshort{It is in \refappendix{linear-spine}.}{
  It can be found in the long version of this paper~\cite{AccattoliGLK2018-long}.
}

\paragraph{Sizes.} The notion of \emph{linear head size} $\lhsize\tm$ 
extends the head size to terms with ES by
counting 1 for variables---note that ES do not
contribute to the linear head size:

\[ \begin{array}{ccl\colspace \colspace ccl}
 \multicolumn{6}{c}{\textsc{Linear head size}}\\
\lhsize \var & :=  & 1 &     \lhsize { \la \var \tm}  & := & \lhsize \tm +1 \\
\lhsize { \tm \tmb } &  := & \lhsize \tm  +1 &     \lhsize {\tm\esub{\var}{\tmb}} & :=& \lhsize \tm 
\end{array}     \]

\paragraph{Multi types.} We consider the same multi types of \refsect{head-skeleton},
but now typing judgements are of the form
$\Deri[(\steps, \esteps, \spine)] {\typctx}{\tm}{\type}$, where
$(\steps, \esteps, \spine)$ is a \emph{triple} of integers whose
intended meaning is explained in the next paragraph.  The typing
system $\systemlsp$ is defined in \reffig{linear head-type-system}.
By abuse of notation, we use for all the typing rules---except
$\esrule$ which is a new rule---the same names used for $\spi$.

The \emph{linear head size} $\syssize\systemlsp\tderiv$ of a typing
derivation $\tderiv$ is the number of rules in $\tderiv$, without
counting the occurrences of $\many$.

\begin{figure}
\centering
\ovalbox{
$\begin{array}{cc}

\infer[\ax]{\Deri[(0, 0, 1)] {\var \col \single \type} \var \type}
			     {} 

&

\infer[\many]{\Deri[(+_{\iI} \steps_i, +\esteps_i, +_{\iI} \spine_i)] 
             { \uplus_{\iI}\typctxtwo_i  } \tm {\MSigma {\type_i} {\iI}}}
             {(\Deri[(\steps_i, \esteps_i, \spine_i)] {\typctxtwo_i} \tm {\type_i})_{\iI}} 

\\\\

\infer[\funsteps]{\Deri[(\steps+ 1, \esteps + \size\M, \spine - \size\M)]
                       {\typctx} {\la\var\tm} {\M \rightarrow \type}}
                 {\Deri[(\steps, \esteps,  \spine)] {\typctx;  \var \col \M} \tm \type } 
&

\infer[\funresult]{\Deri[(\steps,\esteps,  \spine +1)] {\typctx} { \la\var\tm } \abstype}
                  {\Deri[(\steps,\esteps, \spine)] {\typctx; \var \col \mtight} \tm \tight} 
\\\\

\infer[\appsteps]{\Deri[(\steps + \steps' + 1, \esteps+ \esteps', \spine + \spine')]
                       {\typctx \uplus_{\iI} \typctxtwo_i}
                       {\tm \tmtwo} \type}
                 {\Deri[(\steps, \esteps, \spine)] \typctx \tm {\M \rightarrow \type} \quad
                  \Deri[(\steps', \esteps', \spine')] {\typctxtwo} \tmtwo {\M} }  

&

\infer[\appresult]{\Deri[(\steps, \esteps, \spine+1)] 
                        {\typctx} {\tm \tmb} \neutype}
                  {\Deri[(\steps, \esteps, \spine)] \typctx \tm \neutype}
\\\\
\multicolumn{2}{c}{
\infer[\esrule]{\Deri[(\steps + \steps',  \esteps+ \esteps' + \size\M, \spine  + \spine' - \size\M)]
                     {\typctx \uplus \typctxtwo}
                     {\tm \esub\var\tmb} \type}
               {\Deri[(\steps,  \esteps, \spine)] {\typctx; \var  \col \M } \tm \type 
                  \quad \quad
                \Deri[(\steps',  \esteps', \spine')] {\typctxtwo} \tmb \M
                 }
}
\end{array}$
}
\caption{Type system for linear head evaluation.}
\label{fig:linear head-type-system}
\end{figure}

Note that $\Deri[(\steps,  \esteps, \spine)]{\typctx} {\lhc\cwc{\var}} \type $
implies that $\typctx = \typctx' ; \var:\M$ with $\M \neq \emm$.

\paragraph{Indices.} The roles of 
the three components of $(\steps, \esteps, \spine)$ 
in a typing derivation $\Deri[(\steps,  \esteps, \spine)]{\typctx} {\tm} \type$ can be described as follows:

\begin{itemize}
\item \emph{$\steps$ and multiplicative steps}:
  $\steps$ counts the rules of the derivation that can be used to form multiplicative redexes,
  \ie\  subterms of the form $\putinctx \L{\la \var \tm} \tmb$. To count this
  kind of redexes it is necessary to 
  count the  number of $\funsteps$ and $\appsteps$ rules.
  As in the case of head evaluation (\refsect{head-skeleton}), $\steps$ is at least twice the number of 
  multiplicative steps  to normal form
  because typing such  redexes  requires (at least) these  two rules.

\item \emph{$\esteps$ and exponential steps}:
  $\esteps$ counts the rules of the derivation that can be used to
  form exponential redexes, \ie\ subterms of the form
  $\lhc\cwc{\var}\esub{\var}{\tm}$.  To count this kind of redexes it
  is necessary to count, for every variable $\var$, the number of
  occurrences that can be substituted during evaluation. The ES is not
  counted because a single ES can be involved in many
  exponential steps along an evaluation sequence.

\item \emph{$\spine$ and size of the result}:
  $\spine$ counts the rules typing variables, abstractions and applications
  (\ie\ $\ax$, $\funresult$ and $\appresult$)
  that cannot be consumed by $\lhsp$ evaluation, so that they 
  appear in the linear head normal form of a term. 
  Note that the ES constructor is not consider part of
  the head of terms. 
\end{itemize}

Note also that the typing rules assume that variable occurrences
(corresponding to $\ax$ rules) end up in the result, by having the
third index set to $1$. When a variable $\var$ becomes bound by an ES
(rule $\esrule$) or by an abstraction destined to be applied
($\funsteps$), the number of uses of $\var$, expressed by the
multiplicity of the multi-set $\M$ typing it, is subtracted from the
size of the result, because those uses of $\var$ correspond to the
times that it shall be replaced via a linear substitution step, and
thus they should no longer be considered as contributing  to the
result. Coherently, that number instead contributes to the index
tracing linear substitution steps.

\begin{definition}[\Precise derivations]\strut
  A derivation $\tderiv \exder[\systemlsp] \Deri[(\steps, \esteps, \spine)]{\typctx}{\tm}{\typetwo}$
  is \emph{\precise} 
  if $\tightpred\typetwo$ and $\tightpred{\typctx}$. 
\end{definition}

\paragraph{Example}

Let us give  a concret example in  
system $\lsp$. Consider again the term
$\tm_0 = (\la {\var_1} (\la {\var_0} \var_0 \var_1) \var_1) \id$, where
$\id$ is the  identity function $\la \varthree \varthree$. 
The linear head evaluation sequence from $\tm_0$ to $\lsp$ normal-form
is given below, in which we distinguish the multiplicative steps from the exponential ones. 
\[ \begin{array}{llll}
   (\la {\var_1} (\la {\var_0} \var_0 \var_1) \var_1) \id  & \tolspb &
   ((\la {\var_0} \var_0 \var_1) \var_1) \esub{\var_1}{\id} & \tolspb \\
   (\var_0 \var_1)\esub{\var_0}{\var_1} \esub{\var_1}{\id} & \tolspe &
   (\var_1 \var_1)\esub{\var_0}{\var_1} \esub{\var_1}{\id} & \tolspe \\
   (\id \var_1)\esub{\var_0}{\var_1} \esub{\var_1}{\id}  & \tolspb &
   \var_3\esub{\var_3}{\var_1} \esub{\var_0}{\var_1} \esub{\var_1}{\id} & \tolspe \\
   \var_1 \esub{\var_3}{\var_1} \esub{\var_0}{\var_1} \esub{\var_1}{\id} & \tolspe &
   \id \esub{\var_3}{\var_1} \esub{\var_0}{\var_1} \esub{\var_1}{\id}
   \end{array} \]
The evaluation sequence has length $7$: $3$ multiplicative steps and $4$ exponential steps.
The linear head normal form has size $2$.  
We now give a \precise\ typing for the term $\tm_0$, by writing again
$\abstype_1$ for $\ty{\mult{\abstype}}{\abstype}$.
 {\scriptsize
\[ \prooftree
   \prooftree
   \prooftree
     \prooftree
       \prooftree
        \Deri[(0,0,1)]{\var_0:\mult{\abstype_1}}{ \var_0 }{\abstype_1} \quad 
        \prooftree
        \Deri[(0,0,1)]{\var_1:\mult{\abstype}}{\var_1}{\abstype}
        \justifies{\Deri[(0,0,1)]{\var_1:\mult{\abstype}}{\var_1}{\mult{\abstype}}}
        \endprooftree
       \justifies{\Deri[(1,0,2)]{\var_0:\mult{\abstype_1}, \var_1:\mult{\abstype}}{ \var_0 \var_1}{\abstype}}
      \endprooftree 
      \justifies{\Deri[(2,1,1)]{\var_1:\mult{\abstype}}{\la {\var_0} \var_0 \var_1}{\ty{\mult{\abstype_1}}{\abstype}} }
      \endprooftree \quad
      \prooftree
       \Deri[(0,0,1)]{\var_1:\mult{\abstype_1}}{\var_1}{\abstype_1}
       \justifies{\Deri[(0,0,1)]{\var_1:\mult{\abstype_1}}{\var_1}{\mult{\abstype_1}}}
      \endprooftree
    \justifies{\Deri[(3,1,2)] {\var_1:\mult{ \abstype, \abstype_1} }{ (\la {\var_0} \var_0 \var_1) \var_1}{\abstype }}
   \endprooftree
   \justifies{\Deri[(4,3,0)] {}{\la {\var_1} (\la {\var_0} \var_0 \var_1) \var_1}{\ty{\mult{ \abstype, \abstype_1}}{\abstype} } }
   \endprooftree
   \quad
   \prooftree
   \vdots
   \justifies{\Deri[(1,1,2)] {}{ \id  }{\mult{ \abstype, \abstype_1}}}
   \endprooftree
   \justifies{\Deri[(6,4,2)] {}{(\la {\var_1} (\la {\var_0} \var_0 \var_1) \var_1) \id  }{\abstype}} 
   \endprooftree \] 
}
Indeed, the pair $(6,4,2)$ represents $6/2=3$ (resp. $4$)
multiplicative (resp. exponential) evaluation steps to $\lsp$
normal-form, and a linear head normal form of size $2$.

\subsection{Tight Correctness}

As in the case of head and LO evaluation, the correctness proof is based
on three main properties: properties of normal forms---themselves based on a lemma about neutral terms---the  
interaction between (linear head) substitution and typings, and
subject reduction. 

\paragraph{Neutral terms ad properties of normal forms} 
As for the head case, the properties of tight typing of $\lh$ normal forms depend on a spreading property of $\lh$ neutral terms. Additionally, they also need to characterise the shape of typing contexts for tight typings of neutral and normal terms.

\begin{toappendix}
\begin{lemma}[Tight spreading on neutral terms, plus typing contexts]
\label{l:lin-head-headvar-in-context} 
Let $\tderiv \exder[\systemlsp] \Deri[(\steps, \esteps, \spine)]{\typctx}{\tm}{\type}$ be a derivation. 
\begin{enumerate}
  \item 
  If $\lhneutralp\tm\var$ then $\var \in \dom\typctx$. Moreover, if $\typctx(\var) = \mtight$ then $\type = \tight$ and $\dom\typctx = \set\var$.

  \item 
  If $\lhnormalp\tm\var$ then $\var \in \dom\typctx$. Moreover, if $\typctx(\var) = \mtight$ then $\dom\typctx = \set\var$.

  \item 
  If $\lhnormalclose \tm$  and $\type = \tight$ then $\type = \abstype$ and $\typctx$ is empty.
\end{enumerate}
\end{lemma}
\end{toappendix}


\begin{toappendix}
\begin{proposition}[Properties of $\lhsp$ tight typings for normal forms]
\label{prop:tight-normal-forms-indicies-linear-head}
Let $\tm$ be such that $\sysnormal\systemlsp\tm$, and
$\tderiv\exder[\systemlsp] \Deri[(\steps, \esteps, \spine)]\typctx{\tm}\type$ be a
typing derivation. 
\begin{enumerate}
\item 
\emph{Size bound}: $\lspsize\tm \leq \syssize\systemlsp\tderiv$.

\item 
\emph{Tightness}: if $\tderiv$ is \precise\  then $\steps = \esteps = 0$ and $\spine = \lspsize \tm$.

\item 
\emph{Neutrality}: if $\type=\neutype$ then $\sysneutral{\systemlsp}{\tm}$.
\end{enumerate}
\end{proposition}
\end{toappendix}


\paragraph{Partial substitution lemma} The main difference 
in the proof schema with respect to the head case is about the
substitution lemma, that is now expressed differently,
because evaluation no longer relies on meta-level substitution.
Linear substitutions consume one type at a time: performing a linear head substitution on a term of the form
$\lhc\cwc{\var}\esub{\var}{\tm}$
consumes \emph{exactly} one type resource 
associated to the variable $\var$, and
all the other ones remain
in the typing context after the partial substitution. Formally, 
an induction on $\lhc$ is used to show:

\begin{toappendix}
 \begin{lemma}[Partial substitution and typings for $\lhd$]
\label{l:subst:hwmilner}
The following rule is admissible in system $\systemlsp$: 
 \[
   \infer[\partialsubslemma]{
     \Deri[(\steps+\steps_\tmthree, \esteps + \esteps_\tmthree, \spine + \spine_\tmthree -1)]{x:\M\setminus \sig;
  \typctx \mplus \typctx_\tmthree}{\lhc\cwc{\tmthree}}{\tau}
   }{
     \Deri[(\steps, \esteps, \spine)]{x:\M;\typctx}{\lhc\cwc{x}}{\tau}
     \quad
     \Deri[(\steps_\tmthree, \esteps_\tmthree,
    \spine_\tmthree)]{\typctx_\tmthree}{\tmthree}{\sig} \quad \sig \in \M
   }
  \]
\end{lemma}
\end{toappendix}


\paragraph{Subject reduction and correctness.} 
Quantitative subject reduction is also refined, by taking into account the fact that now there are two evaluation steps, whose numbers are traced by two different indices.

\begin{toappendix}
\begin{proposition}[Quantitative subject reduction for $\lhd$]
\label{prop:head-subject-head-reduction}
If  $\tderiv\tri \tyjn{(\steps, \esteps, \spine)}{\tm}\typctx\type$ 
then 
\begin{enumerate}
  \item 
  If $\tm  \tom  \tmb$ then $\steps \geq 2$ and there is a 
 typing  $\tderivtwo$ such that $\tderivtwo\tri \tyjn{(\steps-2, \esteps, \spine)}{\tmb}\typctx\type$.
  \item 
  If $\tm \toe \tmb$ then $\esteps \geq 1$ and there is a 
 typing 
$\tderivtwo$ such that $\tderivtwo \tri \tyjn{(\steps,\esteps -1, \spine)}{\tmb}\typctx\type$.
\end{enumerate}
\end{proposition}
\end{toappendix}

The proof is by induction on $\tm  \tom  \tmb$ and $\tm  \toe  \tmb$, using \reflemma{subst:hwmilner}. Note that quantitative subject reduction does not assume that the typing derivation is tight: as for the head case, the tight hypothesis is only used for the study of normal forms---it is needed for subject reduction / expansion only if evaluation can take place inside arguments, as in the leftmost and maximal cases.

According to the spirit of tight typings,
linear head correctness does not only 
provide the size of (linear head) normal forms,
but also 
 the lengths of evaluation
sequences to (linear head) normal form: the two first integers
$\steps$ and $\esteps$ in
the final judgement count exactly the total number 
of  evaluation steps to (linear-head) normal form. 

\begin{toappendix}
\begin{theorem}[Tight correctness for $\lhd$]
\label{tm:head-correctness} 
Let  $\tderiv \exder[\systemlsp] \Deri[(\steps, \esteps,\spine)] {\typctx}{\tm}{\type}$
be a tight derivation.
Then there exists $\tmtwo$ such that $\tm  \tolsp^{\steps/2 + \esteps} \tmtwo$, 
$\lhnormal\tmtwo$ and $\lspsize\tmtwo = \spine$. Moreover, if $\type=\neutype$
  then $\sysneutral\lh\tmtwo$.
\end{theorem}
\end{toappendix}

\subsection{Tight Completeness}

As in the case of head and LO evaluation the completeness proof is
based on the following properties: typability of linear head normal
forms, interaction between (linear head) anti-substitution and
typings, and subject expansion. The proofs are analogous to
those of the completeness for head and LO evaluation, up to the
changes for the linear case, that are instead analogous to those of
the correctness of the previous subsection. The statements
follow.

\begin{toappendix}
\begin{proposition}[Linear head normal forms are tightly typable for $\lhd$]
  \label{prop:lin-head-normal-forms-are-tightly-typable}
 Let  $\tm$ be such that $\sysnormal\systemlsp\tm$. Then 
there exists a \precise\ typing $\tderiv \exder[\systemlsp] \Deri[(0, 0, \lspsize\tm)]{\typctx}{\tm}{\type}$. Moreover, if 
$\sysneutral\systemlsp\tm$ then $\type = \neutype$, and if $\sysabs\systemlsp\tm$ then $\type = \abstype$.
\end{proposition}
\end{toappendix}


 
\begin{toappendix}
\begin{lemma}[Partial anti-substitution and typings for $\lhd$]
\label{l:anti-subst}
Let  $\tderiv \exder[\systemlsp] \Deri[(\steps, \esteps, \spine)]{\typctx}{\lhc\cwc{\tmthree}}{\type}$,
where $\var \notin \tmthree$. Then  there exists 
\begin{itemize}
\item a type  $\typetwo$
\item a typing derivation $\tderiv_\tmthree  \exder[\systemlsp] \Deri[(\steps_\tmthree, \esteps_\tmthree, \spine_\tmthree)]{\typctx_\tmthree}{\tmthree}{\typetwo}$ 
\item a typing derivation $\Phi_{\lhc\cwc{x}} \exder[\systemlsp] 
 \Deri[(\steps', \esteps', \spine')]{  \typctx'+x{:}\mult{\typetwo} }{\lhc\cwc{x}}{\type}$
\end{itemize}
such that 
\begin{itemize}
\item \emph{Typing contexts}: $\typctx = \typctx' \uplus \typctx_\tmthree$.

\item \emph{Indices}: $(\steps, \esteps, \spine) = (\steps' +  \steps_\tmthree, \esteps' + \esteps_\tmthree, \spine' + \spine_\tmthree - 1)$.
\end{itemize}
Moreover, if $\tderiv$ is \precise\ then so are
  $\tderiv_\tmthree$ and $\Phi_{\lhc\cwc{x}}$.

\end{lemma}
\end{toappendix}


\begin{toappendix}
\begin{proposition}[Quantitative subject expansion for $\lhd$] 
\label{prop:linear-subject-expansion}
If $\tderiv' \exder[\systemlsp] \Deri[(\steps, \esteps, \spine)]{\typctx}{\tm'}{\type}$ then
\begin{enumerate}
\item If  $\tm \tolhb \tm'$ then  there is a derivation $\tderiv \exder[\systemlsp] \Deri[(\steps+2, \esteps, \spine)]{\typctx}{\tm}{\tau}$.
\item If $t\tolhs t'$ then  there is a derivation $\tderiv \exder[\systemlsp]
\Deri[(\steps, \esteps+1, \spine)]{\typctx}{\tm}{\type}$.
\end{enumerate}
\end{proposition}
\end{toappendix}


As  for linear head correctness, linear head completeness 
also \emph{refines} the information provided about the lenghts of the
evaluation sequences: the number $k$ of evaluation steps
to (linear head) normal form is now split
into two integers $k_1$ and $k_2$ representing,
respectively, 
the  multiplicative and exponential steps in such evaluation sequence.

\begin{toappendix}
\begin{theorem}[Tight completeness for $\lhd$]
\label{th:completeness-linear-head}
 Let $\tm \tolsp^k \tmtwo$, where  $\lhnormal\tmtwo$. Then there exists a \precise\ type
derivation  $\tderiv \exder[\systemlsp] \Deri[(2k_1, k_2, \lhsize\tmtwo)]{\typctx}{ \tm}{\type}$, where
$k=k_1+k_2$. Moreover, if $\lhneutral\tmtwo$, then $\type = \neutral$, and if $\sysabs\systemlsp\tmtwo$ then $\type = \abstype$.
\end{theorem}
\end{toappendix}

\subsection{Relationship Between Head and Linear Head}

The head and linear head strategies are specifications at different granularities of the same notion of evaluation. Their type systems are also closely related---in a sense that we now make explicit, they are the same system.

In order to formalise this relationship we
define the transformation $\lintransfpr$ of $\spi$-derivations into (linear, hence the notation) $\lsp$-derivations as:
$\ax$ in $\spi$ is mapped to $\ax$ in $\lsp$, $\funresult$ in
$\spi$ is mapped to $\funresult$ in $\lsp$, $\funsteps$ in
$\spi$ is mapped to $\funsteps$ in $\lsp$, and so on.  This
transformation preserves the context and the type of all the typing
judgements. Of course, if one restricts the $\lsp$ system to $\l$-terms, there is an inverse transformation $\nonlintransfpr$ of $\lsp$-derivations into (non-linear, hence the notation) $\spi$-derivations, defined as expected. Together, the two transformation realise an isomorphism.

\begin{proposition}[Head isomorphism]
\label{prop:head-iso}
Let $\tm$ be a $\l$-term without explicit substitutions. Then 
\begin{enumerate}
  \item \label{p:head-iso-nl-to-l}
  \emph{Non-linear to linear}: if $\tderiv \exder[\spi] \Deri[(\steps, \result)] {\typctx}{\tm}{\type}$
then there exists $\esteps\geq 0$ such that $\lintransf \tderiv \exder[\systemlsp] \Deri[(\steps, \esteps, \result+1)] {\typctx}{\tm}{\type}$. Moreover, $\nonlintransf {\lintransf \tderiv} = \tderiv$.

  \item \label{p:head-iso-l-to-nl}
  \emph{Linear to non-linear}: if $\tderiv \exder[\systemlsp] \Deri[(\steps, \esteps, \result)] {\typctx}{\tm}{\type}$ then $\nonlintransf \tderiv \exder[\spi] \Deri[(\steps, \result-1)] {\typctx}{\tm}{\type}$. Moreover, $\lintransf {\nonlintransf \tderiv} = \tderiv$.
\end{enumerate}
\end{proposition}

The proof is straightforward.

Morally, the same type system measures both head and linear head
evaluations. The difference is that to measure head evaluation and head normal forms
one forgets the number of axiom typing rules, that coincides
exactly with the number of linear substitution steps, plus 1 for the head variable of the linear normal form. In this sense,
multi types more naturally measure linear head evaluation. Roughly, a
tight multi type derivation for a term is nothing else but a coding of
the evaluation in the LSC, including the normal form itself.

\paragraph{On the number of substitution steps.} It is natural to wonder how the index $\esteps$ introduced by $\lintransfpr$ in \refpropp{head-iso}{nl-to-l} is related to the other indices $\steps$ and $\result$. This kind of questions has been studied at length in the literature about reasonable cost models. It is known that
$\esteps = {\mathcal O}(\steps^2)$  for \emph{any} $\l$-term, even for untypable ones, see~\cite{DBLP:conf/rta/AccattoliL12} for details. The bound is typically reached by the diverging term $\delta\delta$, which is untypable, but also by the following
typable term  $\tm_n \defeq (\la {\var_n} \ldots
(\la {\var_1} (\la {\var_0}
(\var_0 \var_1 \ldots \var_n)) \var_1) \var_2\ldots \var_n) I$.
Indeed, $\tm_n$ evaluates in $2n$ multiplicative steps (one for turning 
each $\beta$-redex into an ES, and one for each
time that the identity comes in head position) and $\Omega(n^2)$
exponential steps. 

\paragraph{On terms with ES} Relating typings for $\l$-terms with ES to typings for ordinary $\l$-terms is a bit trickier---we only sketch the idea. One needs to introduce the \emph{unfolding} operation $\unf{(\cdot)}: \ltermslsc \rightarrow \lterms$ on $\l$-terms with ES, that turns all ES into meta-level substitutions, producing the underlying ordinary $\l$-term. 
For instance,
$\unf{(\var\esub{\var}{\vartwo} \esub{\vartwo}{\varthree})} =
\varthree$. As in \refpropp{head-iso}{l-to-nl}, types are preserved:
\begin{lemma}[Unfolding and $\lsp$ derivations]
Let $\tm \in \ltermslsc$. If 
$\tderiv \exder[\systemlsp] \Deri[(\steps, \esteps, \result)] {\typctx}{\tm}{\type}$
then there exists $\tderivtwo \exder[\spi] \Deri[(\steps, \result -1)] {\typctx}{\unf\tm}{\type}$.
\end{lemma}
Note that the indices are also preserved. It is possible to also spell out the relationship between $\tderiv$ and $\tderivtwo$ (as done in~\cite{KesnerVR18}), that simply requires a notion of unfolding of typing derivations, and that collapses on the transformation $\nonlintransfpr$ in the case of ordinary $\l$-terms.

\section{Leftmost Evaluation and Minimal Typings}

\label{sect:shrinking}
This section focuses on the LO system and on the relationship between tight and tight-free---deemed \emph{traditional}---typings. Contributions are manyfold:
\begin{enumerate}
  \item \emph{LO normalisation, revisited}: we revisit the characterisation of LO normalising terms as those typable with \emph{shrinking} typings, that is, those where the empty multi-set has no negative occurrences. The insight is that the shrinking and tight constraints are of a very similar nature, showing that our technique is natural rather than ad-hoc. Moreover, our notion of shrinking derivation can also include the tight constants, thus we provide a strict generalisation of the characterisation in the literature.
  
  \item \emph{Minimality}: we show that tight typings can be seen as a characterisation of minimal traditional shrinking typings. The insight here is that tight typings are simply a device to focalise what traditional types can already observe in a somewhat 
more technical way.
  
  \item \emph{Type bound}: we show that for traditional shrinking typings already the type itself---with no need of the typing derivation---provides a bound on the size of the normal form, and this bound is exact if the typing is minimal. The insight is then the inherent inadequacy of multi types as a tool for reliable complexity measures.
\end{enumerate}

This study is done with respect to LO evaluation because among the case studies of the paper it is the most relevant one for reasonable cost models. It may however be easily adapted, \emph{mutatis mutandis}, to the other systems.

\paragraph{Shrinking typings.} It is standard to characterise LO normalising terms as those typable with intersection types without negative occurrences of $\Omega$~\cite{krivine1993lambda}, or, those typable with multi types without occurrence of the empty 
multi-set $\emm$~\cite{BKV17}. We call this constraint \emph{shrinking}. We need some basic notions. We use the notation $\atype$ to denote a \emph{(multi)-type}, that is, either a type $\type$ or a multi-set of types $\M$. 

\begin{definition}[Positive and negative occurrences]
  Let $\atype$ be a (multi)-type. The sets of positive and negative occurrences of $\atype$ in a type/multi-set of types/typing context are defined by mutual induction as follows:
  \[\begin{array}{cccccc}
    \infer{\type \in \possubtype \type }{}
    &    
    \infer{ \atype \in \possubtype \M }{ \exists\typetwo \in \M \mbox{ such that } \atype \in \possubtype \typetwo }
    &
    \infer{ \atype \in \negsubtype \M }{ \exists\typetwo \in \M \mbox{ such that } \atype \in \negsubtype \typetwo }
    \\\\
    \infer{\M \in \possubtype \M }{}
    &
    \infer{ \atype \in \possubtype {\M \rightarrow \type} }{ \atype \in \negsubtype\M \mbox{ or } \atype \in \possubtype \type }
    &
    \infer{ \atype \in \negsubtype {\M \rightarrow \type} }{ \atype \in \possubtype\M \mbox{ or } \atype \in \negsubtype \type }
    \\\\
    &
    \infer{ \atype \in \possubtype {\var : \M, \typctx} }{ \atype \in \possubtype\M \mbox{ or } \atype \in \possubtype \typctx }
    &
    \infer{ \atype \in \negsubtype {\var : \M, \typctx} }{ \atype \in \negsubtype\M \mbox{ or } \atype \in \negsubtype \typctx }
  \end{array}\]
\end{definition}

Shrinking typings are defined by imposing a condition on the final judgement of the derivation, similarly to tight typings.

\begin{definition}[Shrinking typing]
  Let $\tderiv \exder[\lo] \tyjn{(\steps, \result)} \tm \typctx \type$ be a typing derivation.
  \begin{itemize}
    \item the typing context $\typctx$ is \emph{shrinking} if  $\emm \notin \negsubtype\typctx$;
    \item the final type $\type$ is \emph{shrinking} if 
$\emm \notin \possubtype\type$;
    \item The typing derivation $\tderiv$ is \emph{shrinking} if both $\typctx$ and $\type$ are shrinking.
  \end{itemize}
\end{definition}

Note that 
\begin{itemize}
  \item \emph{Final judgement}: being shrinking depends only on the final judgement of a typing derivation, and that
  \item \emph{Tight implies shrinking}: a tight typing derivation is always shrinking.
\end{itemize}

In this section we also have a close look to traditional derivations without tight constants.

\begin{definition}[\Dynamic typings]\label{def:traditional}
Let $\tderiv \exder[\ske] \tyjn{(\steps, \spine)}{\tm}{\typctx}{\type}$ be a typing derivation. Then $\tderiv$ is \emph{\dynamic} if no type occurring in $\tderiv$ is a tight type (and so rules $\funresult$ and $\appresult<\ske>$ do not occur in $\tderiv$). 
\end{definition}

One of our results is that the type itself bounds the size of $\lo$ normal forms, for traditional typings, according to the following notion of type size.

\begin{definition}[Type size]
The type size $\tysize\cdot$ of types, multi-sets, typing contexts, and derivations is defined as follows:
$$\begin{array}{rcl \colspace rcl }
\tysize\atomtype & \defeq & 0
&
\tysize\tight & \defeq & 0
\\
\tysize\M & \defeq & \sum_{\type\in\M}\tysize\type
&
\tysize{\tarrow\M\type} & \defeq & \tysize\M + \tysize\type + 1
\\
\tysize{\emptyctx} & \defeq & 0
&
\tysize{\var\col\M; \typctx} & \defeq & \tysize\M + \tysize\typctx
\\\\
 \tysize{\tderiv} & \defeq & \tysize\typctx + \tysize\type & \multicolumn{3}{l}{\mbox{if $\tderiv \exder \tyjn{(\steps, \spine)}{\tm}{\typctx}{\type}$}}
\end{array}$$
\end{definition}

\subsection{Shrinking Correctness}
Here we show that shrinking typability is preserved by LO evaluation and that the size of shrinking typings decreases along it---hence the name---so that every shrinkingly typable term is LO normalising. For the sake of completeness, we also show that typability is always preserved, but if the typing is not shrinking then its size may not decrease.

Once more, we follow the abstract schema of the other sections (but
replacing \emph{tight} with \emph{shrinking} and obtaining bounds that
are less tight). The properties of the typings of normal forms and the
substitution lemma for the $\lo$ system have already been proved
(\refprop{tight-normal-forms-indicies}
and \reflemma{typing-substitution-overview}). We deal again with
normal forms, however, because we now focus on traditional typings,
and show that their type bounds the size of the LO normal form. As in the previous sections, neutral terms play a key role, showing that our isolation of the relevance of
neutral terms for characterisation via multi types is not specific to
tight types.

\begin{toappendix}
\begin{proposition}[Traditional types bounds the size of neutral and normal terms]  
\label{prop:shrinking-normal-forms-forall}
  Let $\tderiv \exder[\ske] \tyjn{(\steps, \spine)}{\tm}{\typctx}{\type}$ be a traditional typing. Then:
  \begin{enumerate}
  \item If $\skneutral{\tm}$ then $\tysize\type + \sksize\tm \leq \tysize\typctx$.

 \item If $\sknormal \tm$ then $\sksize\tm \leq \tysize\tderiv$.
 \end{enumerate} 
\end{proposition}  
\end{toappendix}

As usual, shrinking correctness is based on a subject reduction property. In turn, subject reduction depends as usual on a spreading property on neutral terms, that is expressed by the following lemma.

\begin{toappendix}
\begin{lemma}[Neutral terms and positive occurrences]
  \label{l:shrinking-neutral-spreading}
  Let $\tm$ be such that $\loneutral \tm$ and $\tderiv \exder[\lo] \Deri[(\steps, \result)] {\typctx}{\tm}{\type}$ be a typing derivation. Then $\type$ is a positive occurrence of $\typctx$.
\end{lemma}
\end{toappendix}

It is interesting to note that this lemma subsumes the \emph{tight
spreading on neutral terms} property
of \reflemma{tight-spreading-spi-ske}, showing a nice harmony between
the shrinking and tight predicates on derivations. If the typing
context $\typctx$ is tight, indeed, the fact that $\type$ is a
positive occurrence of $\typctx$ implies that $\type$ is
tight.

\begin{toappendix}
\begin{proposition}[Shrinking subject reduction]
\label{prop:shrinking-subject-reduction} 
Let $\tderiv \exder[\ske] \tyjn{(\steps, \result)}{\tm}\typctx\type$. If $\tm \toske \tmtwo$ then $\steps\geq 2$ and there exists $\tderivtwo$ such that $\tderivtwo \exder[\ske]  \tyjn{(\stepstwo, \result)}{\tmtwo}\typctx\type$ with $\stepstwo \leq \steps$ (and so $\sksize\tderivtwo \leq \sksize\tderiv$). Moreover, if $\tderiv$ is shrinking then $\stepstwo \leq \steps - 2$ (and so $\sksize\tderivtwo \leq \sksize\tderiv-2$).
\end{proposition}
\end{toappendix}

Note that a LO diverging term like $\var (\delta\delta)$ is typable in system $\lo$ by assigning to $\var$ the type $\emptymset \rightarrow \atomtype$ and typing $\delta\delta$ with $\emptymset$, and that its type is preserved by LO evaluation, by \refprop{shrinking-subject-reduction}. Note however that the resulting judgement is not shrinking---only shrinkingly typable terms are LO normalising, in fact.

\begin{toappendix}
\begin{theorem}[Shrinking correctness]
\label{tm:shrinking-correctness} 
Let $\tderiv \exder[\ske]  \Deri[(\steps, \result)] {\typctx}{\tm}{\type}$ be a shrinking derivation. Then there exists a $\toske$ normal form  $\tmtwo$ and $k\leq \steps/2$ such that
\begin{enumerate}
 \item \emph{Steps}: $\tm$ $\toske$-evaluates to $\tmtwo$ in $k$ steps, \ie $\tm  \Rewn[k]{\ske} \tmtwo$;

 \item \emph{Result size}:  $\sksize\tmtwo \leq \sksize\tderiv - 2 k$;
\end{enumerate}
Moreover, if $\tderiv$ is traditional then $\sksize\tm \leq \tysize\tderiv$.
\end{theorem}
\end{toappendix}

\paragraph{Minimality} The minimality of tight typings is hidden in the statement of the shrinking correctness theorem. By combining together its two points, indeed, we obtain that $\losize\tmtwo +2k \leq \losize\tderiv$, that is, the size of every typing derivation bounds both the number of LO steps and the size of the LO normal form. Consequently, tight typings---whose size is exactly the sum of these two quantities---are minimal typings. To complete the picture, one should show that there also exist traditional shrinking typings that are minimal. We come back to this point at the end of the completeness part.

\subsection{Shrinking Completeness}
The proof of completeness for shrinking typings also follows, \emph{mutatis mutandis}, the usual schema. Normal forms and anti-substitution have already been treated (\refprop{normal-forms-are-tightly-typable} and \reflemma{anti-substitution}). Again, however, we repeat the study of (the existence of typings for) LO normal forms focussing now on traditional typings and on the bound provided by types. Their study is yet another instance of \emph{spreading on (LO) neutral terms}, in this case of the size bound provided by types: for neutral terms the size of the typing context $\typctx$ allows bounding both the size of the term and the size of its type, which is stronger than what happens for general LO normal terms.

\begin{toappendix}
\begin{proposition}[Neutral and normal terms have minimal traditional shrinking typings]  
\label{prop:shrinking-normal-forms-exist} 
  \hfill
  \begin{enumerate}
  \item If $\skneutral{\tm}$ then 
for every type $\type$ there exists a traditional shrinking typing $\tderiv \exder[\ske] \tyjn{(\sksize\tm, 0)} \tm \typctx \type$ such that $\tysize\type + \sksize\tm = \tysize\typctx$.

 \item If $\sknormal \tm$ then 
 there exists a type $\type$ and a traditional shrinking typing $\tderiv \exder[\ske] \tyjn{(\sksize\tm, 0)} \tm \typctx \type$ such that $\sksize\tm = \tysize\tderiv$.
 \end{enumerate}
\end{proposition}
\end{toappendix}

Note that the typings given by the propositions are minimal because they satisfy $\losize{\tderiv} = \losize\tm$ and $\losize\tm$ is a lower bound on the size of the typings of $\tm$ by \refprop{tight-normal-forms-indicies}.1---note that it is also the size of a tight typing of $\tm$ by \refprop{normal-forms-are-tightly-typable}.

The last bit is a subject expansion property. Note in particular that since $\beta$-redexes are typed using traditional rules, the expansion preserves traditional typings.

\begin{toappendix}
\begin{proposition}[Shrinking subject expansion]
\label{prop:shrinking-subject-expansion}
If $\tm \toske \tmtwo$ and $\tderiv \exder[\ske] \tyjn{(\steps,
  \spine)}\tmtwo\typctx\type$ then there exists $\tderivtwo$ such that
$\tderivtwo \exder[\ske] \tyjn{(\stepstwo, \spine)}\tm\typctx\type$ with $\stepstwo \geq  \steps $. Moreover, if $\tderiv$ is shrinking then $\stepstwo \geq  \steps +2 $, and if $\tderiv$ is traditional then $\tderivtwo$ is traditional.
\end{proposition}
\end{toappendix}

The completeness theorem then follows.

\begin{toappendix}
\begin{theorem}[Shrinking completeness]
\label{thm:shrinking-completeness}
Let $\tm \Rewn[k]{\ske} \tmtwo$ with $\tmtwo$ such that $\sknormal \tmtwo$. Then there exists a traditional shrinking typing $\tderiv \exder[\ske] \tyjn{(\steps, 0)} \tm \typctx \type$ such that $k\leq \steps/2$ and $\sksize\tmtwo = \tysize\tderiv$.
\end{theorem}
\end{toappendix}

\paragraph{Minimality, again} We have shown that there exist minimal traditional shrinking (MTS for short) typings of normal forms of the same size of the tight typings (\refprop{shrinking-normal-forms-exist}). It is possible to lift such a correspondence to all typable terms. This refinement is however left to a longer version of the paper, because it requires additional technicalities, needed to strengthen subject expansion. 

Let us nonetheless assume, for the sake of the discussion, that tight types characterise minimal typings for all typable terms. What's the difference between the two approaches then? It is easy to describe MTS typings for normal forms: an explicit description can be extracted from the proof of \refprop{shrinking-normal-forms-exist}, and essentially it interprets the terms \emph{linearly}, that is, typing arguments of applications only once, similarly to the tight $\appresult<\lo>$ rule. For non-normal terms, however, MTS typings can be obtained only \emph{indirectly}: first obtaining a MTS typing of the normal form, and then pulling the MTS property back via subject expansion. Tightness instead is a predicate that applies \emph{directly} to derivations of \emph{any} typable term, and is then a direct alternative to MTS typings.

\paragraph{Type bounds} The fact that for traditional typings the type itself provides a bound on the size of the normal form is a very strong property. It is in particular the starting point for de Carvalho's transfer of the study of bounds to the relational semantics of terms~\cite{Carvalho07,DBLP:journals/corr/abs-0905-4251}---a term is interpreted as the set of its possible types (including the typing context), that is a notion independent of the typing derivations themselves. Because of the size explosion problem, however, such property also shows that the bounds provided by the relational semantics are doomed to be lax and not really informative.

\emph{Relational denotational semantics.} As we said in the introduction, multi types can be seen as a syntactic presentation of relational denotational semantics, which is the model obtained by interpreting the $\l$-calculus into the relational model of linear logic~\cite{Girard88,DBLP:journals/apal/BucciarelliE01,Carvalho07,DBLP:conf/csl/Carvalho16}, often considered as a canonical model. 

The idea is that the interpretation (or 
semantics) of a term is simply the set of its types, together with their typing contexts. More precisely, let $\tm$ be a term and $\var_1, \dots, \var_n$ (with $n \geq 0$) be pairwise distinct variables.
If $\fv{\tm} \subseteq \{\var_1, \dots, \var_n\}$, we say that the list $\vec{\var} = (\var_1, \dots, \var_n)$ is \emph{suitable for} $\tm$.
If $\vec{\var} = (\var_1, \dots, \var_n)$ is suitable for $\tm$, the (\emph{relational}) \emph{semantics} 
\emph{of} $\tm$ \emph{for} $\vec{\var}$ is
\[
    \sem{\tm}_{\vec{\var}} \defeq \{((\M_1,\dots, \M_n),\type) \mid \exists  \, 
    \tderiv \exder[\ske] \tyjn{(\steps, \result)} \tm {\var_1 \colon\! \M_1, \dots, \var_n \colon\! \M_n} \type
    \} \,.
\]

By subject reduction and expansion, the interpretation $\sem{\tm}_{\vec{\var}} $ is an invariant of evaluation, and by correctness and completeness it is non-empty if and only if $\tm$ is LO normalisable. Said differently, multi types provide an adequate denotational model with respect to the chosen notion of evaluation, here the LO one. If the interpretation is restricted to traditional typing derivations (in the sense of \refdef{traditional}), then it coincides with the one in the relational model in the literature. General derivations still provide a relational model, but a slightly different one, with the two new types $\abstype$ and $\neutral$, whose categorical semantics still has to be studied.

\section{Conclusions}
\label{s:conclusion}

Type systems provide guarantees both \emph{internally} and \emph{externally}.
  Internally,
  a typing discipline ensures that a program in isolation has a given desired property.
  Externally, the property is ensured \emph{compositionally}:
  plugging a typed program in a typed environment preserves the desired property.
  Multi types (a.k.a.~non-idempotent intersection types)
  are used in the literature
  to quantify the resources that are needed to produce normal forms.  
  Minimal typing derivations provide exact upper bounds
  on the number of $\beta$-steps plus the size of the normal form%
  ---this is the internal guarantee.
  Unfortunately, such minimal typings provide almost no compositionality,
  as they essentially force the program to interact with a linear environment.
  Non-minimal typings allow compositions with less trivial environments,
  at the price of laxer bounds.

  In this paper we have engineered typing so that,
  via the use of \emph{tight constants} among base types,
  some typing judgements express compositional properties of programs
  while other typing judgements, namely the \emph{tight} ones,
  provide exact and separate bounds
  on the lengths of evaluation sequences on the one hand,
  and on the sizes of normal forms on the other hand.
  The distinction between the two counts
  is motivated by the size explosion problem,
  where the size of terms can grow exponentially with respect to the number of evaluation steps.

  We conducted this study, building on some of the ideas in~\cite{bernadet13},
  by presenting a flexible and parametric typing framework,
  which we systematically applied
  to three evaluation strategies of the pure $\l$-calculus:
  head, leftmost-outermost, and maximal.

In the case of leftmost-outermost evaluation, we have also developed 
the traditional shrinking approach which does not make use
of tight constants. One of the results is that the number of
(leftmost) evaluation steps can be measured using only the (sizes)
of the types of the final typing judgement, in contrast to
the size of the \emph{whole} typing derivation. Another point, is the connection between tight typings and minimal shrinking typings without tight constants.

In the case of maximal evaluation, we have circumvented the 
traditional techniques to show strong normalisation:
by focusing on the maximal deterministic strategy, we do
not require any use of memory operator or
subtyping for abstractions to recover subject reduction.

We have also extended our (pure) typing framework
to linear head evaluation, presented in the linear substitution calculus (LSC). 
The result is that tight  typings 
naturally  encode evaluation in the LSC,
which can be seen as the natural 
computing devide behind multi types. In particular, and surprisingly, exact bounds for head and linear head evaluation rely on the same type system. 

Different future directions are suggested by our contribution.
It is natural to extend this framework to other evaluation strategies
such as call-by-value and call-by-need. Richer
programming features such as pattern matching, or control operators
deserve also special attention.

\begin{acks}                            
  We are grateful to Alexis Bernadet, Giulio Guerrieri and anonymous reviewers for useful discussions and comments.
  This work has been partially funded by the ANR JCJC grant COCA HOLA (ANR-16-CE40-004-01).
\end{acks}



\longshort{\newpage
  \appendix
  \section{Appendix: Head and Leftmost Evaluation}

\subsection{Tight Correctness}

\gettoappendix {l:tight-spreading-spi-ske}
\begin{proof}
By induction on $\sysneutral\system\tm$. Cases:
\begin{itemize}
	\item \emph{Variable}, \ie $\tm = \var$. Then $\typctx = \var
          : \mult \type$, and $\type$ is tight because $\typctx$ is
          tight by hypothesis.

	\item \emph{Head application}, \ie $\tm = \tmtwo \tmthree$ and $\sysneutral\hd\tm$ because $\sysneutral\hd\tmtwo$. The
          last rule of $\tderiv$ can only be $\appsteps$ or
          $\appresult<\spi>$. In both cases the left subterm $\tmtwo$ is
          typed by a sub-derivation $\tderiv' \exder \Deri[(\steps',
            \result')] {\typctx_\tmtwo}{\tmtwo}{\typetwo}$ such that
          all types in $\typctx_\tmtwo$ appear in $\typctx$, and so
          they are all tight by hypothesis. Since $\sysneutral\hd\tmtwo$, we can apply the \ih and obtain that
          $\typetwo$ is tight. The only possible case is then  $\typetwo = \neutral$ and
          the last rule of $\tderiv$ is then  $\appresult$. Then $\type =\typetwo=
          \neutral$.
          
          \item \emph{LO application}, \ie $\tm = \tmtwo \tmthree$ and $\sysneutral\lo\tm$ because $\sysneutral\lo\tmtwo$ and $\sysnormal\lo\tmthree$. The reasoning of the previous case applies here too, because it does not depend on the right subterm $\tmtwo$ (the only difference, of course, is that the last rule of $\tderiv$ can only be $\appsteps$ or
          $\appresult<\lo>$). 
    
\end{itemize}
\end{proof}

\gettoappendix {prop:tight-normal-forms-indicies}
\begin{proof}
By induction on $\tm$. Note that  $\sysneutralpr\system$ implies $\sysnormalpr\system$ and so we can apply the \ih when $\sysneutralpr\system$ holds on some subterm of $\tm$. If $\sysnormal\system\tm$ because $\sysneutral\system\tm$ there are three cases:
  \begin{itemize}
    \item \emph{Variable}, \ie $\tm = \var$. Then $\tderiv$ has the following form and evidently verifies all the points of the statement:
 
 $$\infer[\ax]{\Deri[(0, 0)] {\var \col \mult\type} \var \type}{}$$
  
    \item \emph{Head application}, \ie $\tm = \tmtwo \tmthree$ and $\hdneutral\tmtwo$. Cases of the last rule of $\tderiv$:
    \begin{itemize}
    \item \emph{$\appsteps$ rule}: 
    \[\begin{array}{c}
\infer[\appsteps]{\Deri
[(\steps_\tmtwo + \steps_\tmthree + 1, \spine_\tmtwo + \spine_\tmthree)]
{\typctx_\tmtwo \mplus \typctx_\tmthree}
{\tmtwo \tmthree} \type}
{ 	\tderiv_\tmtwo \exder[\hd] { \Deri[(\steps_\tmtwo, \spine_\tmtwo)] {\typctx_\tmtwo} \tmtwo {\M \rightarrow \type}}
\quad
\tderiv_\tmthree \exder[\hd] \Deri[(\steps_\tmthree, \spine_\tmthree)] {\typctx_\tmthree} \tmthree{\M}
}\\\\
\end{array}
\]
with $\steps = \steps_\tmtwo + \steps_\tmthree + 1$, $\spine = \spine_\tmtwo + \spine_\tmthree$,  and $\typctx = \typctx_\tmtwo \mplus \typctx_\tmthree$.
		\begin{enumerate}
		  \item \emph{Size bound}: by \ih, $\hdsize\tmtwo \leq \hdsize{\tderiv_\tmtwo}$, from which it follows $\hdsize\tm = \hdsize\tmtwo +1  \leq_{\ih} \hdsize{\tderiv_\tmtwo} + 1 \leq \hdsize\tderiv$.

		  \item \emph{Tightness}: we show that this case is impossible. If $\tderiv$ is tight then $\typctx = \typctx_\tmtwo \mplus \typctx_\tmthree$ is a tight typing context, and so is $\typctx_\tmtwo$. Since $\sysneutral\hd{\tm}$, the tight spreading on neutral terms (\reflemma{tight-spreading-spi-ske}) implies that the type of $\tmtwo$ in $\tderiv_\tmtwo$ has to be tight---absurd.
		  
		  \item \emph{Neutrality}: $\hdneutral\tm$ holds by hypothesis.
		\end{enumerate}

    \item \emph{$\appresult<\hd>$ rule}: 
    $$
\infer[\appresult]{\Deri[(\steps, \spine_\tmtwo+1)] {\typctx} {\tmtwo \tmthree} \neutype}
{
\tderiv_\tmtwo \exder[\hd] {\Deri[(\steps, \spine_\tmtwo)] \typctx \tmtwo \neutype
}}
$$
with $\spine = \spinetwo +1$. 
		\begin{enumerate}
		  \item \emph{Size bound}: by \ih, $\hdsize\tmtwo \leq \hdsize{\tderiv_\tmtwo}$. Then $\hdsize\tm = \hdsize\tmtwo + 1 \leq_{\ih} \hdsize{\tderiv_\tmtwo} + 1 = \hdsize\tderiv$.

		  \item \emph{Tightness}: if $\tderiv$ is tight, then $\tderiv_\tmtwo$ is tight and by \ih $\spine_\tmtwo = \hdsize\tmtwo$ and $\steps = 0$. Then, $\spine = \spine_\tmtwo +1 =_{\ih} \hdsize\tmtwo + 1 = \hdsize{\tmtwo \tmthree} = \hdsize\tm$.

		  \item \emph{Neutrality}: $\hdneutral\tm$ holds by hypothesis.
		\end{enumerate}
    \end{itemize}
    
    \item \emph{LO application}, \ie $\tm = \tmtwo \tmthree$, $\loneutral\tmtwo$ and $\lonormal\tmthree$. The case is essentially as the previous one, with a minor change in the case of a $\appresult<\lo>$ rule---we spell out all details anyway. Cases of the last rule of $\tderiv$:
    \begin{itemize}
    \item \emph{$\appsteps$ rule}: 
    \[\begin{array}{c}
\infer[\appsteps]{\Deri
[(\steps_\tmtwo + \steps_\tmthree + 1, \spine_\tmtwo + \spine_\tmthree)]
{\typctx_\tmtwo \mplus \typctx_\tmthree}
{\tmtwo \tmthree} \type}
{ 	\tderiv_\tmtwo \exder[\lo] { \Deri[(\steps_\tmtwo, \spine_\tmtwo)] {\typctx_\tmtwo} \tmtwo {\M \rightarrow \type}}
\quad
\tderiv_\tmthree \exder[\lo] \Deri[(\steps_\tmthree, \spine_\tmthree)] {\typctx_\tmthree} \tmthree{\M}
}\\\\
\end{array}
\]
with $\steps = \steps_\tmtwo + \steps_\tmthree + 1$, $\spine = \spine_\tmtwo + \spine_\tmthree$,  and $\typctx = \typctx_\tmtwo \mplus \typctx_\tmthree$.
		\begin{enumerate}
		  \item \emph{Size bound}: by \ih, $\losize\tmtwo \leq \losize{\tderiv_\tmtwo}$, from which it follows $\losize\tm = \losize\tmtwo +1  \leq_{\ih} \losize{\tderiv_\tmtwo} + 1 \leq \losize\tderiv$.

		  \item \emph{Tightness}: we show that this case is impossible. If $\tderiv$ is tight then $\typctx = \typctx_\tmtwo \mplus \typctx_\tmthree$ is a tight typing context, and so is $\typctx_\tmtwo$. Since $\sysneutral\lo{\tm}$, the tight spreading on neutral terms (\reflemma{tight-spreading-spi-ske}) implies that the type of $\tmtwo$ in $\tderiv_\tmtwo$ has to be tight---absurd.
		  
		  \item \emph{Neutrality}: $\loneutral\tm$ holds by hypothesis.
		\end{enumerate}

    \item \emph{$\appresult<\lo>$ rule}: 
    $$
\infer[\appresult]{\Deri[(\steps_\tmtwo + \steps_\tmthree, \spine_\tmtwo + \spine_\tmthree + 1)] {\typctx_\tmtwo \mplus \typctx_\tmthree} {\tmtwo \tmthree} \neutype}
{
\tderiv_\tmtwo \exder[\lo] {\Deri[(\steps_\tmtwo, \spine_\tmtwo)] {\typctx_\tmtwo} \tmtwo \neutype}
\quad
\tderiv_\tmthree \exder[\lo] {\Deri[(\steps_\tmthree, \spine_\tmthree)] {\typctx_\tmthree} \tmthree \tight
}}
$$
with $\steps = \steps_\tmtwo + \steps_\tmthree$, $\spine = \spine_\tmtwo + \spine_\tmthree + 1$,  and $\typctx = \typctx_\tmtwo \mplus \typctx_\tmthree$.
		\begin{enumerate}
		  \item \emph{Size bound}: by \ih, $\losize\tmtwo \leq \losize{\tderiv_\tmtwo}$ and $\losize\tmthree \leq \losize{\tderiv_\tmthree}$. Then $\losize\tm = \losize\tmtwo + \losize\tmthree \leq_{\ih} \losize{\tderiv_\tmtwo} + \losize{\tderiv_\tmthree} = \losize\tderiv$.

		  \item \emph{Tightness}: if $\tderiv$ is tight, then $\tderiv_\tmtwo$ and $\tderiv_\tmthree$ are tight and by \ih $\spine_\tmtwo = \losize\tmtwo$ and $\steps_\tmtwo = 0$, and $\spine_\tmthree = \losize\tmthree$ and $\steps_\tmthree = 0$. Then, $\spine = \spine_\tmtwo + \spine_\tmthree + 1 =_{\ih} \losize\tmtwo + \losize\tmthree + 1 = \losize{\tmtwo \tmthree} = \losize\tm$ and $\steps = \steps_\tmtwo + \steps_\tmthree = 0 + 0 = 0$.
		  
		  \item \emph{Neutrality}: $\loneutral\tm$ holds by hypothesis.
		\end{enumerate}
    \end{itemize}
  \end{itemize}\bigskip

Now, there is only one case left for $\sysnormal\system\tm$:
  \begin{itemize}
    \item \emph{Abstraction}, \ie $\tm = \la\var\tmtwo$ and $\sysnormal\system\tm$ 
 because $\sysnormal\system\tmtwo$. Cases of the last rule of $\tderiv$:
    \begin{itemize}
    \item \emph{$\funsteps$ rule}: 
  $$\infer[\funsteps]{
  \Deri[(\steps_\tmtwo+ 1, \spine)]
{\typctx} { \la\var\tmtwo } {\M \rightarrow \type}
}
{
\tderiv_\tmtwo \exder[\system] {\Deri[(\steps_\tmtwo, \spine)] {\typctx; \var \col \M} \tmtwo \type}
}$$
with $\steps = \steps_\tmtwo+ 1$. 
\begin{enumerate}
  \item \emph{Size bound}: by \ih, $\syssize\system\tmtwo \leq \syssize\system{\tderiv_\tmtwo}$. Then, $\syssize\system\tm = \syssize\system\tmtwo +1 \leq_{\ih} \syssize\system{\tderiv_\tmtwo} + 1 = \syssize\system\tderiv$.

  \item \emph{Tightness}: $\tderiv$ is not tight, so the statement trivially holds.
  
    \item \emph{Neutrality}: $\type \neq \neutral$, so the statement trivially holds.
\end{enumerate}
    \item \emph{$\funresult$ rule}: 
      $$\infer[\funresult]{\Deri[(\steps, \spine_\tmtwo +1)] {\typctx} { \la\var\tmtwo } \abstype}
{\tderiv_\tmtwo \tri {\Deri[(\steps, \spine_\tmtwo)] {\typctx; \var \col \mtight} \tmtwo \tight }}$$
with $\spine = \spine_\tmtwo +1$. 

\begin{enumerate}
  \item \emph{Size bound}: by \ih, $\syssize\system\tmtwo \leq \syssize\system{\tderiv_\tmtwo}$.
Then, $\syssize\system\tm = \syssize\system\tmtwo +1 \leq_{\ih} \syssize\system{\tderiv_\tmtwo} + 1 = \syssize\system\tderiv$.

  \item \emph{Tightness}: if $\tderiv$ is tight, then $\tderivtwo$ is tight and by \ih $\spinetwo = \syssize\system\tmtwo$ and $\steps = 0$. Then, $\spine = \spine_\tmtwo +1 =_{\ih} \syssize\system\tmtwo + 1 = \syssize\system\tm$.
  
      \item \emph{Neutrality}: $\type \neq \neutral$, so the statement trivially holds.
\end{enumerate}
 \end{itemize}
  \end{itemize}
  \end{proof}

\gettoappendix {l:typing-substitution-overview}
\begin{proof}
Let $\tderiv_\tm$ and $\tderiv_\tmtwo$ be the typing derivations of final judgements $\Deri[(\steps', \result')] {\typctxtwo; \var: {\M}} {\tm}{\type}$ and $\Deri[(\steps, \result)]{\typctx} {\tmtwo}{{\M}}$ in system $\system \in \set{\spi, \ske}$. We prove that there exists a typing $\tderiv_{\tm\isub\var\tmtwo} \exder[\system]\Deri[(\steps + \steps', \result + \result')]
          {\typctx\mplus\typctxtwo}{\tm\isub\var\tmtwo }{\type}$. The proof is by induction on $\tderiv_\tm$, distinguishing the two systems only on the rule on which they differ, namely $\appresult<\spi>$ and $\appresult<\ske>$. Let us write $\M$ as $\mult{\typetwo_i}_{\iI}$ for some (potentially empty) set of indices $\indset$. We reason by cases of the last rule of $\tderiv_\tm$:
\begin{itemize}
 \item \emph{Rule $\ax$}. Two cases:
 \begin{enumerate}
   \item $\tm = \var$, and so $\tm \isub\var\tmtwo = \var
     \isub\var\tmtwo = \tmtwo$ and $\tderiv_\tm \exder[\system] {\Deri[(0, 0)]
       {\var \col \single \type} \var \type}$. Thus, $i = 1$ and $\M = \mult{\type}$, and the
     second hypothesis comes from $\tderiv_\tmtwo \exder[\system] \tyjn{(\stepstwo,
         \spinetwo)}{\tmtwo} \typctx \type$ followed
     by a unary $\many$ rule. Given that $\var\isub\var\tmtwo =
     \tmtwo$, $\spine + \spinetwo = 0 +\spinetwo  =
     \spinetwo$, and $\steps +\steps' = \steps'$, the typing derivation $\tderiv_{\tm\isub\var\tmtwo} \defeq \tderiv_\tmtwo$
     satisfies the requirements.
 
  \item $\tm = \vartwo$, and so $\M = \emm$, $\stepstwo = \spinetwo =  0$ and $\tm \isub\var\tmtwo = \vartwo \isub\var\tmtwo = \vartwo$. The it is enough to take $\tderiv_{\tm\isub\var\tmtwo} \defeq \tderiv_\tm$.
 \end{enumerate}
 
 \item \emph{Rule $\funsteps$}. Then $\tm = \la\vartwo\tmthree$, and $\tderiv_\tm$ is such that $\steps = \steps_\tmthree +1$,  and it has the following form:
 \[\begin{array}{c}
 \infer[\funsteps]{
 	\Deri[(\steps_\tmthree+1, \spine)] { \typctxtwo;  \var: \M} {\la\vartwo\tmthree} {\N \rightarrow \typetwo}
	}{
	\tderiv_\tmthree \exder[\system] \Deri[(\steps_\tmthree, \spine)] {\typctxtwo; \var: \M;  \vartwo: \N} \tmthree \typetwo
	} 
 \end{array}\]
By \ih there exists $\tderiv_{\tmthree\isub\var\tmtwo}$ such that 
 \[\begin{array}{c}
\tderiv_{\tmthree\isub\var\tmtwo}\exder[\system]  {\tyjn{ (\steps_\tmthree + \stepstwo, \spine + \spinetwo)}{\tmthree\isub\var\tmtwo }{\typctx \mplus \typctxtwo;  \vartwo: \N}{\typetwo}}
\end{array}\]
from which by applying $\funsteps$ we obtain:
 \[\begin{array}{c}
\tderiv_{\tm\isub\var\tmtwo} \exder[\system] {\tyjn{ (\steps_\tmthree + \stepstwo+1, \spine + \spinetwo)}{\la\vartwo\tmthree\isub\var\tmtwo }{\typctx \mplus \typctxtwo}{\N \rightarrow \typetwo}}
\end{array}\]
that satisfies the requirements because 
\begin{enumerate}
 \item $\steps_\tmthree + \stepstwo+1 = \steps + \stepstwo$,
 \item $\spine + \spinetwo = \spine + \spinetwo$.
 \end{enumerate}

 \item \emph{Rule $\funresult$}. Then $\tm = \la\vartwo\tmthree$, and $\tderiv_\tm$ is such that $\spine = \spinethree+1$ and it has the following form:
 \[\begin{array}{c}
 \infer[\funsteps]{
 	\Deri[(\steps, \spinethree+1)] { \typctxtwo; \var: \M} {\la\vartwo\tmthree} {\abstype}
	}{
	\tderiv_\tmthree \exder[\system] \Deri[(\steps, \spinethree)] {\typctxtwo; \var: \M; \vartwo: \mtight } \tmthree \nf
	} 
 \end{array}\]
By \ih there exists $\tderiv_{\tmthree\isub\var\tmtwo}$ such that
 \[\begin{array}{c}
\tderiv_{\tmthree\isub\var\tmtwo} \exder[\system] {\tyjn{ (\steps + \stepstwo, \spinethree + \spinetwo)}{\tmthree\isub\var\tmtwo }{\typctx \mplus \typctxtwo; \vartwo: \mtight}{\nf}}
\end{array}\]
from which by applying $\funresult$ we obtain:
 \[\begin{array}{c}
\tderiv_{\tm\isub\var\tmtwo} \exder[\system] {\tyjn{ (\steps + \stepstwo , \spinethree + \spinetwo +1)}{\la\vartwo\tmthree\isub\var\tmtwo }{\typctx \mplus \typctxtwo }{\abstype}}
\end{array}\]
that satisfies the requirements because $\spinethree + \spinetwo +1 = \spine + \spinetwo $.

 \item \emph{Rule $\appsteps$}. Then $\tm = \tmthree \tmfour$. The
   left premise of the $\appsteps$ rule in $\tderiv_\tm$ assigns a  type $\tmthree :
   \N \rightarrow \type$ and the right premise is a $\many$ rule
   with $k \defeq \size\N$ premises. The multiset $\M$ assigned to
   $\var$ can be partioned in $k+1$ (potentially empty) multisets
   $\M_1,\dots \M_k$ and $\M_\tmthree$, to be distributed among the
   premises of the $\appsteps$ rule of $\tderiv$ as follows (if k=0
   then the $\many$ rule has no premises):
  \[\small\begin{array}{c}
  \infer[\appsteps]{
  \Deri
 [(\stepsthree + \steps^\circ + 1, \spinethree + \spine^\circ)]
 {\typctxtwo_\tmthree + \typctxtwo_\tmfour; \var: \M}
 {\tmthree \tmfour} \type
 }
 {\tderiv_\tmthree \exder[\system] \Deri[(\stepsthree, \spinethree)] {\typctxtwo_\tmthree; \var: \M_\tmthree} \tmthree {\N \rightarrow \type} 
 \quad
 \infer[\many]{
 \tderiv_\tmfour \exder[\system] \Deri[(\steps^\circ, \spine^\circ)] {\typctxtwo_\tmfour; \var: \M^\circ} \tmfour {\N} 
 }
 {
 \tderiv_\tmfour^j \exder[\system] (\Deri[(\steps_j, \spine_j)] {\typctxtwo_\tmfour^j; \var: \M_j} \tmfour {\typethree_j})_{j=1,\dots, k} 
 }
 }
 \end{array}\]
where the notations satisfy:
 \begin{enumerate}
 \item[] $\steps^\circ = \sum_{j=1}^k \steps_j$, $\spine^\circ = \sum_{j=1}^k \spine_j$,  $\typctxtwo_\tmfour = \mplus_{j=1}^k \typctxtwo_\tmfour^j$, and $\M^\circ = \mplus_{j=1}^k \M_j$,  
   \item[] $\typctxtwo = \typctxtwo_\tmthree \mplus \typctxtwo_\tmfour$,
   \item[] $\steps = \stepsthree + \steps^\circ + 1$, 
   \item[] $\spine = \spinethree + \spine^\circ$, and
   
 \end{enumerate}

Moreover the derivation $\tderivtwo \exder[\system] \tyjn{(\stepstwo, \spinetwo)}{\tmtwo}{\typctx}{\M}$ of the second hypothesis gives rise to a derivation for each one of the multisets $\M_1,\dots \M_k$ and $\M_\tmthree$ in which $\M$ is partitioned. Let them be $ \tderivtwo_\tmthree \exder[\system]\tyjn{(\steps_\tmthree, \spine_\tmthree)}{\tmtwo}{\typctx_\tmthree}{\M_\tmthree}$ and $({\tderivtwo_\tmfour}^j \exder[\system] \tyjn{(\stepstwo_j, \spinetwo_j)}{\tmtwo}{\typctx_j}{\M_j})_{j=1,\ldots,k}$ with
\begin{itemize}
	\item[] $\stepstwo = \steps_\tmthree + \sum_{j=1}^k \stepstwo_j$,
	\item[] $\spinetwo = \spine_\tmthree + \sum_{j=1}^k \spinetwo_j$, and
	\end{itemize}

Now, by \ih we can substitute these derivations $\tderivtwo_\tmthree$ and ${\tderivtwo_\tmfour}^j$ into the premises of the $\appsteps$ rule, obtaining the derivations $\tderiv_{\tmthree \isub\var\tmtwo}$, $\tderiv_{\tmfour \isub\var\tmtwo}^j$, and $\tderivfour_{\tmfour \isub\var\tmtwo}$ such that 
 $$\tderiv_{\tmthree \isub\var\tmtwo} \exder[\system]  \Deri[(\stepsthree + \steps_\tmthree, \spinethree + \spine_\tmthree )] {\typctxtwo_\tmthree \mplus \typctx_\tmthree} {\tmthree \isub\var\tmtwo} {\N \rightarrow \type} $$
$$\infer[\many]{
\tderivfour_{\tmfour \isub\var\tmtwo} \exder[\system] \Deri[(\sum_{j=1}^k(\steps_j + \stepstwo_j), \sum_{j=1}^k(\spine_j + \spinetwo_j ))] {\typctxtwo_\tmfour \mplus \typctx_\tmfour} {\tmfour \isub\var\tmtwo} {\mult{\typethree_j}_{j=1,\ldots,k}} 
}
{
(\tderiv_{\tmfour \isub\var\tmtwo}^j \exder[\system] \Deri[(\steps_j + \stepstwo_j, \spine_j + \spinetwo_j )] {\typctxtwo_\tmfour^j \mplus \typctx_j} {\tmfour \isub\var\tmtwo} {\typethree_j})_{j=1,\ldots,k} 
}$$
where $\typctx_\tmfour$ stands for $\mplus_{j=1}^k \typctx_j$. By applying $\appsteps$ we obtain:
 $$
 \tderivtwo \exder[\system] 
 \Deri[(\steps^*, \spine^*)]
{\typctxtwo_\tmthree \mplus \typctxtwo_\tmfour \mplus \typctx_\tmthree \mplus \typctx_\tmfour}
{\tmthree\isub\var\tmtwo \tmfour\isub\var\tmtwo = (\tmthree \tmfour)\isub\var\tmtwo} \type
$$
 where (remember that $\M$ splits into $\size \N +1$ multisets $\M_\tmthree$, and $\M_j$ with $j= 1,\ldots,k = \size \N$):
 \begin{itemize}
   \item $\typctxtwo_\tmthree \mplus \typctxtwo_\tmfour \mplus \typctx_\tmthree \mplus \typctx_\tmfour = \typctxtwo \mplus \typctx$, 
   \item 
      
   \[\begin{array}{llllll}
\steps^* &= &\stepsthree + \steps_\tmthree + \sum_{j=1}^k(\steps_j + \stepstwo_j) +1\\
& = &
\stepsthree + \underbrace{\sum_{j=1}^k\steps_j}_{\steps^\circ} + \underbrace{\steps_\tmthree + \sum_{j=1}^k \stepstwo_j}_{\stepstwo}  +1 \\
& = &
\stepsthree + \steps^\circ + \stepstwo + 1 \\
& = & 
\underbrace{\stepsthree + \steps^\circ +1}_{\steps}+ \stepstwo & = & \steps + \stepstwo
   \end{array}\]
   
   \item 
   \[\begin{array}{llll}
   \spine^* & = & \spinethree + \spine_\tmthree - \size{\M_\tmthree} + \sum_{j=1}^k(\spine_j + \spinetwo_j - \size{\M_j}) \\
   & = &
   \underbrace{\spinethree + \sum_{j=1}^k\spine_j}_{\spine} + \underbrace{\spine_\tmthree + \sum_{j=1}^k\spinetwo_j}_{\spinetwo} - \underbrace{(\size{\M_\tmthree}+ \sum_{j=1}^k\size{\M_j})}_{\size\M} \\
   & = &
   \spine + \spinetwo - \size\M
   \end{array}\]

 \end{itemize}
 
 \item \emph{Rule $\appresult<\spi>$}. Then $\tm = \tmthree \tmfour$ and $\tderiv_\tm$ is such that $\spine = \spine_\tmthree + 1$ and it has the folliwing form:
  $$ 
 \infer[\appresult<\spi>]{ 
 \Deri[(\steps, \spine_\tmthree+1)] {\typctxtwo; \var: \M} {\tmthree \tmfour} \neutype} 
{\tingD{\tderiv_\tmthree}{
\Deri[(\steps, \spine_\tmthree)] {\typctxtwo; \var: \M} \tmthree \neutype}
}
$$
with $\spine=\spine_\tmthree+1$. By \ih we can substitute $\tderiv_\tmtwo$ into $\tderiv_\tmthree$ obtaining $\tderiv_{\tmthree\isub\var\tmtwo}$ such that 
$$\tderivfour \exder[\spi]
\Deri[(\steps + \stepstwo, \spine_\tmthree + \spinetwo  )] {\typctx\mplus\typctxtwo} {\tmthree \isub\var\tmtwo} \neutype$$
By applying $\appresult<\spi>$ we obtain
$$\tderivthree \exder[\spi]
\Deri[(\steps + \stepstwo, \spine_\tmthree  + \spinetwo + 1)] {\typctx\mplus\typctxtwo} {\tmthree \isub\var\tmtwo \tmfour\isub\var\tmtwo = (\tmthree \tmfour)\isub\var\tmtwo} \neutype$$
that satisfies the requirements because $\spine_\tmthree  + \spinetwo + 1 = \spine + \spinetwo $.

\item \emph{Rule $\appresult<\ske>$}.
Now, $\tm = \tmthree \tmfour$ and $\var: \MSigma {\typetwo_i} {\iI}$ splits into two multisets $\MSigma {\typetwo_i} {i \in I_\tmthree}$ and $\MSigma {\typetwo_i} {i \in I_\tmfour}$ so that $\tderiv$ has the following form:
 \[\small \begin{array}{c}
 \infer[\appresult<\ske>]{
 \Deri
[(\stepstwo + \stepsthree, \resulttwo + \resultthree + 1)]
{\typctxtwo_\tmthree \mplus \typctxtwo_\tmfour; \var: \MSigma {\typetwo_i} {\iI}}
{\tmthree \tmfour} \neutype
}
{
\tderiv_\tmthree \exder[\ske] \Deri[(\stepstwo, \resulttwo)] {\typctxtwo_\tmthree; \var: \MSigma {\typetwo_i} {i \in I_\tmthree}} \tmthree {\neutype} 
\quad
\tderiv_\tmfour \exder[\ske] \Deri[(\stepsthree, \resultthree)] {\typctxtwo_\tmfour; \var: \MSigma {\typetwo_i} {i \in I_\tmfour}} \tmfour {\nf} 
}
 \end{array}\]
 with 
 \begin{enumerate}
   \item[] $\steps = \stepstwo + \stepsthree$,
   \item[] $\result = \resulttwo + \resultthree +1 $, and
   \item[] $\typctxtwo = \typctxtwo_\tmthree \mplus \typctxtwo_\tmfour$.
 \end{enumerate}

 By \ih there exist $\tderivtwo_\tmthree$ and $\tderivtwo_\tmfour$ such that for appropriate $\typctx_\tmthree$, $\typctx_\tmfour$, $\steps_i$ and $\result_i$, we have:
 $$\tderivtwo_\tmthree \exder[\ske] \Deri[(\stepstwo +_{\iI_\tmthree} \steps_i, \resulttwo +_{\iI_\tmthree} \result_i, )] {\typctxtwo_\tmthree \mplus \typctx_\tmthree} {\tmthree \isub\var\tmtwo} {\neutype} $$
 $$\tderivtwo_\tmfour \exder[\ske] \Deri[(\stepsthree +_{\iI_\tmfour} \steps_i, \resultthree +_{\iI_\tmfour} \result_i)] {\typctxtwo_\tmfour \mplus \typctx_\tmfour} {\tmfour \isub\var\tmtwo} {\nf} $$
 from which by applying $\appresult<\ske>$ we obtain:
 $$
 \tderivtwo \exder[\ske] 
 \Deri[(\steps^*, \result^*)]
{\typctx \mplus \typctxtwo}
{\tmthree\isub\var\tmtwo \tmfour\isub\var\tmtwo = (\tmthree \tmfour)\isub\var\tmtwo} \neutype$$
 where:
 \begin{enumerate}
   \item[] $\steps^* = \stepstwo +_{\iI_\tmthree} \steps_i + \stepsthree +_{\iI_\tmfour} \steps_i = \underbrace{\stepstwo + \stepsthree}_{\steps} +_{\iI} \steps_i = \steps +_{\iI} \steps_i$, and
   \item[] $\result^* = \resulttwo +_{\iI_\tmthree} \result_i + \resultthree +_{\iI_\tmfour} \result_i  = \underbrace{\resulttwo + \resultthree + 1}_{\result} +_{\iI} \result_i = \result +_{\iI} \result_i$.
 \end{enumerate}
\end{itemize}

\end{proof}

\gettoappendix {prop:subject-reduction}

\begin{proof}
  We prove, by induction on $\tm\Rew{\system}\tmtwo$,
  the stronger statement:
  
  Assume
  $\tm\Rew{\system}\tmtwo$,
  $\tderiv\exder[\system] \Deri[(\steps, \result)]\typctx{\tm}\type$,
  $\tightpred\typctx$,
  and either $\tightpred\type$ or $\sysnotabs \system \tm$.
  
  Then there exists a typing
  $\tderivtwo\exder[\system] \Deri[(\steps-2, \result)]{\typctx}{\tmtwo}\type$.

  We prove this stronger statement by induction on $\tm\Rew{\system}\tmtwo$. Cases:
  \begin{itemize}
  \item Rule
    \[\infer{
      (\la\var \tmthree) \tmfour \Rew{\system} \tmthree \isub \var \tmfour
    }{}
    \]
    Assume 
    $\tderiv \exder[\system]\Deri[(\steps,\result)]
    {\typctx}{(\la\var \tmthree) \tmfour}\type$
    and $\tightpred\typctx$.
    The derivation $\tderiv$ must end with rule $\appsteps$,
    and the derivation of its premiss for $(\la\var \tmthree)$
    must end with $\funsteps$.
    Hence, there are two derivations
    $\tderiv_\tmthree\exder[\system]\Deri[
      (\steps_\tmthree, \result_\tmthree)
    ]{\typctx_\tmthree; \var: {\M}}{\tmthree}{\type}$
    and
    $\tderiv_\tmtwo\exder[\system]\Deri[
      (\steps_\tmtwo, \result_\tmtwo)
    ]{\typctx_\tmtwo}{\tmtwo}{\M}$,
    with $(\steps,\result) = (\steps_u+\steps_\tmfour+2, \result_u+\result_\tmfour)$
    and $\typctx=\typctx_\tmthree\mplus\typctx_\tmtwo$.
   Applying the substitution lemma (\reflemma{typing-substitution-overview}), we obtain $\tderiv'\exder[\system]\Deri[(\steps_u+\steps_\tmfour,\result_u+\result_\tmfour)]
    {\typctx}{\tmthree \isub \var \tmfour}\type$.

  \item Rule
    \[
    \infer{\la\var \tm \Rew{\system} \la\var \tmtwo}{
      \tm \Rew{\system} \tmtwo
    }
    \]
    Assume
    $\tderiv\exder[\system]\Deri[(\steps,\result)]
    {\typctx}{\la\var \tm}\type$
    and $\tightpred\typctx$.
    Since $\sysabs \system {\la\var \tm}$ we must have hypothesis $\tightpred\type$,
    and as $\tderiv$ must then finish with rule $\funresult$
    we must have a subderivation
    $\tderiv_\tm \exder[\system]
    \Deri[
      (\steps, \result-1)
    ]{\typctx, \var \col \mtight}{\tm}{\nf}$.
    As $\tightpred{\typctx, \var \col \mtight}$
    we can apply the \ih and get the premise of the derivation $\tderiv'$ below:
    \[
    \infer{
      \Deri[(\steps+2, \result)]{
        \typctx
      }{\la\var\tmtwo}\abstype
    }{
      \tderiv_\tmtwo \exder[\system]
      \Deri[(\steps+2, \result-1)]{\typctx, \var \col \mtight}{\tmtwo}\nf
    }
    \]

  \item Rule
    \[
    \infer{
      \tm \tmthree \Rew{\system} \tmtwo \tmthree
    }{
      \sysnotabs \system \tm \quad \tm \Rew{\system}\tmtwo 
    }
    \]
    Assume
    $\tderiv\exder[\system]\Deri[(\steps,\result)]
    {\typctx}{\tm \tmthree}\type$
    and $\tightpred\typctx$.
    The derivation $\tderiv$ must end
    with rule $\appsteps$ or $\appresult<\system>$.
    In the simple case of rule $\appresult<\spi>$,
    there is a subderivation
    $\tderiv_\tm\exder[\spi]\Deri[
      (\steps, \result-1)
    ]{\typctx}{\tm}{\neutype}$
    in $\tderiv$,
    and we can apply the \ih to get the premiss of the derivation $\tderiv'$ below:
    \[
    \infer{
      \Deri[
        (\steps-2, \result)
      ]{\typctx}{\tmtwo \tmthree}{\type}
    }{
      \tderiv_\tmtwo\exder[\spi]\Deri[
        (\steps-2, \result-1)
      ]{\typctx}{\tmtwo}{\neutype}
    }
    \]
    In the case of $\appsteps$ or $\appresult<\ske>$,
    there are derivations
    $\tderiv_\tm\exder[\system]\Deri[
      (\steps_\tm, \result_\tm)
    ]{\typctx_\tm}{\tm}{\type_\tm}$
    and
    $\tderiv_\tmthree\exder[\system]\Deri[
      (\steps_\tmthree, \result_\tmthree)
    ]{\typctx_\tmthree}{\tmthree}{\type_\tmthree}$,
    with $\typctx=\typctx_\tm\mplus\typctx_\tmthree$.
    Since $\tightpred{\typctx}$ we have $\tightpred{\typctx_\tm}$,
    and since $\sysnotabs \system \tm$
    we can apply the \ih to get the derivation
    $\tderiv_\tmtwo \exder[\system]
    \Deri[
      (\steps_\tm-2, \result_\tm)
    ]{\typctx_\tm}{\tmtwo}{\type_\tm}$ and build,
    using the same rule $\appsteps$ or $\appresult<\ske>$,
    the derivation $\tderiv'$ below:
    \[
    \infer{
      \Deri[
        (\steps-2, \result)
      ]{\typctx}{\tmtwo \tmthree}{\type}
    }{
      \tderiv_\tmtwo \exder[\system]
    \Deri[
      (\steps_\tm-2, \result_\tm)
    ]{\typctx_\tm}{\tmtwo}{\type_\tm}
    \quad
    \tderiv_\tmthree\exder[\system]\Deri[
      (\steps_\tmthree, \result_\tmthree)
    ]{\typctx_\tmthree}{\tmthree}{\type_\tmthree}
    }
    \]
    
  \item Rule
    \[
    \infer{
      \tmthree \tm \Rew{\ske}\tmthree \tmtwo
    }{
      \skneutral \tmthree \quad  \tm \Rew{\ske}\tmtwo
    }
    \]
    Assume
    $\tderiv\exder[\ske]\Deri[(\steps,\result)]
    {\typctx}{\tmthree \tm}\type$
    and $\tightpred\typctx$.
    The derivation $\tderiv$ must end
    with rule $\appsteps$ or $\appresult<\ske>$,
    and therefore there are two derivations
    $\tderiv_\tmthree\exder[\ske]\Deri[
      (\steps_\tmthree, \result_\tmthree)
    ]{\typctx_\tmthree}{\tmthree}{\type_\tmthree}$
    and
    $\tderiv_\tm\exder[\ske]\Deri[
      (\steps_\tm, \result_\tm)
    ]{\typctx_\tm}{\tm}{\type_\tm}$,
    for some types $\type_\tmthree$ and $\type_\tm$,
    with $\typctx=\typctx_\tmthree\mplus\typctx_\tm$.
    Since $\tightpred{\typctx}$ we have
    $\tightpred{\typctx_\tmthree}$ and $\tightpred{\typctx_\tm}$.
    Theorem~\ref{l:tight-spreading-spi-ske}
    concludes $\tightpred{\type_\tmthree}$ from $\skneutral \tmthree$.
    So the last rule of $\tderiv$ must be $\appresult<\ske>$,
    whence $\type=\neutype$ and $\type_\tm=\tight$.
    Therefore we can apply the \ih to get the derivation
    $\tderiv_\tmtwo \exder[\ske]
    \Deri[
      (\steps_\tm-2, \result_\tm)
    ]{\typctx_\tm}{\tmtwo}{\type_\tm}$
    and build,
    using the same rule $\appresult<\ske>$,
    the derivation $\tderiv'$ below:
    \[
    \infer{
      \Deri[
        (\steps-2, \result)
      ]{\typctx}{\tmthree\tmtwo}{\type}
    }{
    \tderiv_\tmthree\exder[\system]\Deri[
      (\steps_\tmthree, \result_\tmthree)
    ]{\typctx_\tmthree}{\tmthree}{\type_\tmthree}
    \quad
    \tderiv_\tmtwo \exder[\system]
    \Deri[
      (\steps_\tm-2, \result_\tm)
    ]{\typctx_\tm}{\tmtwo}{\type_\tm}
    }
    \]
  \end{itemize}
\end{proof}

\gettoappendix {thm:correctness}
\begin{proof} 
By induction on $\syssize\system\tderiv$---the proof is modular in $\system \in \set{\hd,\lo}$. If $\tm$ is a $\tostrat$ normal form  then by taking $\tmtwo \defeq \tm$ and 
$k=0$ the statement follows from the \emph{tightness} property of tight typings of normal forms (\refprop{tight-normal-forms-indicies}.2)---the \emph{moreover} part follows from the \emph{neutrality} property (\refprop{tight-normal-forms-indicies}.3). Otherwise, $\tm \tostrat \tmthree$ and by quantitative subject reduction (\refprop{subject-reduction}) there is a derivation $\tderivtwo \exder[\system] \Deri[(\steps-2, \spine)] {\typctx}{\tmthree}{\type}$. By \ih, there exists $\tmtwo$ such that $\sysnormal\system \tmtwo$ and $\tmthree
  \Rewn[(\steps-2)/2]{\system} \tmtwo$ and $\syssize\system\tmtwo = \result$. Just note that $\tm \tostrat \tmthree
  \Rewn[\steps/2-1]{\system} \tmtwo$, that is, $\tm 
  \Rewn[\steps/2]{\system} \tmtwo$. 
\end{proof}

\subsection{Tight Completeness}

\gettoappendix {prop:normal-forms-are-tightly-typable}
\begin{proof}

By induction on $\sysnormal\system\tm$. Cases:
  \begin{enumerate}
    \item \emph{Variable}, \ie $\tm = \var$. Then the following derivation evidently satisfies all points of the statement:
 $$\infer[\ax]{\Deri[(0,0)] {\var \col \single \neutype} \var \neutype}{}$$
  
    \item \emph{Abstraction}, \ie $\tm = \la\vartwo\tmtwo$ with $\hdnormal{\tmtwo}$. By \ih
      there is a tight derivation $\tderiv_\tmtwo \exder[\system] \Deri[(0,\syssize\system\tmtwo)] \typctxtwo \tmtwo
      \nf$. Since the derivation $\tderiv_\tmtwo$ is tight, the typing context $\typctxtwo$ has the shape $\typctx; \vartwo \col \mtight$ (potentially, $\vartwo \col \emptymset$). Then the following is a \precise derivation for
      $\la\vartwo\tmtwo$:
 $$\infer[\funresult]{\Deri[(0,\syssize\system\tmtwo+1)] {\typctx} {\la\vartwo\tmtwo} \abstype}
{\tderiv_\tmtwo   \exder[\system] \Deri[(0,\syssize\system\tmtwo)] {\typctx; \vartwo \col \mtight} \tmtwo \nf } $$
Moreover, $\tm$ is not neutral so the part about neutral terms is trivially true, while it is an abstraction and it is indeed typed with $\abstype$.

\item \emph{Application}, \ie $\tm = \tmtwo \tmthree$. Two sub-cases, depending on $\system \in \set{\hd,\lo}$.

\begin{itemize}
\item \emph{Head normal form}: $\hdnormal \tm$ implies $\hdneutral \tm$, that implies $\hdneutral \tmtwo$, that implies $\hdnormal \tmtwo$. By \ih, there is a \precise derivation 
$\tderivtwo \exder[\hd] \Deri[(0,\hdsize\tmtwo)] \typctx \tmtwo \nf$ typing $\tmtwo$ with $\neutype$. Then the following is a \precise derivation $\tderiv$ types $\tm = \tmtwo \tmthree$ with $\neutype$, and having as second index satisfies $\hdsize\tmtwo + 1 = \hdsize{\tmtwo\tmthree} = \hdsize\tm$, as required:
$$
\infer[\appresult<\hd>]{\Deri[(0,\hdsize\tmtwo+1)] {\typctx} {\tmtwo \tmthree} \neutype}
{\tderivtwo \exder[\hd] \Deri[(0,\hdsize\tmtwo)] \typctx \tmtwo \neutype}
$$
Moreover, $\hdneutral{\tm}$ and $\tderiv$ does indeed type $\tm$ with $\neutype$. Dually, $\tm$ is not an abstraction and so that point trivially holds.

\item \emph{LO normal form}: $\lonormal \tm$ implies $\loneutral \tm$, that implies $\loneutral \tmtwo$ and $\lonormal \tmthree$, and the first implies $\lonormal \tmtwo$. By \ih, there are \precise derivations 
\begin{itemize}
\item $\tderiv_\tmtwo \exder[\lo] \Deri[(0,\losize\tmtwo)] {\typctx_\tmtwo} \tmtwo \nf$ typing $\tmtwo$ with $\neutype$ (because $\loneutral{\tmtwo}$), and
\item $\tderiv_\tmthree \exder[\lo] \Deri[(0,\losize\tmthree)] {\typctx_\tmthree} \tmthree \nf$.
\end{itemize}
Then the following is a \precise derivation $\tderiv$ for $\tm = \tmtwo \tmthree$ whose second index satisfies $\losize\tmtwo + \losize\tmthree + 1 = \losize\tm$, as required:
$$
\infer[\appresult<\lo>]{\Deri[(0,\losize\tmtwo + \losize\tmthree + 1)] {\typctx_\tmtwo \mplus \typctx_\tmthree} {\tmtwo \tmthree} \neutype}
{\tderiv_\tmtwo \exder[\lo] \Deri[(0,\losize\tmtwo)] {\typctx_\tmtwo} \tmtwo \nf
\quad
\tderiv_\tmthree \exder[\lo] \Deri[(0, \losize\tmthree)] {\typctx_\tmthree} \tmtwo \nf
}\\\\
$$
Moreover, $\loneutral{\tm}$ and $\tderiv$ does indeed type $\tm$ with $\neutype$. Dually, $\tm$ is not an abstraction and so that point trivially holds.

  \end{itemize}
  \end{enumerate}
  
\end{proof}

\gettoappendix {l:anti-substitution}
\begin{proof}
By induction on $\tm$. Cases:
\begin{itemize}
  \item   
  \emph{Variable}, \ie $\tm = \vartwo$. Two subcases, depending on the identity of $\vartwo$:
  \begin{enumerate}
   \item $\var = \vartwo$. Then $\tm\isub\var\tmtwo = \var\isub\var\tmtwo = \tmtwo$, so that $\tderiv \exder[\system] \tyjn{(\steps, \spine)}{\tmtwo} \typctx \type$. There is only one possibility: $\size \M = 1$, $\tderiv_\tmtwo$ is 
   \[\begin{array}{c}
 \infer[\many]{
 	\tyjn{(\steps, \spine)}{\tmtwo} \typctx {\mult\type}
	}{
	\tderiv \exder[\system] \tyjn{(\steps, \spine)}{\tmtwo} \typctx \type
	} 
 \end{array}\]
 and $\tderiv_\tm$ is $$\infer[\ax]{\Deri[(0, 0)] {\var \col \single \type} \var \type}{}$$
   
    \item $\var \neq \vartwo$. Then $\tm\isub\var\tmtwo = \vartwo\isub\var\tmtwo = \vartwo$. There is only one possibility: $\size \M = 0$, $\tderiv_\tm$ is exactly $\tderiv$, that is, 
    $$\infer[\ax]{\Deri[(0, 0)] {\vartwo \col \single \type} \vartwo \type}{}$$
    and $\tderiv_\tmtwo$ is
    \[\begin{array}{c}
      \infer[\many]{
 	\tyjn{(0, 0)}{\tmtwo} {} {\emptymset}
	}{
	
	} 
 \end{array}\]
  \end{enumerate}
 
  \item \emph{Abstraction}, \ie $\tm = \la\vartwo\tmthree$. Then $ \tm \isub\var\tmtwo = \la\vartwo \tmthree \isub\var\tmtwo $. Two sub-cases, depending on the last rule of $\tderiv$:
  \begin{enumerate}
    \item \emph{Rule $\funsteps$}. Then $\tderiv$ has the following form:
 \[\begin{array}{c}
 \infer[\funsteps]{
 	\Deri[(\steps_{\tmthree \isub\var\tmtwo}+1, \spine)] { \typctx} {\la\vartwo\tmthree\isub\var\tmtwo } 
        {\mtypetwo \rightarrow \typefour}
	}{
	\tderiv_\tmtwo \exder[\system] \Deri[(\steps_{\tmthree \isub\var\tmtwo}, \spine)]  {\typctx;  \vartwo: \mtypetwo  } {\tmthree\isub\var\tmtwo} \typefour
	} 
 \end{array}\]
 with $\steps = \steps_{\tmthree \isub\var\tmtwo} +1$. By \ih there exist $\mtype$ and typing derivations 
$\tderiv_\tmthree \exder[\system] \tyjn{(\steps_\tmthree, \spine_\tmthree)} \tmthree {\typctxtwo_\tmthree;  \vartwo: \mtypetwo;  \var \col \mtype} \type$ 
and 
$\tderiv_\tmtwo \exder[\system] \tyjn{(\steps_\tmtwo, \spine_\tmtwo)}{\tmtwo } {\typctxtwo_\tmtwo} \mtype$
such that:
\begin{itemize}
\item \emph{Typing context}: $(\typctx;  \vartwo: \mtypetwo) = (\typctxtwo_\tmthree;  \vartwo: \mtypetwo \mplus \typctxtwo_\tmtwo)$;

\item \emph{Indices}: $(\steps_{\tmthree \isub\var\tmtwo}, \spine) = (\steps_\tmthree + \steps_\tmtwo, \spine_\tmthree + \spine_\tmtwo)$.
\end{itemize}
  Then the derivation $\tderiv_\tm$ defined as 
 \[\begin{array}{c}
 \infer[\funsteps]{
 	\Deri[(\steps_\tmthree+1, \spine_\tmthree)] { \typctx ;  \var: \mtype } {\la\vartwo\tmthree} {\mtypetwo \rightarrow \typefour}
	}{
	\tderiv_\tmthree \exder[\system] \Deri[(\steps_\tmthree, \spine_\tmthree)]  {\typctx;   \vartwo: \mtypetwo;  \var: \mtype  } \tmthree \typefour
	} 
 \end{array}\]
satisfies the statement with respect to $\steps_\tm \defeq \steps_\tmthree+1$ and $\spine_\tm \defeq \spine_\tmthree$
because:

\begin{itemize}
\item \emph{Typing context}: the \ih implies $\typctx = (\typctxtwo_\tmthree \mplus \typctxtwo_\tmtwo)$;
\item \emph{Indices}:
\begin{enumerate}
 \item $\steps_\tm + \steps_\tmtwo = \steps_\tmthree + 1 + \steps_\tmtwo =_{\ih} \steps_{\tmthree \isub\var\tmtwo} + 1 = \steps$,
 \item $\spine_\tm + \spine_\tmtwo = \spine_\tmthree + \spine_\tmtwo  =_{\ih} \spine_{\tmthree \isub\var\tmtwo} = \spine$.
\end{enumerate}
\end{itemize}

 \item \emph{Rule $\funresult$}. Then $\tderiv$ has the following form:
 \[\begin{array}{c}
 \infer[\funresult]{
 	\Deri[(\steps, \spine_{\tmthree\isub\var\tmtwo}+1)] { \typctx} {\la\vartwo\tmthree\isub\var\tmtwo } \abstype
	}{
	\tderiv_\tmtwo \exder[\system] \Deri[(\steps, \spine_{\tmthree\isub\var\tmtwo})]  {\typctx;  \vartwo: \mneutral} {\tmthree\isub\var\tmtwo} \nf
	} 
 \end{array}\]
 with $\spine = \spine_{\tmthree\isub\var\tmtwo} +1$. By \ih there exist $\mtype$ and typing derivations 
$\tderiv_\tmthree \tri \tyjn{(\steps_\tmthree, \spine_\tmthree)} \tmthree {\typctxtwo_\tmthree;  \vartwo: \mneutral;  \var \col \mtype} \nf$ 
and 
$\tderiv_\tmtwo \tri \tyjn{(\steps_\tmtwo, \spine_\tmtwo)}{\tmtwo } {\typctxtwo_\tmtwo} \mtype$
such that:
\begin{itemize}
\item \emph{Typing context}: $(\typctx; \vartwo: \mneutral) = (\typctxtwo_\tmthree; \vartwo: \mneutral  \mplus \typctxtwo_\tmtwo)$;

\item \emph{Indices}: $(\steps, \spine_{\tmthree\isub\var\tmtwo}) = (\steps_\tmthree + \steps_\tmtwo, \spine_\tmthree + \spine_\tmtwo)$.
\end{itemize}
  Then the derivation $\tderiv_\tm$ defined as 
 \[\begin{array}{c}
 \infer[\funsteps]{
 	\Deri[(\steps_\tmthree, \spine_\tmthree +1)] { \typctx; \var: \mtype } {\la\vartwo\tmthree} {\abstype}
	}{
	\tderiv_\tmthree \exder[\system] \Deri[(\steps_\tmthree, \spine_\tmthree)]  {\typctx;   \vartwo: \mneutral; \var: \mtype  } \tmthree \nf
	} 
 \end{array}\]
satisfies the statement with respect to $\spine_\tm \defeq \spine_\tmthree + 1$ because:
\begin{itemize}
\item \emph{Typing context}: the \ih implies $\typctx = (\typctxtwo_\tmthree \mplus \typctxtwo_\tmtwo)$
\item \emph{Indices}:
\begin{enumerate}
 \item $\steps_\tmthree + \steps_\tmtwo =_{\ih} \steps$,
 \item $\spine_\tm + \spine_\tmtwo = \spine_\tmthree + 1 + \spine_\tmtwo  =_{\ih} \spine_{\tmthree\isub\var\tmtwo} + 1 = \spine$.
 \end{enumerate}
\end{itemize}

 \end{enumerate}
 
 \item \emph{Application}, \ie $\tm = \tmthree \tmfour$. Then $ \tm \isub\var\tmtwo = \tmthree \isub\var\tmtwo  \tmfour \isub\var\tmtwo $. Three sub-cases, depending on the last rule of $\tderiv$: 
 \begin{enumerate}
 \item \emph{Rule $\appsteps$}. 
 Let $\tm = \tmthree \tmfour$ so that
$ \tm \isub\var\tmtwo = \tmthree \isub\var\tmtwo  \tmfour \isub\var\tmtwo $.
Then $\tderiv$ has the following form:

 \[\begin{array}{c}
 \infer[\appsteps]{\Deri[(\steps_1 + \steps_2 +1, \spine_1 + \spine_2)]
                  {\typctx_1 \uplus  \typctx_2}
                  {\tmthree\isub\var\tmtwo \tmfour\isub\var\tmtwo} {\type}
                  }
{
  { \tderiv_{\tmthree \isub\var\tmtwo} \exder[\system] \Deri[(\steps_1, \spine_1)]
                             {\typctx_1}
                             {\tmthree \isub\var\tmtwo}
                             {\tarrow{\M}{\type} }  }
\quad
 { \tderiv_{\tmfour \isub\var\tmtwo} \exder[\system]\Deri[(\steps_2, \spine_2)] {\typctx_2} {\tmfour \isub\var\tmtwo} {\M} }
}
 \end{array}\]
 with $\typctx = \typctx_1 \uplus  \typctx_2$, 
    $\steps = \steps_1 + \steps_2+1$, and
    $\spine = \spine_1 + \spine_2$. 
  
 By \ih applied to $\tmthree \isub\var\tmtwo$ and $\tmfour \isub\var\tmtwo$, there exist (disjoint) finite sets $\M_\tmthree$ and $\M_\tmfour$ and typing derivations:
 $$\tderiv_\tmthree \exder[\system] \Deri[(\steps_\tmthree, \spine_\tmthree)]
   {\typctxtwo_\tmthree; \var: \M_\tmthree} {\tmthree} {\tarrow{\M}{\type}} $$
$$\tderiv_\tmfour \exder[\system] \Deri[(\steps_\tmfour, \spine_\tmfour)] {\typctxtwo_\tmfour; \var: \M_\tmfour} {\tmfour} {\M} $$
$$\tderiv_\tmtwo^\tmthree \exder[\system] \tyjn{(\steps_\tmtwo^\tmthree , \spine_\tmtwo^\tmthree )}{\tmtwo } {\typctxthree_\tmthree} {\M_\tmthree}$$
$$\tderiv_\tmtwo^\tmfour \exder[\system] \tyjn{(\steps_\tmtwo^\tmfour, \spine_\tmtwo^\tmfour)}{\tmtwo } {\typctxthree_\tmfour} {\M_\tmfour}$$
 such that:
 \begin{itemize}
 \item \emph{Typing context}: $\typctx_1 = \typctxtwo_\tmthree \uplus \typctxthree_\tmthree$ and
   $\typctx_2 = \typctxtwo_\tmfour \uplus  \typctxthree_\tmfour$. 
 \item \emph{Indices}:
   $(\steps_1, \spine_1) = (\steps_\tmthree+\steps_\tmtwo^\tmthree, \spine_\tmthree+\spine_\tmtwo^\tmthree)$
   and $(\steps_2, \spine_2) = (\steps_\tmfour+\steps_\tmtwo^\tmfour, \spine_\tmfour+\spine_\tmtwo^\tmfour)$.
 \end{itemize}

The derivations $\tderiv_\tmtwo^\tmthree$ and $\tderiv_\tmtwo^\tmfour$ can be summed (by inverting their $\many$ final rule and reapplying a many rule to the union of the premises) obtaining a derivation
$\tderiv_\tmtwo \exder[\system] \tyjn{(\steps_\tmtwo, \spine_\tmtwo)}{\tmtwo } {\typctxthree} {\M }$,  
where $\typctxthree = \typctxthree_\tmthree \uplus \typctxthree_\tmfour$ and
$\steps_\tmtwo = \steps_\tmtwo^\tmthree+\steps_\tmtwo^\tmfour$ and
$\spine_\tmtwo = \spine_\tmtwo^\tmthree+\spine_\tmtwo^\tmfour$
and $\M = \M_\tmthree + \M_\tmfour $. 
We  then apply  $\appsteps$ to obtain the following derivation $\tderiv'$: 
 \[\begin{array}{c}
 \infer[\appresult]{
 \Deri
[(\steps_\tmthree + \steps_\tmfour+1, \spine_\tmthree + \spine_\tmfour)]
{\typctxtwo_\tmthree \uplus \typctxtwo_\tmfour; \var: \M_\tmthree + \M_\tmfour }
{\tmthree \tmfour} \nf}
{ \tderiv_\tmthree \exder[\system] \Deri[(\steps_\tmthree, \spine_\tmthree)]
   {\typctxtwo_\tmthree; \var: \M_\tmthree} {\tmthree} {\tarrow{\M}{\type}} \quad \quad \tderiv_\tmfour \exder[\system] \Deri[(\steps_\tmfour, \spine_\tmfour)] {\typctxtwo_\tmfour; \var: \M_\tmfour} {\tmfour} {\M}}
 \end{array}\]

 We let $\typctxtwo :=\typctxtwo_\tmthree \uplus \typctxtwo_\tmfour $,
 $\steps_\tm:=\steps_\tmthree+\steps_\tmfour $ and $\spine_\tm:=\spine_\tmthree+\spine_\tmfour+1$
 and then conclude because of the following statements: 
 \begin{enumerate}
 \item \emph{Typing context}: $\typctx = \typctx_1 \uplus  \typctx_2 =
   \typctxtwo_\tmthree \uplus \typctxthree_\tmthree \uplus
   \typctxtwo_\tmfour \uplus  \typctxthree_\tmfour = \typctxtwo \uplus \typctxthree$.
   \item \emph{Indices}: 
     $(\steps, \spine) = (\steps_1+\steps_2, \spine_1+
     \spine_2+1) = (\steps_\tm+\steps_\tmtwo, \spine_\tm+\spine_\tmtwo)$.
   \end{enumerate}

  \item \emph{Rule $\appresult<\lo>$}. Let $\tm = \tmthree \tmfour$ so that
$ \tm \isub\var\tmtwo = \tmthree \isub\var\tmtwo  \tmfour \isub\var\tmtwo $.
Then $\tderiv$ has the following form:

 \[\begin{array}{c}
 \infer[\appresult]{\Deri[(\steps_1 + \steps_2 , \spine_1 + \spine_2  + 1)]
                  {\typctx_1 \uplus  \typctx_2}
                  {\tmthree\isub\var\tmtwo \tmfour\isub\var\tmtwo} \nf
                  }
{
  { \tderiv_{\tmthree \isub\var\tmtwo} \exder[\system] \Deri[(\steps_1, \spine_1)]
                             {\typctx_1}
                             {\tmthree \isub\var\tmtwo}
                             {\neutype}  }
\quad
 { \tderiv_{\tmfour \isub\var\tmtwo} \exder[\system] \Deri[(\steps_2, \spine_2)] {\typctx_2} {\tmfour \isub\var\tmtwo} {\nf} }
}
 \end{array}\]
 with $\typctx = \typctx_1 \uplus  \typctx_2$, 
    $\steps = \steps_1 + \steps_2$, 
    $\spine = \spine_1 + \spine_2+1$. 
  
 By \ih applied to $\tmthree \isub\var\tmtwo$ and $\tmfour \isub\var\tmtwo$,  there exist (disjoint) finite sets $\M_\tmthree$ and $\M_\tmfour$ and typing derivations:
 $$\tderiv_\tmthree \exder[\lo]  \Deri[(\steps_\tmthree, \spine_\tmthree)]
   {\typctxtwo_\tmthree; \var: \M_\tmthree} {\tmthree} {\neutype} $$
$$\tderiv_\tmfour \exder[\lo]  \Deri[(\steps_\tmfour, \spine_\tmfour)] {\typctxtwo_\tmfour; \var: \M_\tmfour} {\tmfour} {\nf} $$
$$\tderiv_\tmtwo^\tmthree \exder[\lo] \tyjn{(\steps_\tmtwo^\tmthree , \spine_\tmtwo^\tmthree )}{\tmtwo } {\typctxthree_\tmthree} {\M_\tmthree}$$
$$\tderiv_\tmtwo^\tmfour \exder[\lo] \tyjn{(\steps_\tmtwo^\tmfour, \spine_\tmtwo^\tmfour)}{\tmtwo } {\typctxthree_\tmfour} {\M_\tmfour}$$
 such that:
 \begin{itemize}
 \item \emph{Typing context}: $\typctx_1 = \typctxtwo_\tmthree \uplus \typctxthree_\tmthree$ and
   $\typctx_2 = \typctxtwo_\tmfour \uplus  \typctxthree_\tmfour$. 
 \item \emph{Indices}:
   $(\steps_1, \spine_1) = (\steps_\tmthree+\steps_\tmtwo^\tmthree, \spine_\tmthree+\spine_\tmtwo^\tmthree)$
   and $(\steps_2, \spine_2) = (\steps_\tmfour+\steps_\tmtwo^\tmfour, \spine_\tmfour+\spine_\tmtwo^\tmfour)$.
 \end{itemize}

The derivations $\tderiv_\tmtwo^\tmthree$ and $\tderiv_\tmtwo^\tmfour$ can be summed (by inverting their $\many$ final rule and reapplying a many rule to the union of the premises) obtaining a derivation
$\tderiv_\tmtwo \exder[\lo] \tyjn{(\steps_\tmtwo, \spine_\tmtwo)}{\tmtwo } {\typctxthree} {\M }$,  
where $\typctxthree = \typctxthree_\tmthree \uplus \typctxthree_\tmfour$ and
$\steps_\tmtwo=\steps_\tmtwo^\tmthree+\steps_\tmtwo^\tmfour$ and
$\spine_\tmtwo= \spine_\tmtwo^\tmthree+\spine_\tmtwo^\tmfour$
and $\M = \M_\tmthree + \M_\tmfour $. 
We  then apply  $\appresult<\lo>$ to obtain the following derivation $\tderiv_\tm$: 
 \[\begin{array}{c}
 \infer[\appresult<\lo>]{
 \Deri
[(\steps_\tmthree + \steps_\tmfour, \spine_\tmthree + \spine_\tmfour +1)]
{\typctxtwo_\tmthree \uplus \typctxtwo_\tmfour; \var: \M_\tmthree + \M_\tmfour }
{\tmthree \tmfour} \nf}
{ \tderiv_\tmthree \exder[\lo]  \Deri[(\steps_\tmthree, \spine_\tmthree)]
   {\typctxtwo_\tmthree; \var: \M_\tmthree} {\tmthree} {\neutype} 
  \quad \quad    
  \tderiv_\tmfour \exder[\lo]  \Deri[(\steps_\tmfour, \spine_\tmfour)] {\typctxtwo_\tmfour; \var: \M_\tmfour} {\tmfour} {\nf}}
 \end{array}\]

 We let $\typctxtwo :=\typctxtwo_\tmthree \uplus \typctxtwo_\tmfour $,
 $\steps_\tm:=\steps_\tmthree+\steps_\tmfour $ and $\spine_\tm:=\spine_\tmthree+\spine_\tmfour+1$
 and then conclude because of the following statements: 
 \begin{enumerate}
 \item \emph{Typing context}: $\typctx = \typctx_1 \uplus  \typctx_2 =
   \typctxtwo_\tmthree \uplus \typctxthree_\tmthree \uplus
   \typctxtwo_\tmfour \uplus  \typctxthree_\tmfour = \typctxtwo \uplus \typctxthree$.
   \item \emph{Indices}: 
     $(\steps, \spine) = (\steps_1+\steps_2, \spine_1+
     \spine_2+1) = (\steps_\tm+\steps_\tmtwo, \spine_\tm+\spine_\tmtwo)$.
   \end{enumerate}

  \item \emph{Rule $\appresult<\hd>$}. Similar to the previous case, but simpler.
\end{enumerate}
 \end{itemize}
\end{proof}

\gettoappendix {prop:subject-expansion}

\begin{proof}
  We prove, by induction on $\tm\Rew{\system}\tmtwo$,
  the stronger statement:
  
  Assume
  $\tm\Rew{\system}\tmtwo$,
  $\tderiv\exder[\system] \Deri[(\steps, \result)]\typctx{\tmtwo}\type$,
  $\tightpred\typctx$,
  and either $\tightpred\type$ or $\sysnotabs \system \tm$.
  
  Then there exists a typing
  $\tderivtwo\exder[\system] \Deri[(\steps+2, \result)]{\typctx}{\tm}\type$.

  \begin{itemize}
  \item Rule
    \[\infer{
      (\la\var \tmthree) \tmfour \Rew{\system} \tmthree \isub \var \tmfour
    }{}
    \]
    Assume 
    $\tderiv \exder[\system]\Deri[(\steps,\result)]
    {\typctx}{\tmthree \isub \var \tmfour}\type$
    and $\tightpred\typctx$.
    By applying Lemma~\ref{l:anti-substitution}
    we get the premisses of the following derivation $\tderiv'$:
    \[
    \infer{
      \Deri[
        (\steps_u+\steps_\tmfour+2,\result_u+\result_\tmfour)
      ]{\typctx_u\mplus \typctx_\tmfour}{(\la\var \tmthree) \tmfour}\type
    }{
      \infer{
        \Deri[
          (\steps_u+1,\result_u)
        ]{\typctx_u}{\la\var\tmthree}{\M\rightarrow\type}
      }{
        \tderiv_\tmthree \exder[\system]
        \Deri[
          (\steps_u,\result_u)
        ]{\typctx_u, \var \col \M}{\tmthree}\type
      }
      \qquad
      \tderiv_\tmfour \exder[\system]
      \Deri[
        (\steps_\tmfour,\result_\tmfour)
      ]{\typctx_\tmfour}{\tmfour}{\M}
    }
    \]
    with $(\steps,\result) = (\steps_u+\steps_\tmfour, \result_u+\result_\tmfour)$
    and $\typctx=\typctx_u\mplus\typctx_\tmfour$.
    
  \item Rule
    \[
    \infer{\la\var \tm \Rew{\system} \la\var \tmtwo}{
      \tm \Rew{\system} \tmtwo
    }
    \]
    Assume
    $\tderiv\exder[\system]\Deri[(\steps,\result)]
    {\typctx}{\la\var \tmtwo}\type$
    and $\tightpred\typctx$.
    Since $\sysabs \system {\la\var \tm}$ we must have hypothesis $\tightpred\type$,
    and as $\tderiv$ must then finish with rule $\funresult$
    we must have a subderivation
    $\tderiv_\tmtwo \exder[\system]
    \Deri[
      (\steps, \result-1)
    ]{\typctx, \var \col \mtight}{\tmtwo}{\nf}$.
    As $\tightpred{\typctx, \var \col \mtight}$
    we can apply the \ih and get the premiss of the derivation $\tderiv'$ below:
    \[
    \infer{
      \Deri[(\steps+2, \result)]{
        \typctx
      }{\la\var\tm}\type
    }{
      \tderiv_\tm \exder[\system]
      \Deri[(\steps+2, \result-1)]{\typctx, \var \col \mtight}{\tm}\nf
    }
    \]

  \item Rule
    \[
    \infer{
      \tm \tmthree \Rew{\system} \tmtwo \tmthree
    }{
      \sysnotabs \system \tm \quad \tm \Rew{\system}\tmtwo 
    }
    \]
    Assume
    $\tderiv\exder[\system]\Deri[(\steps,\result)]
    {\typctx}{\tmtwo \tmthree}\type$
    and $\tightpred\typctx$.
    The derivation $\tderiv$ must end
    with rule $\appsteps$ or $\appresult<\system>$.
    In the simple case of rule $\appresult<\spi>$,
    there is a subderivation
    $\tderiv_\tmtwo\exder[\spi]\Deri[
      (\steps, \result-1)
    ]{\typctx}{\tmtwo}{\neutype}$
    in $\tderiv$,
    and we can apply the \ih to get the premiss of the derivation $\tderiv'$ below:
    \[
    \infer{
      \Deri[
        (\steps+2, \result)
      ]{\typctx}{\tm \tmthree}{\type}
    }{
      \tderiv_\tm\exder[\spi]\Deri[
        (\steps+2, \result-1)
      ]{\typctx}{\tm}{\neutype}
    }
    \]
    In the case of $\appsteps$ or $\appresult<\ske>$,
    there are derivations
    $\tderiv_\tmtwo\exder[\system]\Deri[
      (\steps_\tmtwo, \result_\tmtwo)
    ]{\typctx_\tmtwo}{\tmtwo}{\type_\tmtwo}$
    and
    $\tderiv_\tmthree\exder[\system]\Deri[
      (\steps_\tmthree, \result_\tmthree)
    ]{\typctx_\tmthree}{\tmthree}{\type_\tmthree}$,
    with $\typctx=\typctx_\tmtwo\mplus\typctx_\tmthree$.
    Since $\tightpred{\typctx}$ we have $\tightpred{\typctx_\tmtwo}$,
    and since $\sysnotabs \system \tm$
    we can apply the \ih to get the derivation
    $\tderiv_\tm \exder[\system]
    \Deri[
      (\steps_\tmtwo+2, \result_\tmtwo)
    ]{\typctx_\tmtwo}{\tm}{\type_\tmtwo}$ and build,
    using the same rule $\appsteps$ or $\appresult<\ske>$,
    the derivation $\tderiv'$ below:
    \[
    \infer{
      \Deri[
        (\steps+2, \result)
      ]{\typctx}{\tm \tmthree}{\type}
    }{
      \tderiv_\tm \exder[\system]
    \Deri[
      (\steps_\tmtwo+2, \result_\tmtwo)
    ]{\typctx_\tmtwo}{\tm}{\type_\tmtwo}
    \quad
    \tderiv_\tmthree\exder[\system]\Deri[
      (\steps_\tmthree, \result_\tmthree)
    ]{\typctx_\tmthree}{\tmthree}{\type_\tmthree}
    }
    \]
    
  \item Rule
    \[
    \infer{
      \tmthree \tm \Rew{\ske}\tmthree \tmtwo
    }{
      \skneutral \tmthree \quad  \tm \Rew{\ske}\tmtwo
    }
    \]
    Assume
    $\tderiv\exder[\ske]\Deri[(\steps,\result)]
    {\typctx}{\tmthree \tmtwo}\type$
    and $\tightpred\typctx$.
    The derivation $\tderiv$ must end
    with rule $\appsteps$ or $\appresult<\ske>$,
    and therefore there are two derivations
    $\tderiv_\tmthree\exder[\ske]\Deri[
      (\steps_\tmthree, \result_\tmthree)
    ]{\typctx_\tmthree}{\tmthree}{\type_\tmthree}$
    and
    $\tderiv_\tmtwo\exder[\ske]\Deri[
      (\steps_\tmtwo, \result_\tmtwo)
    ]{\typctx_\tmtwo}{\tmtwo}{\type_\tmtwo}$,
    for some types $\type_\tmthree$ and $\type_\tmtwo$,
    with $\typctx=\typctx_\tmthree\mplus\typctx_\tmtwo$.
    Since $\tightpred{\typctx}$ we have
    $\tightpred{\typctx_\tmthree}$ and $\tightpred{\typctx_\tmtwo}$.
    Theorem~\ref{l:tight-spreading-spi-ske}
    concludes $\tightpred{\type_\tmthree}$ from $\skneutral \tmthree$.
    So the last rule of $\tderiv$ must be $\appresult<\ske>$,
    whence $\type=\neutype$ and $\type_\tmtwo=\tight$.
    Therefore we can apply the \ih to get the derivation
    $\tderiv_\tm \exder[\ske]
    \Deri[
      (\steps_\tmtwo+2, \result_\tmtwo)
    ]{\typctx_\tmtwo}{\tm}{\type_\tmtwo}$
    and build,
    using the same rule $\appresult<\ske>$,
    the derivation $\tderiv'$ below:
    \[
    \infer{
      \Deri[
        (\steps+2, \result)
      ]{\typctx}{\tmthree\tm}{\type}
    }{
    \tderiv_\tmthree\exder[\system]\Deri[
      (\steps_\tmthree, \result_\tmthree)
    ]{\typctx_\tmthree}{\tmthree}{\type_\tmthree}
    \quad
    \tderiv_\tm \exder[\system]
    \Deri[
      (\steps_\tmtwo+2, \result_\tmtwo)
    ]{\typctx_\tmtwo}{\tm}{\type_\tmtwo}
    }
    \]
  \end{itemize}
\end{proof}

\gettoappendix {thm:completeness}
\begin{proof}
  By induction on $\tm \Rewn[k]{\system} \tmtwo$.
  If $k = 0$ the statement is given by the existence of tight typings for $\sysnormalpr\system$ terms
  (\refprop{normal-forms-are-tightly-typable}), that also provides the \emph{moreover} part.
  Let $k > 0$ and $\tm \Rew{\system} \tmthree \Rewn[k-1]{\system}\tmtwo$.
  By \ih, there exists a tight typing derivation
  $\tderivtwo \tri \tyjnpre{(2(k-1), \syssize\system\tmtwo)} \tmthree$.
  By subject expansion (\refprop{subject-expansion})
  there exists a typing derivation $\tderiv$ of $\tmthree$
  with the same types in the ending judgement of $\tderivtwo$---then $\tderiv$ is tight---and with indices $(2k, \syssize\system\tmtwo)$.
\end{proof}

  \section{Appendix: Maximal Evaluation}

\subsection{Tight Correctness}
\gettoappendix {prop:mxsubject-reduction}
\gettoappendix {prop:mxsubject-reductionproof}

\subsection{Tight Completeness}
\gettoappendix {prop:mxsubject-expansion}
\gettoappendix {prop:mxsubject-expansionproof}

\section{Appendix: Linear Head Evaluation}
\label{app:linear-spine}
\newcommand{\proofslinearhead}{./proofs/linear-head}

\gettoappendix {prop:memory-bievaluation}
\begin{proof}
  The determinism of $\tolh$ is straightforward. We prove here the characterisation of 
  $\lsp$-normal terms and $\lsp$-neutral terms.

  $\Rightarrow)$  Let $\tm$ be $\tolh$-normal.  Then $\tm$ has either a free head variable $\var$
  or a bound head variable. We then refine the general statement as
  follows:
  \begin{enumerate}
  \item If $\tm$ is $\tolh$-normal and
    has a free head variable $\var$ and  is not a (potentially) substituted abstraction, then
    $\lhneutralp\tm\var$.
  \item If $\tm$ is $\tolh$-normal and has a free head variable $\var$ and
    is  a (potentially) substituted abstraction, then
    $\lhnormalp\tm\var$.
  \item If $\tm$ is $\tolh$-normal has  a bound head variable, then
    $\lhnormalclose \tm$.
  \end{enumerate}

  We show simultaneously the three 
  statement by induction on terms.

  If $\tm$ is a variable, then it corresponds to case (1) and we
  conclude by rule $\lhvar$.

  If $\tm = \la \vartwo \tmtwo$, then $\tmtwo$ is also $\tolh$-normal.
  There are two cases: case (2) or (3). If $\la \var \tmtwo$ corresponds to case (2), then $\vartwo \neq \var$
  and $\tmtwo$ corresponds to case (1) or (2). In the first case
  the \ih(1) gives that 
  $\lhneutralp\tmtwo \var$ and thus we conclude by
  rules $\lhnn$ and  $\lhnlam$. In the second  case
  the \ih(2) gives that 
  $\lhnormalp\tmtwo \var$ and we conclude with rule $\lhnlam$.
  
  If  $\la \vartwo \tmtwo$ corresponds to case (3), then either
  $\tmtwo$ corresponds to case (3), or $\tmtwo$ corresponds to cases (1) or (2) with
  $\vartwo = \var$. In the first case we get that
  $\lhnormalclose\tmtwo$ by the \ih(3) and thus  $\lhnormalclose{\la \var \tmtwo}$
  by rule $\lhhclose$. In the second case we get  that
  $\lhneutralp\tmtwo\var$ by the \ih(1) (resp. $\lhnormalp\tmtwo\var$ by the \ih(2)).
  We conclude with rules $\lhnn$, $\lhnlam$ and $\lhlamm$ (resp. $\lhnlam$ and $\lhlamm$).

  If $\tm = \tmtwo \tmthree$, then  $\tmtwo$ is also $\tolh$-normal, otherwise rule $\lhabs$
  would apply, and $\tmtwo$ is not a (potentially) substituted abstraction,
  otherwise rule $\lhb$ would apply. The term $\tmtwo \tmthree$ necessarily corresponds to case (1)
  for some variable $\var$ and  the same for $\tmtwo$. 
  We thus obtain that  $\lhneutralp\tm\var$ by the \ih(1) and we conclude by rule $\lhnapp$.

  If $\tm = \tmtwo \esub\vartwo \tmthree$, $\tmtwo$ is also $\tolh$-normal, otherwise rule $\lhsub$
  would apply, and $\tmtwo$ has no free head variable $\vartwo$,
  otherwise rule $\lhs$ would apply. Then $\tmtwo\esub\vartwo \tmthree$ corresponds to one of cases (1)-(2)-(3). If $\tmtwo$ corresponds to (1), then
  $\lhneutralp\tmtwo\var$ by the \ih(1) and we conclude with rule $\lhnesub$.
  If $\tmtwo$ corresponds to (2), then
  $\lhnormalp\tmtwo\var$ by the \ih(2) and we conclude with rule $\lhnosub$.
  If $\tmtwo$ corresponds to (3), then
  $\lhnormalclose\tmtwo\var$ by the \ih(3) and we conclude with rule $\lhmsub$.

  Now, given $\tm$ in $\tolh$-normal: if case (1) holds we conclude $\lhneutralp \tm \var$ with
  the previous statement (1), then rules $\lhnn$ and $\lhnpn$; 
  if case (2) holds we conclude $\lhnormalp \tm \var$ with
  the previous statement (2), then rule $\lhnpn$; if case (3)
  holds we conclude $\lhnormalp \tm \var$ with
  the previous statement (3), then rule $\lhnmn$; 

  $\Leftarrow)$ By induction on $\lhnormal \tm$. We remark that
  two cases are possible: either $\lhnormalp \tm \var$ for some variable $\var$
  or $\lhnormalclose \tm $.   We then refine the statement as follows:

  \begin{enumerate}
  \item If $\lhneutralp \tm \var$, then $\tm$ is $\tolh$-normal and
        $\tm$ has a head free variable $\var$ and $\tm$ is not a (potentially) substituted abstraction.
  \item If $\lhnormalp \tm \var$, then $\tm$ is $\tolh$-normal and
        $\tm$ has a head free variable $\var$.
  \item If $\lhnormalclose \tm $, then $\tm$ is $\tolh$-normal and
        $\tm$ has a head bound variable. 
  \end{enumerate}

  We reason by induction on the definition.

  If $\lhneutralp \tm \var$ by rule $\lhvar$, then property (1) trivially holds.

  If $\lhneutralp {\tmtwo \tmthree} \var$ because
   $\lhneutralp {\tmtwo } \var$ by rule $\lhnapp$, then 
  by the \ih(1) $\tmtwo$ is $\tolh$-normal --so rule $\lhb$ does not apply-- and
  $\tmtwo$ has a head free variable $\var$ and is not a (potentially) substituted abstraction
  --so rule $\lhapp$ does not apply. Then $\tmtwo \tmthree$ is $\tolh$-normal,
  it has a head free variable $\var$ and is not a (potentially) substituted abstraction.
  
  If $ \lhneutralp {\tmtwo \esub\vartwo\tmthree} \var$ because
  $ \lhneutralp {\tmtwo } \var$ and $\vartwo \neq \var$ by rule $\lhnesub$,
  then   by the \ih(1) $\tmtwo$ is $\tolh$-normal --so rule $\lhsub$ does not apply-- and
  $\tmtwo$ has a head free variable $\var$ and is not a (potentially) substituted abstraction
  --so rule $\lhs$ does not apply. Then $\tmtwo  \esub\vartwo\tmthree$ is $\tolh$-normal,
  it has a head free variable $\var$ and is not a (potentially) substituted abstraction.
  
  If $ \lhnormalp \tm \var$ because $\lhneutralp \tm \var$ by rule $ \lhnn$,
  then by the \ih(1) $\tm$ is $\tolh$-normal and
  has a head free variable $\var$. We are then done for this case.

  If $\lhnormalp {\la\vartwo\tmtwo} \var$ because
  $ \lhnormalp \tm \var$ and $\vartwo \neq \var$ by rule $\lhnlam$, 
 then by the \ih(2) $\tmtwo$ is $\tolh$-normal --so that rule $\lhabs$ does not apply--
 and $\tmtwo$ has a head free variable $\var$. We conclude
 $\la\vartwo\tmtwo$ is $\tolh$-normal and  has a head free variable $\var$.

 If $\lhnormalp {\tmtwo \esub\vartwo\tmthree} \var$ because
 $ \lhnormalp \tmtwo \var$ and $\vartwo\neq\var$ by rule $\lhnosub$,
 then by the \ih(2) $\tmtwo$ is $\tolh$-normal --so that rule $\lhsub$ does not apply--
 and $\tmtwo$ has a head free variable $\var$ --so that rule $\lhs$ does not apply--. We conclude
 $\tmtwo \esub\vartwo\tmthree$ is $\tolh$-normal and  has a head free variable $\var$.

 If $\lhnormalclose {\la\var\tmtwo} $ because
 $\lhnormalp {\tmtwo} \var$ by rule $\lhlamm$,
 then by the \ih(2) $\tmtwo$ is $\tolh$-normal --so that rule $\lhabs$ does not apply--.
 We conclude $\la\var\tmtwo$ is $\tolh$-normal and  has a  bound head variable.

 If $\lhnormalclose {\la\vartwo\tmtwo}  $ because $\lhnormalclose \tmtwo $ by rule $\lhhclose$,
 then by the \ih(3) $\tmtwo$ is $\tolh$-normal --so that rule $\lhabs$ does not apply--
 and $\tmtwo$ has a bound head variable.
 We conclude $\la\var\tmtwo$ is $\tolh$-normal and  has a bound head  variable.

 If $\lhnormalclose {\tmtwo \esub\vartwo\tmthree}  $ because
 $ \lhnormalclose \tmtwo  $ by rule $\lhmsub$, then
 by the \ih(3) $\tmtwo$ is $\tolh$-normal --so that rule $\lhabs$ does not apply--
 and $\tmtwo$ has a bound head variable. 
 We conclude $\tmtwo \esub\vartwo\tmthree$ is $\tolh$-normal and  has a bound head  variable.

\end{proof}


\subsection{Tight Correctness}

\gettoappendix {l:lin-head-headvar-in-context}
 \begin{proof}
 \hfill
 \begin{enumerate}
 \item By induction on $\lhneutralp \tm\var$. Cases:
 \begin{itemize}
 	\item \emph{Variable}, \ie $\tm = \var$. Then $\tderiv$ is 
 	$$\infer[\ax]{\Deri[(0, 0, 1)] {\var \col \single \type} \var \type}{}$$
	
 and so $\dom\typctx = \set\var$. If $\typctx(\var) = \mtight$ then it must be $\type = \tight$.

     \item \emph{Application}, \ie $\tm = \tmtwo \tmthree $. The last rule of $\tderiv$ can only be $\appsteps$ or $\appresult$. In both cases the left subterm $\tmtwo$ is typed by a sub-derivation
     $\tderiv' \exder[\systemlsp] \Deri[(\stepstwo, \estepstwo, \spinetwo)] {\typctx_\tmtwo}{\tmtwo}{\typetwo}$ such that all assignments in $\typctx_\tmtwo$ appear in $\typctx$. Since $\tm$  lh-neutral on $\vartwo$ implies $\tmtwo$ lh-neutral on $\vartwo$, we can apply the \ih and obtain that $\var \in \dom{\typctx_\tmtwo} \subseteq \dom\typctx$. If moreover, $\typctx(\var) = \mtight$ then $\typctx_\tmtwo(\var) = \mtight$ and
 by \ih  $\typetwo=\tight$ and $\dom{\typctx_\tmtwo} = \set\var$. This forces $\typetwo = \neutral$ and the last rule of $\tderiv$ to be $\appresult$. Then $\type = \neutral$ and $\typctx = \typctx_\tmtwo$, that implies $\dom\typctx = \set\var $.
    
     \item \emph{Explicit substitution}, \ie $\tm = \tmtwo
       \esub{\vartwo}{\tmthree}$ and $\vartwo \neq \var$. The
       last rule of $\tderiv$ is $\esrule$ and the left subterm
       $\tmtwo$ is typed by a sub-derivation $\tderiv' \exder[\systemlsp]
       \Deri[(\stepstwo, \estepstwo, \spinetwo)]
            {\typctx_\tmtwo; \vartwo \col \M}{\tmtwo}{\type}$ such that all types in $\typctx_\tmtwo$
            appear in $\typctx$. Since $\tm$ lh-neutral on $\var$
            implies $\tmtwo$ lh-neutral on $\var$, we can apply the \ih
            and obtain that $\var \in \dom{\typctx_\tmtwo} \subseteq
            \dom\typctx$. If moreover, $\typctx(\var) = \mtight$ then
            $(\typctx_\tmtwo; \vartwo \col \M) (\var) = \mtight$
              and by the \ih $\type=\tight$ and $\dom{\typctx_\tmtwo; \vartwo \col \M}
            = \set\var$. This forces $\M = \emm$ and the
            $\esrule$ rule to have no right premise. Then $\typctx =
            \typctx_\tmtwo$, that implies $\dom\typctx = \set\var $.
 \end{itemize}

 \item By induction on $\lhnormalp\tm\var$. If $\lhnormalp\tm\var$
    because  $\lhneutralp  \tm \var$ then it follows from the previous point. The two other cases are:
 \begin{itemize}
 \item \emph{Abstraction}, \ie $\tm = \la\vartwo\tmtwo$ with
   $\lhnormalp\tmtwo\var$ and $\vartwo \neq \var$. The last rule of $\tderiv$ can only be $\funsteps$ or $\funresult$. In both cases the subterm $\tmtwo$ is typed by a sub-derivation
     $\tderiv' \exder[\systemlsp] \Deri[(\stepstwo, \estepstwo, \spinetwo)] {\typctx;  \vartwo \col \M }{\tmtwo}{\typetwo}$. By \ih, $\var \in \dom{\typctx;  \vartwo \col \M }$ and so $\var \in \dom\typctx$, because $\vartwo\neq\var$. If moreover, $\typctx(\var) = \mtight$ then by \ih $\dom{\typctx;  \vartwo \col \M } = \set\var$, that is, $\M = \emm$. Then $\dom\typctx = \set\var$.
    
     \item \emph{Explicit substitution}, \ie $\tm = \tmtwo \esub\vartwo\tmthree$ with $ \lhnormalp\tmtwo\var$ and $\vartwo \neq \var$. The last rule of $\tderiv$ is $\esrule$ and the left subterm $\tmtwo$ is typed by a sub-derivation
     $\tderiv' \exder[\systemlsp] \Deri[(\stepstwo, \estepstwo, \spinetwo)] {\typctx_\tmtwo;  \vartwo \col \M }{\tmtwo}{\type}$ such that all types in $\typctx_\tmtwo$ appear in $\typctx$. By \ih, $\var \in \dom{\typctx_\tmtwo} \subseteq \dom\typctx$. If moreover, $\typctx(\var) = \mtight$ then by \ih  $\dom{\typctx_\tmtwo; \vartwo \col \M } = \set\var$, that is, $\M = \emm$. Therefore, the $\esrule$ rule has no right premise. Then $\typctx = \typctx_\tmtwo$, that implies $\dom\typctx = \set\var $.
 \end{itemize}

 \item By induction on $\lhnormalclose\tm$. Cases:
 \begin{itemize}
     \item \emph{Abstraction on the head variable}, \ie $\tm = \la\var\tmtwo$ with $ \lhnormalp\tmtwo\var$. If $\type = \tight$ then the last rule of $\tderiv$ can only be $\funresult$ and $\type = \abstype$: 
     $$
     \infer[\funresult]{\Deri[(\steps,\esteps,  \spine +1)] {\typctx} { \la\var\tmtwo } \abstype}
                   {\Deri[(\steps,\esteps, \spine)] {\typctx; \var \col \mtight} \tmtwo \tight} 
 $$
     By the previous point, $\dom{\typctx; \var \col \mtight} = \set\var$, that is, $\typctx$ is empty.    
    
   \item \emph{Abstraction on a non-head variable}, \ie $\tm = \la\var\tmtwo$
     with $ \lhnormalclose\tmtwo$. If $\type = \tight$ then the last rule of $\tderiv$ can only be $\funresult$ and $\type = \abstype$: 
     $$
     \infer[\funresult]{\Deri[(\steps,\esteps,  \spine +1)] {\typctx} { \la\var\tmtwo } \abstype}
                   {\Deri[(\steps,\esteps, \spine)] {\typctx; \var \col \mtight} \tmtwo \tight} 
 $$
     By \ih, $\typctx$ is empty.
    
     \item \emph{Explicit substitution}, \ie $\tm = \tmtwo \esub\vartwo\tmthree$ with $ \lhnormalclose\tmtwo$. The last rule of $\tderiv$ is $\esrule$ and the left subterm $\tmtwo$ is typed by a sub-derivation
     $\tderiv' \exder[\systemlsp] \Deri[(\stepstwo, \estepstwo, \spinetwo)] {\typctx_\tmtwo;  \vartwo \col \M }{\tmtwo}{\tight}$ such that all types in $\typctx_\tmtwo$ appear in $\typctx$. By \ih, the typing context $\typctx_\tmtwo;  \vartwo \col \M$ is empty, that forces $\M = \emm$. Therefore, the $\esrule$ rule has no right premise. Then $\typctx = \typctx_\tmtwo$, \ie $\typctx$ is empty.
 \end{itemize}

 \end{enumerate}
 \end{proof}

\gettoappendix{prop:tight-normal-forms-indicies-linear-head}
\begin{proof}
By induction on $\tderiv$. Cases of $\tm$:
  \begin{itemize}
    \item \emph{Variable}, \ie $\tm = \var$. Then $\tderiv$ has  the following  form and evidently verifies all the points of the statement:
 \[\begin{array}{cccc}
 \infer[\ax]{\Deri[(0,0,1)] {\var \col \mult\type} \var \type}{} 
  \end{array}\]
 The derivation verifies $\spine = 1 = \lhsize\var = \size\tderiv$, $\steps =  \esteps = 0$, as required.
  
\item \emph{Abstraction}, \ie $\tm = \la\var\tmtwo$ with
  $\lhnormal\tmtwo$. Cases of the last rule of $\tderiv$:
    \begin{itemize}
    \item \emph{$\funsteps$ rule}: 
  $$\infer[\funsteps]{
  \Deri[(\stepstwo+ 1, \esteps, \spine)]
{\typctx} { \la\var\tmtwo } {\M \rightarrow \type}
}
{
\tderivtwo \exder[\systemlsp] {\Deri[(\stepstwo, \esteps, \spine)] {\typctx; \var \col \M} \tmtwo \type }
}$$
with $\steps = \stepstwo+ 1$. 
\begin{enumerate}
  \item \emph{Size bound}: by \ih, $\lhsize\tmtwo \leq \size\tderivtwo$. Then, $\lhsize\tm = \lhsize\tmtwo +1 \leq_{\ih} \size\tderivtwo + 1 = \size\tderiv$.

  \item \emph{Tight bound}: $\tderiv$ is not tight, so the statement trivially holds.
\end{enumerate}
    \item \emph{$\funresult$ rule}: 
      $$\infer[\funresult]{\Deri[(\steps, \esteps, \spinetwo +1)] {\typctx} { \la\var\tmtwo } \abstype}
{\tderivtwo \exder[\systemlsp] {\Deri[(\steps, \esteps, \spinetwo)] {\typctx; \var \col \mtight} \tmtwo \tight }}$$
with $\spine = \spinetwo +1$. 

\begin{enumerate}
  \item \emph{Size bound}: by \ih, $\lhsize\tmtwo \leq \size\tderivtwo$.
Then, $\lhsize\tm = \lhsize\tmtwo +1 \leq_{\ih} \size\tderivtwo + 1 = \size\tderiv$.

  \item \emph{Tight bound}: if $\tderiv$ is tight, then $\tderivtwo$ is tight and by \ih $\spinetwo = \lhsize\tmtwo$ and $\steps = \esteps = 0$. Then, $\spine = \spinetwo +1 =_{\ih} \lhsize\tmtwo + 1 = \lhsize\tm$. 
\end{enumerate}
 \end{itemize}

  \item \emph{Application}, \ie $\tm = \tmtwo \tmthree$ with
    $\lhneutralp\tmtwo\var$ for some $\var$. Cases of the last rule of $\tderiv$:
    \begin{itemize}
    \item \emph{$\appsteps$ rule}: 
    \[\begin{array}{c}
\infer[\appsteps]{\Deri
[(\stepstwo + \stepsthree + 1, \estepstwo + \estepsthree, \spinetwo + \spinethree)]
{\typctxtwo \mplus \typctxthree}
{\tmtwo \tmthree} \type}
{ 	\tderivtwo \exder[\systemlsp] { \Deri[(\stepstwo, \estepstwo, \spinetwo)] \typctxtwo \tmtwo {\M \rightarrow \type}}
\quad
\TDeri[(\stepsthree, \estepsthree, \spinethree)] {\tderivthree}{\typctxthree} \tmthree{\M}
}\\\\
\end{array}
\]
with $\steps = \stepstwo + \stepsthree + 1$, $\esteps = \estepstwo + \estepsthree$, $\spine = \spinetwo + \spinethree$,  and $\typctx = \typctxtwo \mplus \typctxthree$.
		\begin{enumerate}
		  \item \emph{Size bound}: by \ih, $\lhsize\tmtwo \leq \size\tderivtwo$, from which it follows $\lhsize\tm = \lhsize\tmtwo +1  \leq_{\ih} \size\tderivtwo + 1 =  \size\tderiv$.

		  \item \emph{Tight bound}: we show that this case is impossible. If $\tderiv$ is tight then $\typctx = \typctxtwo \mplus \typctxthree$ is a tight typing context, and so is $\typctxtwo$. Then by \reflemma{lin-head-headvar-in-context}.1 the type of $\tmtwo$ in $\tderivtwo$ has to be tight---absurd.
		\end{enumerate}

    \item \emph{$\appresult$ rule}: 
    $$
\infer[\appresult]{\Deri[(\steps, \esteps, \spinetwo+1)] {\typctx} {\tmtwo \tmthree} \neutype}
{
\tderivtwo \exder[\systemlsp] {\Deri[(\steps, \esteps, \spinetwo)] \typctx \tmtwo \neutype
}}
$$
with $\spine = \spinetwo +1$. 
    \begin{enumerate}
		  \item \emph{Size bound}: by \ih, $\lhsize\tmtwo \leq \size\tderivtwo$. Then $\lhsize\tm = \lhsize\tmtwo + 1 \leq_{\ih} \size\tderivtwo + 1 = \size\tderiv$.

		  \item \emph{Tight bound}: if $\tderiv$ is tight, then $\tderivtwo$ is tight and by \ih $\spinetwo = \lhsize\tmtwo$ and $\steps = \esteps = 0$. Then, $\spine = \spinetwo +1 =_{\ih} \lhsize\tmtwo + 1 = \lhsize{\tmtwo \tmthree} = \lhsize\tm$. 
		\end{enumerate}
    \end{itemize}
    
  \item \emph{Explicit substitution}, \ie $\tm = \tmtwo \esub\var\tmthree$ and the last rule of $\tderiv$ is:
  $$
  \infer[\esrule]{\Deri[(\steps + \steps',  \esteps+ \esteps' + \size\M, \spine  + \spine' - \size\M)]
                     {\typctxtwo + \typctxthree}
                     {\tmtwo \esub\var\tmthree} \type}
               {\tderivtwo \exder[\systemlsp] \Deri[(\steps,  \esteps, \spine)] {\typctxtwo; \var  \col \M } \tmtwo \type 
                  \quad \quad
                \Deri[(\steps',  \esteps', \spine')] {\typctxthree} \tmthree \M                 
                 }
$$
with $\steps = \stepstwo + \stepsthree$, $\esteps = \estepstwo + \estepsthree$, $\spine = \spinetwo + \spinethree$,  and $\typctx = \typctxtwo \mplus \typctxthree$.
  \begin{enumerate}
		  \item \emph{Size bound}: by \ih, $\lhsize\tmtwo \leq \size\tderivtwo$. Then $\lhsize\tm = \lhsize\tmtwo  \leq_{\ih} \size\tderivtwo < \size\tderiv$.

\item \emph{Tight bound}: There are two cases:
		  \begin{itemize}
		    \item \emph{$\lhnormalp \tmtwo\vartwo$ for some $\vartwo \neq \var$}. By \reflemma{lin-head-headvar-in-context}.2 $\vartwo \in \dom{\typctxtwo}$. All assignments in $\typctxtwo$ are $\mtight$ because $\tderiv$ is tight, and so applying \reflemma{lin-head-headvar-in-context}.2 again we obtain that $\dom{\typctxtwo} = \set\vartwo$, that is, that $\M = \emm$. Two consequences: first, the $\esrule$ has no right premise, that is, it rather has the following shape:
		    $$
  \infer[\esrule]{\Deri[(\steps,  \esteps, \spine  )]
                     {\typctx}
                     {\tmtwo \esub\var\tmthree} \type}
               {\tderivtwo \exder[\systemlsp] \Deri[(\steps,  \esteps, \spine)] {\typctx} \tmtwo \type }
$$
second, $\tderivtwo$ is tight, and so by \ih $\steps = \esteps = 0$ and $\spine = \lhsize\tmtwo$. The statement follows from the fact that $\lhsize\tmtwo = \lhsize{ \tmtwo \esub\var\tmthree }$.
		    \item \emph{$ \lhnormalclose\tmtwo$}. If $\tderiv$ is tight then $\type = \tight$ and by \reflemma{lin-head-headvar-in-context}.3 the context $\typctxtwo; \var  \col \M $ is empty, that is, $\M = \emm$. Two consequences: first, the $\esrule$ has no right premise, that is, it rather has the following shape:
		    $$
  \infer[\esrule]{\Deri[(\steps,  \esteps, \spine  )]
                     {}
                     {\tmtwo \esub\var\tmthree} \type}
               {\tderivtwo \exder[\systemlsp] \Deri[(\steps,  \esteps, \spine)] {} \tmtwo \type }
$$
second, $\tderivtwo$ is tight, and so by \ih $\steps = \esteps = 0$ and $\spine = \lhsize\tmtwo$. The statement follows from the fact that $\lhsize\tmtwo = \lhsize{ \tmtwo \esub\var\tmthree }$.
  \end{itemize}
\end{enumerate}

    \end{itemize}
  \end{proof}

\gettoappendix{l:subst:hwmilner}
\begin{proof} 
By induction on $\lhc$. 
\begin{itemize}

\item If $\lhc = \ctxhole$, then by construction 
$\Gam = \emptyset$, $I = \cset{\iz}$, $ \tau= \sig_{\iz}$.
Then  $\Phi_\var$ has necessarily the following form:
$$\infer{\Deri[(0,0,1)]{\var:\mult{\tau}}{\var}{\tau}}{ \phantom{.}}\; \ax$$
where $\steps=0$, $\esteps=0$,
and $\spine=1$.

In this case  we have  $\lhc\cwc{\tmb} = \tmb$ and 
we let $\Phi_{\tmb} := \Phi_\tmb^{\iz}$. Moreover, $\steps_{\iz} = 0 + \steps_{\iz} = \steps + \steps_{\iz}$,
$\esteps_{\iz} = 0  +  \esteps_{\iz}  = \esteps + \esteps_{\iz}$,
$\spine_{\iz} = 1   + \spine_{\iz} -1  = \spine + \spine_{\iz} - 1$. 

Thus the statement holds. 

\item In all the other cases the property is straightforward by the \ih\

\end{itemize}

\end{proof}

\gettoappendix {prop:head-subject-head-reduction}

\begin{proof} By induction on the reduction relation $\tolh$. 
\begin{itemize}
\item $t =  \putinctx{\L}{\l \var. v} s    \tolhb    \putinctx{\L}{v\esub{\var}{s}} = t'$,
then 
we proceed by induction on $\L$. Let $\L = \ctxhole$. By 
construction the derivation $\Phi$ is of the form: 

$$\infer{\Deri[(\steps +  \steps' +2, \esteps + \esteps'+|\M|, \spine + \spine'-|\M|)]{\Pi  \mplus   \Gam}{(\la \var. v)s}{\tau}}
        {\infer{\Deri[(\steps +1, \esteps +|\M|, \spine -|\M|)]{\Pi}{\la \var. v}{ \M \rightarrow \sig}}
               {\Deri[(\steps, \esteps, \spine)]{\var:\M;\Pi}{v}{\sig}} \quad 
          \Deri[(\steps', \esteps', \spine')]{\Gam}{s}{\M}  } $$

We notice that $\Steps= \steps + \steps' +2 \geq 2$ as required. We construct the 
following derivation $\tderivtwo$:

$$ \infer{\Deri[(\steps + \steps', \esteps + \esteps'+|\M|, \spine + \spine'-|\M|)]
  {\Pi \mplus_{\iI} \Gam_i}{v\esub{\var}{s}}{\sig}}
         {\Deri[(\steps, \esteps, \spine)]{\var:\M;\Pi}{v}{\sig}  \quad 
           \Deri(\steps', \esteps', \spine')]{\Gam}{s}{\M}   }$$

So that we can verify $\steps+  \steps' = \Steps -2$,
$\esteps + \esteps' = \ESteps$ and
$\spine + \spine'  = \Spine$.

For  $\L = \L' \esub{\vartwo}{s}$, the statement follows from the \ih\ 

\item $\tm = \lhc\cwc{\var}\esub{\var}{\tmbb} \tolh  \lhc\cwc{\tmbb}\esub{\var}{\tmbb} = \tmb$, 
then $\Phi$ is of the form 
$$\infer{\Deri[(\steps +  \steps', \esteps +  \esteps'+|\M|, \spine +  \spine'-|\M|)]
              {\Pi \mplus   \Delta}{\lhc\cwc{\var}\esub{\var}{ \tmbb}}{\tau}}
        {\Deri[(\steps,\esteps,\spine)]{\var:\M;\Pi}{\lhc\cwc{\var}}{\tau} \quad 
         \Deri[(\steps',\esteps',\spine')]{\Delta}{\tmbb}{\M} }$$
where $\Steps=\steps +  \steps'$, $\ESteps=\esteps +  \esteps'+|\M|$ and $\Spine=\spine +  \spine'-|\M|$.

It is not difficult to see
that  $\size \M  \neq 0$ and thus $\ESteps \geq 1$ 
as required. 

Let $\M = \mult{\rho_i}_{\iI}$, where $I \neq \emptyset$.  We know
that $\Delta = \mplus_{\iI} \Delta_i$, where
$\Deri[(\steps_i,\spine_i,\garbage_i)]{\Delta_i}{ \tmbb}{\rho_i}$,
$\steps' = +_{\iI} \steps_i$, $\esteps' = +_{\iI} \esteps_i$, and
$\spine'= +_{\iI} \spine_i$.  By Lemma~\ref{l:subst:hwmilner} we have
$\Phi_{\lhc\cwc{ \tmbb}} \exder[\systemlsp]
\Deri[(\Steps', \ESteps', \Spine')]{\var:\mult{\rho_i}_{i \in I \setminus
    \iz}; \Pi \mplus  \Delta_{\iz}}{\lhc\cwc{ \tmbb}}{\tau}$, for some $\iz \in I$ where
$\Steps ' = \steps+ \steps_{\iz}$, $\ESteps ' = \esteps+ \esteps_{\iz}$, 
and $\Spine' = \spine + \spine_{\iz}-1$. Hence we construct
the following derivation $\Phi'$, where $J = I \setminus \iz$,
$\steps'' = +_{\jJ} \steps_{j}$, 
$\esteps''= +_{\jJ} \esteps_{j}$,
$\spine''= +_{\jJ} \spine_{j}$,
$$\infer{\Deri[(\Steps'+ \steps'',
                \ESteps'+ \esteps'' +|J|,
                \Spine' + \spine'' - |J|)]{\Pi \mplus   \Delta_{\iz} \mplus_{\jJ} \Delta_j }{\lhc\cwc{ \tmbb}\esub{\var}{ \tmbb}}{\tau}}
        {\Deri[(\Steps', \ESteps', \Spine')]{\var:\mult{\rho_j}_{\jJ}; \Pi \mplus  \Delta_{\iz}}{\lhc\cwc{ \tmbb}}{\tau} \quad  
          \Deri[(\steps'', \esteps'', \spine'')]{\mplus_{\jJ} \Delta_j}{ \tmbb}{\mult{\rho_j}_{\jJ}}  }    $$

Notice that 
$\Steps'+  \steps'' =  \steps+ \steps_{\iz}  +_{\jJ} \steps_{j}  = \steps+ \steps' = \Steps$, 
$\ESteps' + \esteps''+|J| =  \esteps+ \esteps_{\iz}  +_{\jJ} \esteps_{j} +|J|= 
\esteps+ \esteps' +|J|  = \esteps+ \esteps' +|\M| - 1 = \ESteps-1$, 
$\Spine' + \spine''-|J| =\spine+ \spine_{\iz} -1  +_{\jJ} \spine_{j}-|J| =  
\spine+ \spine' -| \M| = \Spine$.

\item All the other cases follow from the \ih\
\end{itemize}

\end{proof}

\gettoappendix {tm:head-correctness}
\begin{proof} 
By induction on $\lhsize\tderiv$. If $\tm$ is a $\tolh$ normal form  then by taking $\tmtwo \defeq \tm$ and 
$k=0$ the statement follows from the \emph{tightness} property of tight typings of normal forms (\refprop{tight-normal-forms-indicies-linear-head}.2)---the \emph{moreover} part follows from the \emph{neutrality} property (\refprop{tight-normal-forms-indicies-linear-head}.3). Otherwise, two cases:
\begin{enumerate}
\item \emph{Multiplicative steps}: $\tm \tom \tmthree$ and by quantitative subject reduction (\refprop{subject-reduction}) there is a derivation $\tderivtwo \exder[\lh] \Deri[(\steps-2, \esteps, \spine)] {\typctx}{\tmthree}{\type}$. By \ih, there exists $\tmtwo$ such that $\sysnormal\lh \tmtwo$ and $\tmthree
  \Rewn[(\steps-2)/2 + \esteps]{\lh} \tmtwo$ and $\syssize\lh\tmtwo = \result$. Just note that $\tm \tom \tmthree
  \Rewn[\steps/2-1 + \esteps]{\lh} \tmtwo$, that is, $\tm 
  \Rewn[\steps/2 + \esteps]{\lh} \tmtwo$. 
  
  \item \emph{Exponential steps}: $\tm \toe \tmthree$ and by quantitative subject reduction (\refprop{subject-reduction}) there is a derivation $\tderivtwo \exder[\lh] \Deri[(\steps, \esteps-1, \spine)] {\typctx}{\tmthree}{\type}$. By \ih, there exists $\tmtwo$ such that $\sysnormal\lh \tmtwo$ and $\tmthree
  \Rewn[(\steps)/2 + \esteps-1]{\lh} \tmtwo$ and $\syssize\lh\tmtwo = \result$. Just note that $\tm \toe \tmthree
  \Rewn[\steps/2-1 + \esteps-1]{\lh} \tmtwo$, that is, $\tm 
  \Rewn[\steps/2 + \esteps]{\lh} \tmtwo$. 
  \end{enumerate}
\end{proof}


\subsection{Tight Completeness}

\gettoappendix{prop:lin-head-normal-forms-are-tightly-typable}
\begin{proof}
In the proof, for the sake of simplicity, we let the indicies on the judgements generic, and not as precise as in the statement, because once one knows that there is a tight derivation then the indicies are forced by \refprop{tight-normal-forms-indicies-linear-head}.
  \begin{enumerate}
    \item By induction on $\lhneutralp\tm\var$:
  \begin{itemize}
\item \emph{Variable}, \ie $\tm = \var$. Then the derivation 
  $$\infer[\axres]{\Deri[(0, 0, 1)] {\var \col \single \neutral} \var \neutral}
                                { } $$  
  is tight and types $\var$ with $\neutral$.

\item \emph{Application}, \ie $\tm = \tmtwo \tmthree$ and
  $\lhneutral\tm$  because $\lhneutral\tmtwo$. By \ih, there is a \precise derivation 
$\tderivtwo \exder[\systemlsp] \Deri[(\steps,\esteps,\spine)] \typctx \tmtwo \neutral$. Then the following is a \precise derivation $\tderiv$ typing  $\tm = \tmtwo \tmthree$ with $\neutral$:
$$
\infer[\appresult]{\Deri[(\steps,\esteps, \spine+1)] {\typctx} {\tmtwo \tmthree} \neutype}
{\tderivtwo \exder[\systemlsp] \Deri[(\steps,\esteps, \spine)] \typctx \tmtwo \neutype}
$$

\item \emph{Explicit substitution}, \ie $\tm =\tmtwo \esub\vartwo\tmthree$
  and $\lhneutral\tm$  because $\lhneutralp\tmtwo\var$ and $\var \neq \vartwo$. By \ih, there is a \precise derivation 
$\tderivtwo \exder[\systemlsp] \Deri[(\steps,\spine)] \typctx \tmtwo \neutral$. By \reflemma{lin-head-headvar-in-context}.1, $\dom\typctx = \set\vartwo$, that is, in $\typctx$ the variable $\var$ is implicitly typed with $\emm$. Then the following tight derivation $\tderiv$ types $\tm =\tmtwo \esub\var\tmthree$ with $\neutral$:
$$
\infer[\esrule]{\Deri[(\steps,  \esteps, \spine  )]
                     {\typctx}
                     {\tm \esub\var\tmb} \neutral}
               {\Deri[(\steps,  \esteps, \spine)] {\typctx; \var  \col \emm } \tm \neutral
                 }
$$
\end{itemize}

  \item First, by induction on $ \lhnormalp\tm\var$:
  \begin{itemize}
    \item $\lhnormalp\tm\var$ because $ \lhneutralp \tm\var$. Then it follows from the previous point.
    
    \item \emph{Abstraction}, \ie $\tm = \la\vartwo\tmtwo$
      and $\lhnormalp \tm\var$ because $\lhnormalp\tmtwo\var$ and $\var \neq \vartwo$. By \ih
      there is a tight derivation $\tderivtwo \exder[\systemlsp] \Deri[(\steps,\esteps,\spine)] \typctxtwo \tmtwo
      \tight$. Since the derivation $\tderivtwo$ is tight, the typing context $\typctxtwo$ has the shape $\typctx; \vartwo \col \mtight$ (potentially, $\vartwo \col \emm$).  Then the following is a \precise derivation for
      $\la\vartwo\tmtwo$ with $\abstype$:
 $$\infer[\funresult]{\Deri[(\steps,\esteps,\spine+1)] {\typctx} {\la\vartwo\tmtwo} \abstype}
{\tderivtwo   \exder[\systemlsp] \Deri[(\steps,\esteps,\spine+1)] {\typctx; \vartwo \col \mtight} \tmtwo \tight } $$
    
\item \emph{Explicit substitution}, \ie $\tm =\tmtwo \esub\vartwo\tmthree$
  and $ \lhnormalp\tm \var$ because $\lhnormalp\tmtwo\var$ and $\var \neq \vartwo$. It is essentially like in the neutral case. By \ih, there is a \precise derivation 
$\tderivtwo \exder[\systemlsp] \Deri[(\steps,\esteps,\spine)] {\typctxtwo} \tmtwo \tight$. By \reflemma{lin-head-headvar-in-context}.1, $\dom\typctxtwo  = \set \var$, that is, in $\typctxtwo$ the variable $\vartwo$ is implicitly typed with $\emm$. Then using the notation $\typctxtwo = \typctx; \vartwo  \col \emm$ the following tight derivation $\tderiv$ types $\tm =\tmtwo \esub{\vartwo}{\tmthree}$:
$$
\infer[\esrule]{\Deri[(\steps,  \esteps, \spine  )]
                     {\typctx}
                     {\tmtwo \esub\vartwo\tmb} \tight}
               {\Deri[(\steps,  \esteps, \spine)] {\typctx; \vartwo  \col \emm } {\tmtwo} \tight
                 }
$$
The part about predicates follows from the \ih
  \end{itemize}

  Now, by induction on $ \lhnormalclose \tm$:
  \begin{itemize}
  \item \emph{Abstraction on the head variable}, \ie $\tm = \la\var\tmtwo$
    and $\lhnormalp \tm \var$ because $ \lhnormalp \tmtwo\var$. By \ih
      there is a tight derivation $\tderivtwo \exder[\systemlsp] \Deri[(\steps,\esteps,\spine)] \typctxtwo \tmtwo
      \tight$. Since the derivation $\tderivtwo$ is tight, the typing context $\typctxtwo$ has the shape $\typctx; \vartwo \col \mtight$ (potentially, $\vartwo \col \emm$).  Then the following is a \precise derivation for
      $\la\vartwo\tmtwo$ with $\abstype$:
 $$\infer[\funresult]{\Deri[(\steps,\esteps,\spine+1)] {\typctx} {\la\vartwo\tmtwo} \abstype}
{\tderivtwo   \exder[\systemlsp] \Deri[(\steps,\esteps,\spine)] {\typctx; \vartwo \col \mtight} \tmtwo \tight } $$
    
\item \emph{Abstraction on a non-head variable}, \ie $\tm = \la\var\tmtwo$
  and $\lhnormalp \tm \var$ because $\lhnormalclose\tmtwo$. It is exactly as in the previous sub-case. By \ih
      there is a tight derivation $\tderivtwo \exder[\systemlsp] \Deri[(\steps,\esteps,\spine)] \typctxtwo \tmtwo
      \tight$. Since the derivation $\tderivtwo$ is tight, the typing context $\typctxtwo$ has the shape $\typctx; \vartwo \col \mtight$ (potentially, $\vartwo \col \emm$).  Then the following is a \precise derivation for
      $\la\vartwo\tmtwo$ with $\abstype$:
 $$\infer[\funresult]{\Deri[(\steps,\esteps,\spine+1)] {\typctx} {\la\vartwo\tmtwo} \abstype}
{\tderivtwo   \exder[\systemlsp] \Deri[(\steps,\esteps,\spine)] {\typctx; \vartwo \col \mtight} \tmtwo \tight } $$

\item \emph{Explicit substitution}, \ie $\tm =\tmtwo \esub\vartwo\tmthree$
  and $\lhnormalclose \tm$ because $\lhnormalclose\tmtwo$. By \ih, there is a \precise derivation 
$\tderivtwo \exder[\systemlsp] \Deri[(\steps,\esteps,\spine)] \typctxtwo \tmtwo \tight$. By \reflemma{lin-head-headvar-in-context}.3, $\typctxtwo$ is empty, that is, the variable $\vartwo$ is implicitly typed with $\emm$. Then the following tight derivation $\tderiv$ types $\tm =\tmtwo \esub\var\tmthree$:
$$
\infer[\esrule]{\Deri[(\steps,  \esteps, \spine  )]
                     {}
                     {\tmtwo \esub\vartwo\tmb} \tight}
               {\Deri[(\steps,  \esteps, \spine)] {\vartwo  \col \emm } \tmtwo \tight
                 }
$$
The part about predicates follows from the \ih
  \end{itemize}
  \end{enumerate}
\end{proof}

\gettoappendix {l:anti-subst}
\begin{proof} By induction on $\lhc$.
\begin{itemize}

\item If $\lhc = \ctxhole$, then 
we let $\Gam_0 = \emptyset$, $ \sig_{1}=  \tau$,
$\Delta_1 = \Gam$. We have $(\steps_1, \esteps_1, \spine_1)= (\Steps, \ESteps, \Spine)$
and $(\steps, \esteps, \spine) = (0,0,1)$. All the equalities are verified. 

\item In all the other cases the property is straightforward by the \ih\
\end{itemize}
\end{proof}

\gettoappendix {prop:linear-subject-expansion}
\begin{proof}
The proof is by induction on $t  \tolh  t'$.
 \begin{itemize}
\item  If 
$t = \L\cwc{\l {\var}.\tmtwo}\tmthree \to \L\cwc{\tmtwo[\var/\tmthree]} = t'$, then 
we proceed by induction on $\L$. Let $\L = \ctxhole$, then by 
construction
 $\Gam = \Del \uplus \Pi$ and we have the following derivation: 

$$\infer{\Deri[(\steps+\steps', \esteps +\esteps' + |\M|, \spine + \spine' -|\M|)]
              {\Del\uplus \Pi}{\tmtwo[\var/\tmthree]}{\tau}}
        {\Deri[(\steps, \esteps, \spine)]{\var{:}\M;\Del}{\tmtwo}{\tau}\quad 
         \Deri[(\steps', \esteps', \spine')]{\Pi}{\tmthree}{\M} }
        $$

We then construct the following derivation

$$\infer{\Deri[(\steps+\steps'+2, \esteps +\esteps' + |\M|, \spine + \spine' -|\M|]{\Del\uplus \Pi}{(\l \var. \tmtwo)\tmthree}{\tau}}
        {\infer{\Deri[(\steps+1, \esteps+|\M|, \spine-|\M|)]{\Del}{\l \var. \tmtwo}{\M \rightarrow \tau}}
               {\Deri[(\steps, \esteps, \spine)]{\var{:}\M;\Del}{\tmtwo}{\tau}} \quad 
         \Deri[(\steps', \esteps', \spine')]{\Pi}{\tmthree}{\M}}$$ 

For  $\L = \L' [\vartwo/u]$, the statement follows from the \ih\

\item If $t = \lhc\cwc{\var}[\var/\tmthree]  \Rew{}   \lhc\cwc{\tmthree}[\var/\tmthree] = t'$, then by 
construction $\Gam = \Del \uplus  \Pi $ and the 
type  derivation of $t'$ has the following form:

$$\infer{\Deri[(\Steps', \ESteps', \Spine')]{\Del   \uplus  \Pi}{\lhc\cwc{\tmthree}[\var/\tmthree]}{\tau}}
        {\Deri[(\steps_\lhc, \esteps_\lhc, \spine_\lhc)]{\var{:}\M; \Del}{\lhc\cwc{\tmthree}}{\tau} \quad 
         \Deri[(\steps_\tmthree, \esteps_\tmthree, \spine_\tmthree)]{\Pi}{\tmthree}{ \M}}$$

where $(\Steps', \ESteps', \Spine') = (\steps_\lhc+\steps_\tmthree,
\esteps_\lhc+\esteps_\tmthree+|\M|, \spine_\lhc+\spine_\tmthree-|\M|)$.

By Lemma~\ref{l:anti-subst} 
$\exder[\systemlsp]  \Deri[(\steps, \esteps, \spine)]{\Gam_0 + \var{:}\mult{\sig_1}  }{\lhc\cwc{\var}}{\tau}$ and 
$\exder[\systemlsp]  \Deri[(\steps_1, \esteps_1, \spine_1)]{\Del_1}{\tmthree}{\sig_1}$, where
$\steps_\lhc=\steps+\steps_1$, 
$\esteps_\lhc=\esteps+\esteps_1$ and 
$\spine_\lhc=\spine+\spine_1-1$. Note that $\var \notin \fv{\tmthree}$. We let 
$I = K \uplus \{ 1 \}$ where  $\M = \mult{\sig_i}_{i \in K}$. 
We have necessarily $\Gam_0 = \Gam'_0 ;   \var{:}\mult{\sig_k}_{k \in K}$.  

We remark that $\tmthree$ has necesarily been typed with a $(\many)$ rule so that
there are derivations $\Deri[(\steps_k,\esteps_k,\spine_k)]{\Pi_k}{\tmthree}{\sig_k}\ (k \in K)$,
such that $\Pi = \uplus_{k \in K } \Pi_k$, and
$\M = +_{k \in K} \sig_k$ and $\steps_\tmthree=+_{k \in K}\steps_k$, 
$\esteps_\tmthree= +_{k \in K}\esteps_k$, $\spine_\tmthree=+_{k \in K}\spine_k$. 
By applying rule $(\many)$ again we obtain 
$\Deri[(\steps_\tmthree+\steps_1,\esteps_\tmthree+\esteps_1,\spine_\tmthree+\spine_1)]
      {\Pi + \Delta_1}{\tmthree}{\M + \mult{\sig_1}}$. 

We can now construct the following derivation

$$\infer{\Deri(\steps+\steps_\tmthree+\steps_1, 
               \esteps+\esteps_\tmthree+\esteps_1 +|I|, 
               \spine+\spine_\tmthree+\spine_1 -|I|)]{\Gam'_0\uplus \Pi \uplus  \Delta_1}{\lhc\cwc{\var}[\var/\tmthree]}{\tau}}     
        {\Deri(\steps, \esteps, \spine)]{\Gam'_0 ;  \var{:}\mult{\sig_i}_{\iI}}{\lhc\cwc{\var}}{\tau} \quad 
         \Deri[(\steps_\tmthree+\steps_1,\esteps_\tmthree+\esteps_1,\spine_\tmthree+\spine_1)]
             {\Pi \uplus  \Delta_1} {\tmthree}{\mult{\sig_i}_{\iI}}}$$

We conclude since
$\steps+\steps_\tmthree+\steps_1=
\steps_\lhc+\steps_\tmthree=\Steps'$,
$\esteps+\esteps_\tmthree+\esteps_1+|I|=
\esteps_\lhc+\esteps_\tmthree+|I|=\esteps_\lhc+\esteps_\tmthree+|\M|+1=\ESteps'+1$,
$\spine+\spine_\tmthree+\spine_1 -|I| = 
\spine_\lhc+\spine_\tmthree -|\M| = \Spine'$. 

\item All the inductive cases are straightforward.

 \end{itemize}
\end{proof}

\gettoappendix{th:completeness-linear-head}
\begin{proof}
By induction on $\tm \tolh^k \tmtwo$. If $k = 0$ then $\tm =  \tmtwo$. 
\refprop{lin-head-normal-forms-are-tightly-typable}
gives the existence of  a \precise\
typing $\tderiv \exder[\systemlsp]  \tyjnpre{(\steps, \esteps, \spine)} \tm$. 
\refprop{tight-normal-forms-indicies-linear-head} then gives
$\spine = \lhsize\tm= \lhsize\tmtwo$ and $\steps =  \esteps = 0$.
The property then holds for $k_1=k_2=0$. 

Let $0 < k = k'+1$ and
$\tm \tolh \tmthree \tolh^{k'}\tmtwo$. By \ih\  there exists a tight
typing derivation $\tderivtwo\exder[\systemlsp]  \tyjnpre{(2k'_1, k'_2, \lhsize\tmtwo)}
\tmthree$, where $k'=k'_1+k'_2$. By quantitative subject expansion \refprop{linear-subject-expansion}
there exists a typing derivation $\tderiv$ of $\tmthree$ with the same
types in the ending judgement of $\tderivtwo$---then $\tderiv$ is
tight---and with indices $(2k'_1+2, k'_2, \lhsize\tmtwo)$ or $(2k'_1, k'_2+1,
\lhsize\tmtwo)$. 

In the first case we let $k_1 = k'_1 +1$ and $k_2 = k'_2$, so that
$k = 1 + k' =_{\ih} 1+ k'_1+k'_2 = k_1 +k_2$ as required. Moreover,
$\tderiv \exder[\systemlsp]  \tyjnpre{(2k_1, k_2, \lhsize\tmtwo)} \tm$.

In the second  case we let $k_1 = k'_1$ and $k_2 = k'_2+1$, so that
$k = 1 + k' =_{\ih} 1+ k'_1+k'_2  = k_1 +k_2$ as required. Moreover,
$\tderiv \exder[\systemlsp]  \tyjnpre{(2k_1, k_2, \lhsize\tmtwo)} \tm$.
\end{proof}

  \section{Appendix: Leftmost Evaluation and Minimal Typings}

The relation \emph{to be a positive/negative occurrence} is transitive in the following sense.

\begin{lemma}[Transitivity of polarities]
\label{l:trans-pol}
  Let $\atype,\atypetwo,\atypethree$ be (multi)-types and $a,b \in \set{+,-}$. Then $\atypetwo \in \typeocc a \atype$ and $\atypethree \in \typeocc b \atypetwo$ then $\atypethree \in \typeocc {\polcomp a b} \atype $, where 
  \[\begin{array}{ccc\colspace \colspace ccc \colspace \colspace ccc\colspace \colspace ccc}
    \polcomp + + & \defeq & +
    &
    \polcomp - + & \defeq & -
    &
    \polcomp - - & \defeq & +
    &
    \polcomp + - & \defeq & -
  \end{array}\]
  \end{lemma}
  
 \begin{proof}
  Let $\neg + \defeq -$ and $\neg - \defeq +$.   By induction on $\atypetwo \in \typeocc \atype a$. The proof can be presented in a way that is completely parametric in the polarities, but for readability reasons we spell out the positive and negative cases separetely. Cases in which $a = +$:
  \begin{itemize}
    \item \emph{Axioms}, \ie $\atypetwo = \atype$. Note that $\polcomp + b = b$. Then $\atypethree \in \typeocc \atypetwo b$ becomes $\atypethree \in \typeocc \atype b = \typeocc \atype {\polcomp + b}$ as required.
    
    \item \emph{Positive occurrence in an element $\type$ of a multiset $\mtype$}, \ie $\atype = \mtype$ and $\atypetwo \in \typeocc \mtype +$ because $\atypetwo \in \typeocc \type +$. By ih, $\atypethree \in \typeocc \type {\polcomp + b}$ and so $\atypethree \in \typeocc \mtype {\polcomp + b}$ by one of the two rules about multisets.
    
    \item \emph{Positive occurrence on the right of $\tarrow\mtype\type$}, \ie $\atype = \tarrow\mtype\type$ and $\atypetwo \in \typeocc {\tarrow\mtype\type} +$ because $\atypetwo \in \typeocc \type +$. By \ih, $\atypethree \in \typeocc \type {\polcomp + b}$ and so $\atypethree \in \typeocc \mtype {\polcomp + b}$ by one of the two rules about arrow types.
    
    \item \emph{Negative occurrence on the left of $\tarrow\mtype\type$}, \ie $\atype = \tarrow\mtype\type$ and $\atypetwo \in \typeocc {\tarrow\mtype\type} +$ because $\atypetwo \in \typeocc \mtype -$. By \ih, $\atypethree \in \typeocc \type {\polcomp - b}$ and so $\atypethree \in \typeocc \mtype {\neg\polcomp - b} = \typeocc \mtype {\polcomp + b}$ by one of the two rules about arrow types.
  \end{itemize}
  
  Cases in which $a = -$:
  \begin{itemize}    
    \item \emph{Negative occurrence in an element $\type$ of a multiset $\mtype$}, \ie $\atype = \mtype$ and $\atypetwo \in \typeocc \mtype -$ because $\atypetwo \in \typeocc \type -$. By ih, $\atypethree \in \typeocc \type {\polcomp - b}$ and so $\atypethree \in \typeocc \mtype {\polcomp - b}$ by one of the two rules about multisets.
    
    \item \emph{Negative occurrence on the right of $\tarrow\mtype\type$}, \ie $\atype = \tarrow\mtype\type$ and $\atypetwo \in \typeocc {\tarrow\mtype\type} -$ because $\atypetwo \in \typeocc \type -$. By \ih, $\atypethree \in \typeocc \type {\polcomp - b}$ and so $\atypethree \in \typeocc \mtype {\polcomp - b}$ by one of the two rules about arrow types.
    
    \item \emph{Poisitive occurrence on the left of $\tarrow\mtype\type$}, \ie $\atype = \tarrow\mtype\type$ and $\atypetwo \in \typeocc {\tarrow\mtype\type} -$ because $\atypetwo \in \typeocc \mtype +$. By \ih, $\atypethree \in \typeocc \type {\polcomp + b}$ and so $\atypethree \in \typeocc \mtype {\neg\polcomp + b} = \typeocc \mtype {\polcomp  b}$ by one of the two rules about arrow types.
  \end{itemize}
\end{proof}

\subsection{Shrinking Correctness}
\gettoappendix {prop:shrinking-normal-forms-forall}
\begin{proof}
  By mutual induction on $\skneutral \tm$ and $\sknormal \tm$. 
  \begin{enumerate} 
    \item Cases of $\skneutral \tm$:  
  \begin{itemize}
	\item \emph{Variable}, \ie $\tm = \var$. Then 
	    $$\infer[\ax]{\Deri[(0, 0)] {\var \col \single \type} \var \type}{}$$
	    Moreover, $\tysize \type + \sksize\var =  \tysize \type + 0 = \tysize \type = \tysize{\single \type} = \tysize{\var \col \single \type}$.

	\item \emph{Application}, \ie $\tm = \tmtwo \tmthree$ with $\skneutral \tmtwo$ and $\sknormal \tmthree$. The hypothesis that $\tderiv$ is traditional forces the last rule of $\tderiv$ to be $\appsteps$ and $\tderiv$ to have the following form:	
	$$	\infer[\appsteps]{
\Deri[(\steps_\tmtwo+ \sum_\iI \stepstwo_i + 1, 0)]
{\typctx_\tmtwo \mplus (\mplus_\iI \typctxtwo_i)}
{\tmtwo \tmthree} \type
}
{	
\tderiv_\tmtwo \exder[\lo] \Deri[(\steps_\tmtwo, 0)] {\typctx_\tmtwo} \tmtwo {\tarrow{\mult{\typetwo_i}_\iI} \type} \quad
\infer[\many]{
	\Deri[(\sum_\iI \stepstwo_i, 0)] {\mplus_\iI \typctxtwo_i} \tmthree {\mult{\typetwo_i}_\iI} 
     }
	{
    		(\tderiv_\tmthree \exder[\lo] \Deri[(\stepstwo_i, 0)] {\typctxtwo_i} \tmthree {\typetwo_i})_\iI
	}
	}$$

	By \ih (Point 2) (repeatedly) applied to $\tmthree$, $\sksize\tmthree \leq \tysize {\typctxtwo_i} + \tysize {\typetwo_i}$ for every $\iI$, and so $\sksize\tmthree \leq \tysize {\mplus_\iI\typctxtwo_i} + \tysize {\mult{\typetwo_i}_\iI}$.

By \ih (Point 1) applied to $\tmtwo$,  $\sksize\tmtwo + \tysize {\tarrow {\mult{\typetwo_i}_\iI} \type} \leq \tysize {\typctx_\tmtwo}$.
	
	Then:
	$$\begin{array}{rcllllll}
	  \sksize\tm + \tysize {\type}  
	  & = &
	  \underbrace{\sksize\tmtwo + \sksize\tmthree + 1}_{=\sksize\tm} + \tysize {\type} 
	  \\
	  & \leq_{\tiny \mbox{\ih on $\tmthree$}} &
	  \sksize\tmtwo + \tysize {\mplus_\iI\typctxtwo_i} + \underbrace{\tysize {\mult{\typetwo_i}_\iI} + 1 + \tysize {\type}}_{=\tysize {\tarrow {\mult{\typetwo_i}_\iI} \type}}\\
	  & = &
	  \tysize {\mplus_\iI\typctxtwo_i} + \underbrace{\sksize\tmtwo + \tysize {\tarrow {\mult{\typetwo_i}_\iI} \type}}_{\leq \tysize {\typctx_\tmtwo}}
	  \\
	  & \leq_{\tiny \mbox{\ih on $\tmtwo$}} &
	  \tysize {\mplus_\iI\typctxtwo_i} + \tysize {\typctx_\tmtwo}
	  \\
	  & = & \tysize{\typctx_\tmtwo \mplus (\mplus_\iI \typctxtwo_i)}
	\end{array}$$

\end{itemize}

\item Cases of $\sknormal{\tm}$:
\begin{enumerate}
\item \emph{$\skneutral{\tm}$}. By \ih, $\tysize\type + \sksize\tm \leq \tysize\typctx$, from which it trivially follows $\sksize\tm \leq \tysize\typctx + \tysize\type = \tysize\tderiv$. 

\item \emph{Abstraction}, \ie $\tm = \la\vartwo\tmtwo$ and $\sknormal{\tmtwo}$. Since $\tderiv$ is traditional, its last rule is necessarily $\funsteps$. Two sub-cases:
	\begin{enumerate}
		\item $\vartwo \in \dom{\typctx_\tmtwo}$. Then let $\vartwo \col \M$ the declaration of $\vartwo$ in the premise of $\funsteps$. Then $\tderiv$ has the following form:
		$$\infer[\funsteps]{
				\tyjn{(\steps, 0)} {\la{\vartwo}\tmtwo} {\typctx} {\tarrow \M \typetwo}
			}
			{
				\tderiv_\tmtwo \exder[\lo] \tyjn{(\steps-1, 0)} \tmtwo {\vartwo:\M; \typctx} \typetwo}$$
		with $\type = \tarrow \M \typetwo$. We have
		  $$\begin{array}{rcllllll}
		    \sksize{\la\vartwo\tmtwo}  & = & \sksize\tmtwo +1 
		    \\
		    & \leq_{\ih} &
		    \tysize {\vartwo:\M; \typctx} + \tysize\typetwo+1
		    \\
		    & = &
		    
		    \tysize\typctx + \underbrace{\tysize \M+ \tysize\typetwo+1}_{= \tysize {\tarrow \M \typetwo} = \tysize \type }
		    \\
		    & = &
		    \tysize\typctx + \tysize \type 
		    \\
		    & = &
		    \tysize \tderiv
		   \end{array}$$
	

		\item $\vartwo \notin \dom{\typctx_\tmtwo}$. Then $\tderiv$ is the derivation:
		$$\infer[\funsteps]{
				\tyjn{(\steps-1, 0)} {\la{\vartwo}\tmtwo} {\typctx} {\tarrow {\emptymset} \typetwo}
			}
			{
				\tderiv_\tmtwo \exder[\lo] \tyjn{(\steps, 0)} \tmtwo {\typctx} \typetwo}$$		
		with $\type = \tarrow \emptymset \typetwo$. Then:
		  $$\begin{array}{rcllllll}
		    \sksize{\la\vartwo\tmtwo}  & = & \sksize\tmtwo +1 
		    \\
		    & \leq_{\ih} &
		    
		    \tysize {\typctx}+ \underbrace{\tysize\typetwo+1}_{= \tysize {\tarrow \emptymset \typetwo} = \tysize \type }
		    \\
		    & = &
		    \tysize\typctx + \tysize \type 
		   \\
		    & = &
		    \tysize \tderiv
		    \end{array}$$
	
	\end{enumerate}
\end{enumerate}
\end{enumerate}
\end{proof}

\gettoappendix {l:shrinking-neutral-spreading}
\begin{proof}
  By induction on
  $\loneutral \tm$: 
  \begin{itemize} 
  \item \emph{Variable}, \ie
  $\tm=\var$. Then $\typctx= \var:\mult{\type}$ and $\type$ is a
  positive occurrence of $\typctx$.
  
  \item \emph{Application}, \ie
  $\tm=\tmtwo\tmthree$, the last rule of $\tderiv$ can only be
  $\appsteps$ or $\appresult<\lo>$. In both cases the left subterm $\tmtwo$
  is typed by a sub-derivation
  $\tderiv_\tmtwo \exder[\lo] \Deri[(\steps', \result')]
  {\typctx_\tmtwo}{\tmtwo}{\typetwo}$ such that all types in
  $\typctx_\tmtwo$ appear in $\typctx$. Since $\loneutral \tm$ implies
  $\loneutral \tmtwo$, we can apply the \ih and obtain that $\typetwo$
  has a positive occurrence in $\typctx$, that is, that there is a
  declaration $\var\col\mtype$ in $\typctx$ such that
  $\typetwo \in \possubtype\mtype$. There are two cases, either
  $\typetwo = \type = \neutral$ or $\typetwo = \ty{\M'}\type$. In both
  cases $\type$ is a positive occurrence of $\typetwo$. By
  transitivity of polarised occurrences (\reflemma{trans-pol}), $\type$ is a positive occurrence of $\mtype$, and thus
  of $\typctx$.
   \end{itemize}
\end{proof}

\gettoappendix {prop:shrinking-subject-reduction}
\begin{proof}
The first part (without the shrinking hypothesis) is an easy induction on $\tm\tolo\tmtwo$. The \emph{moreover} part is also by induction on $\tm\tolo\tmtwo$, but it requires a strengthened statement, along the same lines of the proof for the tight case:

  If
  $\tm\Rew{\lo}\tmtwo$,
  $\tderiv\exder[\lo] \Deri[(\steps, \result)]\typctx{\tm}\type$,
  $\typctx$ is shrinking,
  and either $\type$ is shrinking or $\sysnotabs \lo \tm$, then there exists a typing
  $\tderivtwo\exder[\lo] \Deri[(\stepstwo, \result)]{\typctx}{\tmtwo}\type$ with $\stepstwo \leq \steps-2$.

The cases of evaluation at top level, under abstraction, and in the left subterm of an application follows exactly the schema of the tight case: at top level the tight/shrinking hypothesis does not play any role, the abstraction case immediately follows from the \ih, and the left application case follows from the reinforced hypothesis that the left subterm is not an abstraction. We treat the case of evaluation in the right subterm of an application, that is the delicate one, where shrinkness plays a crucial role.

The rule is:
\[\begin{array}{ccccc}
\infer{\tm = \tmthree \tmfour \tolo \tmthree \tmfive = \tmtwo}{\loneutral\tmthree \quad \tmfour \tolo \tmfive} 
\end{array}\]
There are two cases for the last rule of the derivation $\tderiv$:
\begin{itemize}
\item \emph{$\appsteps$ rule}:
\[\begin{array}{c}
\infer[\appsteps]{\Deri
[(\steps_\tmthree +_{\iI} \steps_i + 1, \result_\tmthree +_{\iI} \result_i)]
{\typctx = \typctxthree \mplus_{\iI} \typctxtwo_i}
{\tmthree \tmfour} \type}
{
	
 	\tderiv_\tmthree \exder[\lo] { \Deri[(\steps_\tmthree, \result_\tmthree)] \typctxthree \tmthree {\MSigma {\typetwo_i} {\iI} \rightarrow \type}}

\quad
 \infer[\many]{
\Deri[(+_{\iI} \steps_i, +_{\iI} \result_i)] {\mplus_{\iI} \typctxtwo_i} \tmfour {\mult{\typetwo_i}_{\iI}}
}
 {
 (\tderiv_{\tmfour_i} \exder[\lo] \Deri[(\steps_i, \result_i)] {\typctxtwo_i} \tmfour {\typetwo_i})_{\iI}
 }
}\\\\
\end{array}
\]
The \ih applied to each $\tderiv_{\tmfour_i}$ and $\tmfour \tolo \tmfive$ gives $\tderiv_{\tmfive_i}$ such that $\tderiv_{\tmfive_i} \exder[\lo] \Deri[(\stepstwo_i, \result_i)] {\typctxtwo_i} \tmfive {\typetwo_i}$ with $\stepstwo_i \leq \steps_i$ and $\size{\tderiv_{\tmfive_i}} \leq \size{\tderiv_{\tmfour_i}}$. Then the derivation $\tderivtwo$ given by:

\[\begin{array}{c}
\infer[\appsteps]{\Deri
[(\steps_\tmthree +_{\iI} \stepstwo_i+1, \result_\tmthree +_{\iI} \result_i)]
{\typctx = \typctxthree +_{\iI} \typctxtwo_i}
{\tmthree \tmfive} \type}
{\tderiv_\tmthree \exder[\lo] {\Deri[(\steps_\tmthree, \result_\tmthree)] \typctxthree \tmthree {\ty {\MSigma {\typetwo_i} {\iI} } \type}}
\quad
\infer{\Deri[(+_{\iI} \steps'_i, +_{\iI} \result_i)] {\mplus_{\iI} \typctxtwo_i} \tmfive {\mult{\typetwo_i}_{\iI}}}
{
 (\tderiv_{\tmfive_i} \exder[\lo] \Deri[(\stepstwo_i, \result_i)] {\typctxtwo_i} \tmfive {\typetwo_i})_{\iI}
}

}\\\\
\end{array}
\]
verifies the statement. 

\emph{Shrinking}: if
$\tderiv$ is shrinking then $\typctx$ is shrinking, and so is
$\typctxthree$. Since $\loneutral\tmthree$ holds by hypothesis, by
\reflemma{shrinking-neutral-spreading} $\ty {\MSigma {\typetwo_i} {\iI} } \type$ is
a positive occurrence of $\typctxthree$. By shrinkness of
$\typctxthree$, the multiset $\MSigma {\typetwo_i} {\iI} $ is not
empty. Moreover, every $\typetwo_i$ is shrinking: if---by
contradiction---one of them is not shrinking, then $\emm$ occurs
positively in $\typetwo_i$ and so it occurs negatively in $\ty
{\MSigma {\typetwo_i} {\iI} } \type$, and then by transitivity of occurrences (\reflemma{trans-pol})
  $\emm$ occurs negatively in $\typctxthree$ ---absurd. Note that all
derivations $\tderiv_{\tmfive_i}$ are then shrinking: the contexts
$\typctxtwo_i$ are shrinking because they are sub-contexts of
$\typctx$ and we just showed that the types $\typetwo_i$ are
shrinking. Then by \ih $\stepstwo_i \leq \steps_i -2$ for every $\iI$,
and so $\stepstwo = \steps_\tmthree +_{\iI} \stepstwo_i+1 \leq
\steps_\tmthree +_{\iI} (\steps_i -2)+1 \leq 
\underbrace{\steps_\tmthree +_{\iI}
\steps_i + 1}_{=\steps}- |I| \cdot 2  = \steps -|I| \cdot 2 \leq_{I \neq \es} \steps -2$, as required.

\item \emph{$\appresult<\lo>$ rule}: 
$$
\infer[\appresult]{
\Deri[(\steps_\tmthree + \steps_\tmfour, \result_\tmthree + \result_\tmfour+1)] {\typctx = \typctx_\tmthree \mplus \typctx_\tmfour} {\tmthree \tmfour} \neutype
}
{
\tderiv_\tmthree \exder[\lo] \Deri[(\steps_\tmthree, \result_\tmthree)] {\typctx_\tmthree} \tmthree \neutype 
\quad
\tderiv_\tmfour \exder[\lo] \Deri[(\steps_\tmfour, \result_\tmfour)] {\typctx_\tmfour} \tmfour \nf
}
$$

with $\steps = \steps_\tmthree + \steps_\tmfour$  and $\result = \result_\tmthree + \result_\tmfour+1$. 
The \ih applied to $\tderiv_\tmfour$ and $\tmfour \tolo \tmfive$ gives $\tderiv_\tmfive$ such that $\tderiv_\tmfive \exder[\lo] \Deri[(\steps_\tmfive, \result_\tmfour)] {\typctx_\tmfour} \tmfive \nf$ with $\steps_\tmfive \leq \steps_\tmfour$ and so $\size{\tderiv_\tmfive} \leq \size{\tderiv_\tmfour}$. Then the derivation $\tderivtwo$ given by:
$$
\infer[\appresult]{
\Deri[(\steps_\tmthree + \steps_\tmfive, \result_\tmthree + \result_\tmfour+1)] {\typctx = \typctx_\tmthree \mplus \typctx_\tmfour} {\tmthree \tmfour} \neutype
}
{
\tderiv_\tmthree \exder[\lo] \Deri[(\steps_\tmthree, \result_\tmthree)] {\typctx_\tmthree} \tmthree \neutype 
\quad
\tderiv_\tmfive \exder[\lo] \Deri[(\steps_\tmfive, \result_\tmfour)] {\typctx_\tmfour} \tmfive \nf
}
$$
verifies the statement. 

\emph{Shrinking}: if $\tderiv$ is shrinking then $\typctx_\tmfour$ is shrinking, and so is $\tderivtwo_\tmfour$ (because tight types are shrinking). By \ih then $\steps_\tmfive \leq \steps_\tmfour - 2$, and so $\stepstwo = \steps_\tmthree + \steps_\tmfive \leq \steps_\tmthree + \steps_\tmfour - 2 = \steps -2$, as required.

\end{itemize}
\end{proof}

\gettoappendix {tm:shrinking-correctness}
\begin{proof} 
By induction on $\losize\tderiv$. If $\tm$ is a $\tolo$ normal form then
by taking $\tmtwo \defeq \tm$ and $k:=0$ the statement follows from the
typing of normal forms given by
\refprop{shrinking-normal-forms-forall}, including the \emph{moreover} part. Otherwise, $\tm \tolo
\tmthree$ and by shrinking subject reduction
(\refprop{shrinking-subject-reduction}) there is a shrinking derivation $\tderivtwo \exder[\lo] \Deri[(\steps',
  \result)] {\typctx}{\tmthree}{\type}$ such that $\steps' \leq \steps-2$ and $\size \tderivtwo \leq \size \tderiv -2$. By \ih, there exists a
$\tolo$ normal form $\tmtwo$ and a natural number $k' \leq
\steps'/2$ satisfying the statement with respect to
$\tmthree$. Let $k \defeq k'+1$. Then:
\begin{enumerate}
  \item \emph{Steps}: $\tm \tolo^k \tmtwo$ because $\tm \tolo \tmthree
    \tolo^{k'} \tmtwo$. Moreover, $k = k'+1 \leq_{\ih} \steps'/2 +1 \leq
     (\steps-2)/2 +1 = \steps/2$.

  \item \emph{Result size}: $\losize\tmtwo \leq_{\ih} \losize\tderivtwo - 2k' \leq \losize\tderiv -2 -2k' = \losize\tderiv -2(k' +1) =  \losize\tderiv -2k$.
\end{enumerate}
The \emph{moreover} part follows from the \ih

\end{proof}

\subsection{Shrinking Completeness}

\gettoappendix {prop:shrinking-normal-forms-exist}
\begin{proof}
By mutual induction on $\loneutral{\tm}$ and $\lonormal{\tm}$. 
\begin{enumerate}
  \item Cases of $\loneutral{\tm}$:
\begin{itemize}
	\item \emph{Variable}, \ie $\tm = \var$. Then 
	    $$\infer[\ax]{\Deri[(0, 0)] {\var \col \single \type} \var \type}{}$$
	    Moreover, $\tysize \type + \sksize\var =  0 = \tysize{\single \type} = \tysize{\var \col \single \type}$.
	    
	\item \emph{Application}, \ie $\tm = \tmtwo \tmthree$ with $\skneutral \tmtwo$ and $\sknormal \tmthree$. By \ih applied to $\tmthree$, there exists a type $\typetwo$ and a traditional shrinking typing $\tderiv_\tmthree \exder[\lo] \tyjn{(\sksize\tmthree, 0)} \tmthree {\typctx_\tmthree} \typetwo$ with 
	$\sksize\tmthree = \tysize {\typctx_\tmthree} + \tysize \typetwo$.

	Now, let $\type$ be a type, and consider the type $\tarrow {\mult\typetwo} \type$. By \ih applied to $\tmtwo$ and $\tarrow {\mult\typetwo} \type$ there exists $\tderiv_\tmtwo \exder[\lo] \tyjn{(\sksize\tmtwo, 0)} \tmtwo {\typctx_\tmtwo} {\tarrow {\mult\typetwo} \type}$ with $\sksize\tmtwo + \tysize {\tarrow {\mult\typetwo} \type} = \tysize {\typctx_\tmtwo}$.
	
	Then the derivation $\tderiv$ built as follows:
$$	\infer[\appsteps]{
\Deri[(\sksize\tmtwo + \sksize\tmthree + 1, 0)]
{\typctx_\tmtwo \mplus \typctx_\tmthree}
{\tmtwo \tmthree} \type
}
{	
\tderiv_\tmtwo \exder[\lo] \Deri[(\sksize\tmtwo, 0)] {\typctx_\tmtwo} \tmtwo {\tarrow{\mult\typetwo} \type} \quad
\infer[\many]{
	\Deri[(\sksize\tmthree, 0)] {\typctx_\tmthree} \tmthree {\mult\typetwo} 
     }
	{
    		\tderiv_\tmthree \exder[\lo] \Deri[(\sksize\tmthree, 0)] {\typctx_\tmthree} \tmthree \typetwo
	}
	}$$
It is traditional and shrinking because $\tderiv_\tmtwo$ and $\tderiv_\tmthree$ are. Moreover,
	$$\begin{array}{rcllllll}
	  \sksize\tm + \tysize {\type}  
	  & = &
	  \underbrace{\sksize\tmtwo + \sksize\tmthree + 1}_{=\sksize\tm} + \tysize {\type} 
	  \\
	  & =_{\tiny \mbox{\ih on $\tmthree$}} &
	  \sksize\tmtwo + \tysize {\typctx_\tmthree} + \underbrace{\tysize \typetwo + 1 + \tysize {\type}}_{=\tysize {\tarrow {\mult\typetwo} \type}}\\
	  & = &
	  \tysize {\typctx_\tmthree} + \underbrace{\sksize\tmtwo + \tysize {\tarrow {\mult\typetwo} \type}}_{\leq \tysize {\typctx_\tmtwo}}
	  \\
	  & =_{\tiny \mbox{\ih on $\tmtwo$}} &
	  \tysize {\typctx_\tmthree} + \tysize {\typctx_\tmtwo}
	  \\
	  & = & \tysize{\typctx_\tmtwo \mplus \typctx_\tmthree}
	\end{array}$$
\end{itemize}

\item Cases of $\lonormal{\tm}$:
\begin{enumerate}
\item \emph{$\skneutral{\tm}$}. By \ih, \emph{for every} type $\type$ there exists a traditional shrinking typing $\tderiv \exder[\ske] \tyjn{(\sksize\tm, 0)} \tm \typctx \type$ satisfying $\tysize\type + \sksize\tm = \tysize\typctx$. It is then enough to pick $\sigma \defeq \atomtype$, so that $\tysize\type = 0$ and the statement trivially holds, because then  $\sksize\tm = \tysize\type + \sksize\tm =_{\ih} \tysize\typctx = \tysize\typctx + \tysize\type = \tysize \tderiv$.

\item \emph{Abstraction}, \ie $\tm = \la\vartwo\tmtwo$ and $\sknormal{\tmtwo}$. By \ih, there exist a type $\typetwo$ and a shrinking \communicative typing $\tderiv_\tmtwo \exder[\lo] \tyjn{(\sksize\tmtwo, 0)} \tmtwo {\typctx_\tmtwo} \typetwo$ with
$\sksize\tmtwo = \tysize \tderiv = \tysize {\typctx_\tmtwo} + \tysize\typetwo $ . 

Two sub-cases: 
	\begin{enumerate}
		\item $\vartwo \in \dom{\typctx_\tmtwo}$. Then let $\vartwo \col \M$ the declaration of $\vartwo$ in $\typctx_\tmtwo$ and set $\typctx$ be $\typctx_\tmtwo$ without $\vartwo \col \M$. Then let $\tderiv$ be the derivation
		$$\infer[\funsteps]{
				\tyjn{(\sksize\tmtwo+1, 0)} {\la{\vartwo}\tmtwo} {\typctx} {\tarrow \M \typetwo}
			}
			{
				\tderiv_\tmtwo \exder[\lo] \tyjn{(\sksize\tmtwo, 0)} \tmtwo {\vartwo:\M; \typctx} \typetwo}$$
		which is traditional and shrinking because $\tderiv_\tmtwo$ is.
		Setting $\type \defeq \tarrow \M \typetwo$, we have
		  $$\begin{array}{rcllllll}
		    \sksize{\la\vartwo\tmtwo} & = & \sksize\tmtwo +1 
		    \\
		    & =_{\ih} &
		    \tysize {\vartwo:\M; \typctx} + \tysize\typetwo+1
		    \\
		    & = &
		    
		    \tysize\typctx + \underbrace{\tysize \M+ \tysize\typetwo+1}_{= \tysize {\tarrow \M \typetwo} = \tysize \type }
		    \\
		    & = &
		    \tysize\typctx + \tysize \type 
		    \\
		    & = &
		\tysize \tderiv
		   \end{array}$$
	

		\item $\vartwo \notin \dom{\typctx_\tmtwo}$. Then let $\tderiv$ be the derivation
		$$\infer[\funsteps]{
				\tyjn{(\sksize\tmtwo+1, 0)} {\la{\vartwo}\tmtwo} {\typctx_\tmtwo} {\tarrow {\emptymset} \typetwo}
			}
			{
				\tderiv_\tmtwo \exder[\lo] \tyjn{(\sksize\tmtwo, 0)} \tmtwo {\typctx_\tmtwo} \typetwo}$$
		which is traditional and shrinking (because $\tarrow \emptymset \typetwo$ is shrinking).
		Now, if we set $\type \defeq \tarrow \emptymset \typetwo$ and $\typctx \defeq \typctx_\tmtwo$ we obtain
		  $$\begin{array}{rcllllll}
		    \sksize{\la\vartwo\tmtwo}  & = & \sksize\tmtwo +1 
		    \\
		    & =_{\ih} &
		    
		    \underbrace{\tysize {\typctx_\tmtwo}}_{\tysize {\typctx}}+ \underbrace{\tysize\typetwo+1}_{= \tysize {\tarrow \emptymset \typetwo} = \tysize \type }
		    \\
		    & = &
		    \tysize\typctx + \tysize \type 
		   \\
		    & = &
		    \tysize \tderiv
		    \end{array}$$
	\end{enumerate}
\end{enumerate}
\end{enumerate}

\end{proof}

\gettoappendix {prop:shrinking-subject-expansion}
\begin{proof}
The first part (without the shrinking hypothesis) is an easy induction on $\tm\tolo\tmtwo$. The part about shrinking minimal typings is also by induction on $\tm\tolo\tmtwo$, but it requires a strengthened statement, along the same lines of the proof for the tight case and of subject reduction:

  if
  $\tm\Rew{\lo}\tmtwo$,
  $\tderiv\exder[\lo] \Deri[(\steps, \result)]\typctx{\tmtwo}\type$,
  $\typctx$ is shrinking and of minimal size,
  and either $\type$ is shrinking or $\sysnotabs \lo \tm$, then there exists a typing
  $\tderivtwo\exder[\lo] \Deri[(\steps +2, \result)]{\typctx}{\tm}\type$.

The cases of evaluation at top level, under abstraction, and in the left subterm of an application follows exactly the schema of the tight case: at top level the tight / shrinking hypothesis does not play any role, the abstraction case immediately follows from the \ih, and the left application case follows from the reinforced hypothesis that the left subterm is not an abstraction. We treat the case of evaluation in the right subterm of an application, that is the delicate one, where shrinkness plays a crucial role.

The rule is:
\[\begin{array}{ccccc}
\infer{\tm = \tmthree \tmfour \tolo \tmthree \tmfive = \tmtwo}{\loneutral\tmthree \quad \tmfour \tolo \tmfive} 
\end{array}\]
There are two cases for the last rule of the derivation $\tderiv$:
\begin{itemize}
\item \emph{$\appsteps$ rule}:
\[\begin{array}{c}
\infer[\appsteps]{\Deri
[(\steps_\tmthree +_{\iI} \steps_i + 1, \result_\tmthree +_{\iI} \result_i)]
{\typctx = \typctxthree \mplus_{\iI} \typctxtwo_i}
{\tmthree \tmfive} \type}
{
	
 	\tderiv_\tmthree \exder[\lo] { \Deri[(\steps_\tmthree, \result_\tmthree)] \typctxthree \tmthree {\MSigma {\typetwo_i} {\iI} \rightarrow \type}}

\quad
 \infer[\many]{
\Deri[(+_{\iI} \steps_i, +_{\iI} \result_i)] {\mplus_{\iI} \typctxtwo_i} \tmfive {\mult{\typetwo_i}_{\iI}}
}
 {
 (\tderiv_{\tmfive_i} \exder[\lo] \Deri[(\steps_i, \result_i)] {\typctxtwo_i} \tmfive {\typetwo_i})_{\iI}
 }
}\\\\
\end{array}
\]
The \ih applied to each $\tderiv_{\tmfive_i}$ and $\tmfour \tolo \tmfive$ gives $\tderiv_{\tmfour_i}$ such that $\tderiv_{\tmfour_i} \exder[\lo] \Deri[(\stepstwo_i, \result_i)] {\typctxtwo_i} \tmfour {\typetwo_i}$ with $\stepstwo_i \geq \steps_i$ and $\size{\tderiv_{\tmfour_i}} \geq \size{\tderiv_{\tmfive_i}}$. Then the derivation $\tderivtwo$ given by:
\[\begin{array}{c}
\infer[\appsteps]{\Deri
[(\steps_\tmthree +_{\iI} \stepstwo_i+1, \result_\tmthree +_{\iI} \result_i)]
{\typctx = \typctxthree +_{\iI} \typctxtwo_i}
{\tmthree \tmfour} \type}
{\tderiv_\tmthree \exder[\lo] {\Deri[(\steps_\tmthree, \result_\tmthree)] \typctxthree \tmthree {\ty {\MSigma {\typetwo_i} {\iI} } \type}}
\quad
\infer{\Deri[(+_{\iI} \steps'_i, +_{\iI} \result_i)] {\mplus_{\iI} \typctxtwo_i} \tmfour {\mult{\typetwo_i}_{\iI}}}
{
 (\tderivtwo_{\tmfour_i} \exder[\lo] \Deri[(\stepstwo_i, \result_i)] {\typctxtwo_i} \tmfour {\typetwo_i})_{\iI}
}

}\\\\
\end{array}
\]
verifies the statement. 

\emph{Shrinking}: exactly the same reasoning used for shrinking subject reduction proves that $\tderiv_\tmthree$ is shrinking, $I$ is non-empty, and the $\typetwo_i$ are all shrinking. The \ih then provides $\stepstwo_i \geq \steps_i +2$ for every $\iI$, from which the property follows.

\item \emph{$\appresult<\lo>$ rule}: 
$$
\infer[\appresult]{
\Deri[(\steps_\tmthree + \steps_\tmfive, \result_\tmthree + \result_\tmfive+1)] {\typctx = \typctx_\tmthree \mplus \typctx_\tmfive} {\tmthree \tmfive} \neutype
}
{
\tderiv_\tmthree \exder[\lo] \Deri[(\steps_\tmthree, \result_\tmthree)] {\typctx_\tmthree} \tmthree \neutype 
\quad
\tderiv_\tmfive \exder[\lo] \Deri[(\steps_\tmfive, \result_\tmfive)] {\typctx_\tmfive} \tmfive \nf
}
$$

with $\steps = \steps_\tmthree + \steps_\tmfive$ and $\result = \result_\tmthree + \result_\tmfive+1$. 
The \ih applied to $\tderiv_\tmfive$ and $\tmfour \tolo \tmfive$ gives $\tderivtwo_\tmfour$ such that $\tderivtwo_\tmfour \exder[\lo] \Deri[(\steps_\tmfive, \result_\tmfour)] {\typctx_\tmfour} \tmfive \nf$ with $\steps_\tmfour \geq \steps_\tmfive$ and so $\size{\tderivtwo_\tmfour} \geq \size{\tderiv_\tmfive}$. Then the derivation $\tderivtwo$ given by:
$$
\infer[\appresult]{
\Deri[(\steps_\tmthree + \steps_\tmfour, \result_\tmthree + \result_\tmfive+1)] {\typctx = \typctx_\tmthree \mplus \typctx_\tmfour} {\tmthree \tmfour} \neutype
}
{
\tderiv_\tmthree \exder[\lo] \Deri[(\steps_\tmthree, \result_\tmthree)] {\typctx_\tmthree} \tmthree \neutype 
\quad
\tderivtwo_\tmfour \exder[\lo] \Deri[(\steps_\tmfour, \result_\tmfive)] {\typctx_\tmfour} \tmfive \nf
}
$$
verifies the statement. 

\emph{Shrinking}: if $\tderiv$ is shrinking then $\typctx_\tmfour$ is shrinking, and so is $\tderivtwo_\tmfour$ (because tight types are shrinking). By \ih then $\steps_\tmfour \geq \steps_\tmfive + 2$, and so $\stepstwo = \steps_\tmthree + \steps_\tmfour \leq \steps_\tmthree + \steps_\tmfive + 2 = \steps +2$, as required.

\end{itemize}
\end{proof}

\gettoappendix {thm:shrinking-completeness}
\begin{proof}
By induction on $\tm \tolo^k \tmtwo$. If $k = 0$ the statement is
given by the existence of tight typings for $\tolo$-normal terms
(\refprop{shrinking-normal-forms-exist}). Let $k > 0$ and
$\tm \tolo \tmthree \tolo^{k-1}\tmtwo$. By \ih, there exists a tight
typing derivation $\tderivtwo \exder[\lo] \tyjn{(\stepstwo, 0)}  
\tmthree \typctx \type$ with $(k-1) \leq \stepstwo /2$, and $\sksize\tmtwo = \tysize\tderivtwo$. By shrinking subject expansion (\refprop{shrinking-subject-expansion})
there exists a typing derivation $\tderiv$ of $\tm$ with the same
types in the ending judgement of $\tderivtwo$---then $\tderiv$ is
shrinking and $\tysize\tderiv = \tysize\tderivtwo$---and with indices $(\steps, 0)$ with $\steps \geq \stepstwo+2$. Then $k = k-1+1 \leq_{\ih} \stepstwo/ 2 +2 = (\stepstwo +2)/2 \leq \steps/ 2$.
\end{proof}

}
          {}

\end{document}